\documentclass[lettersize,journal,10pt]{IEEEtran}
\usepackage{amsmath,amsfonts}
\usepackage{algorithmic}
\usepackage{algorithm}
\usepackage{array}
\usepackage{textcomp}
\usepackage{stfloats}
\usepackage{url}
\usepackage{verbatim}
\usepackage{graphicx}
\usepackage{cite}
\usepackage[nolist,nohyperlinks]{acronym}
\begin{acronym}
    \acro{HAP}{High-altitude platform station}
    \acro{NTN}{non-terrestrial network}
    \acro{LEO}{low Earth orbit}
    \acro{SCR-CP}{Small-Circle Ring Cox Process}
    \acro{TBS}{Terrestrial Base Stations}
    \acro{CCDF}{complementary cumulative density function}
    \acro{CDF}{cumulative density function}
    \acro{PDF}{probability density function}
    \acro{PPP}{Poisson point process}
    \acro{BPP}{binomial point process}
    \acro{SINR}{signal-to-interference-plus-noise ratio}
    \acro{SIR}{signal-to-interference ratio}
    \acro{LOS}{line-of-sight}
    \acro{NLOS}{non line-of-sight}
    \acro{PGFL}{probability generating functional}
    \acro{5G}{fifth generation}
    \acro{PLP}{Poisson line process}
    \acro{BLP}{binomial line process}
    \acro{SG}{Stochastic geometry}
    \acro{PLE}{path-loss exponent}
    \acro{NOMA}{non-orthogonal multiple access}
    \acro{UE}{user equipment}
    \acro{CEE}{coverage energy efficiency}
    \acro{SCR-PCP}{small-circle ring Poisson Cox process}
    \acro{SCR-BCP}{small-circle ring binomial cluster process}
\end{acronym}

\usepackage{subfig}

\usepackage{xcolor}
\usepackage{float}
\usepackage{svg}
\usepackage{pdfpages}
\usepackage{comment}

\usepackage{graphicx,wrapfig,lipsum}
\usepackage{rotating}
\usepackage{amsmath, amssymb, amsthm}
\usepackage{stmaryrd}
\usepackage{tikz}
\usepackage{verbatim}
\usepackage{url}
\usepackage[utf8]{inputenc}
\usepackage[english]{babel}

\newtheorem{theorem}{Theorem}[]
\newtheorem{corollary}{Corollary}[]

\newtheorem{remark}{Remark}[]

\newtheorem{lemma}[]{Lemma}

\usepackage{multicol}

\usepackage{scalerel}
\usepackage{multicol}
\usepackage{setspace}
\newtheorem{definition}{Definition}
\usepackage{booktabs}
\usepackage{bbm}
\usepackage[normalem]{ulem}
\usepackage{color}
\usepackage{dsfont}
\usepackage{bm}
\usepackage{setspace}

\usetikzlibrary{automata}
\usetikzlibrary{arrows}
\usetikzlibrary{shapes.geometric, arrows}
\usepackage[]{footmisc}
\usetikzlibrary{positioning}
\usetikzlibrary{calc}
\usetikzlibrary{arrows}
\usepackage{dirtytalk}
\allowdisplaybreaks

\usepackage{lipsum}

\usepackage[margin = 0.3cm]{caption}
\usepackage{mathtools}
\usepackage{csquotes}

\makeatletter
\newcommand{\vast}{\bBigg@{3}}
\newcommand{\Vast}{\bBigg@{4}}
\makeatother

\usepackage{multirow}
\usepackage{tabularx}

\begin{document}

\title{A Spherical Stochastic Geometry Framework for Patrol-Based HAPs Network: Coverage and Energy Efficiency Analysis}

\author{
\IEEEauthorblockN{Mohammad~Taha~Shah}, {\it Graduate Student Member, IEEE}, \IEEEauthorblockN{ Mohamed-Slim~Alouini}, {\it Fellow, IEEE}
\thanks{Mohammad Taha Shah and Mohamed-Slim Alouini are with the CEMSE division, KAUST, Thuwal, Saudi Arabia, 23955-6900; Email: \{mohammadtaha.shah, slim.alouini\}@kaust.edu.sa}
\vspace{-0.5cm}}

\maketitle

\begin{abstract}
This paper develops a stochastic geometry framework for high-altitude platform station (HAP) networks in which platforms execute cyclic patrol trajectories anchored to designated service regions. We introduce two small-circle ring point-process models on the spherical Earth. In the small-circle ring Poisson Cox process (SCR-PCP), platforms form one-dimensional Poisson point processes on localized patrol rings, whereas in the small-circle ring binomial cluster process (SCR-BCP), each ring contains a fixed number of uniformly distributed platforms. We establish the isotropy of both models and derive spatial statistics, including the distributions of the nearest-anchor, nearest-ring, and nearest-HAP distances, together with the serving-ring angle distribution. Building on these results, we derive coverage probability expressions under nearest-visible-HAP association by decomposing aggregate interference into same-ring and other-ring components and characterizing their conditional Laplace transforms. To account for the flight dynamics of patrol-based HAPs, we integrate a steady circular flight propulsion model with the communication analysis and introduce a coverage energy efficiency (CEE) metric. This yields a necessary optimality condition for an interior energy-optimal patrol radius that balances coverage performance against the propulsion cost of circular flight. Numerical results reveal fundamental differences between intensity-driven (SCR-PCP) and finite-fleet (SCR-BCP) deployments and demonstrate that patrol geometry, platform density, and cruising velocity should be jointly considered to achieve energy-efficient HAP operation.
\end{abstract}

\section{Introduction}
\subsection{Background}
\acp{HAP} are emerging as a pivotal component of future \acp{NTN}, designed to complement terrestrial infrastructure and \ac{LEO} satellite constellations in beyond-5G and 6G architectures \cite{kurt2021vision,abbasi2024haps}. Operating in the stratosphere at altitudes between 20 and 50 km, \acp{HAP} combine key advantages of both satellite and terrestrial systems: they provide the wide coverage footprint and dominant \ac{LOS} conditions of spacecraft while retaining the lower latency, agility, and upgradeability of \ac{TBS} \cite{kurt2021vision,lou2025coverage}. As a result, \acp{HAP} are increasingly viewed as an attractive solution for extending broadband connectivity to rural areas, maritime regions, island communities, and supporting integrated space-air-ground-sea networks \cite{lin2025connectivity,lin2026haps,xu2023space,yuan2023joint}.

As \ac{HAP} deployments evolve toward large-scale networks, quantifying aggregate interference, spatial reuse, and macro-diversity requires network-level analysis \cite{huang2024system,bliss2024orchestrating}. \ac{SG} provides a tractable framework for characterizing coverage, rate, and connectivity \cite{haenggi2013stochastic,andrews2011tractable}. However, modeling \acp{HAP} requires accounting for their distinctive operational geometry: geographically anchored service regions, bounded patrol trajectories, Earth's curvature and horizon constraints. The relevance of bounded station keeping is supported by ITU-R technical reports and flight demonstrations. ITU-R report F.2439-0~\cite{itu2018haps} models a HAP operating within a bounded cylindrical region around a nominal service location, with horizontal flight radii of $1$--$5$~km. Circular loitering was demonstrated during HAPSMobile's Sunglider stratospheric broadband flight test \cite{hapsalliance2021flight}. Similarly, \cite{belmekki2024cellular} reports a repeatable circular trajectory of approximately $12$~km radius at an altitude of $14$~km during a 5G demonstration. More recently, SoftBank demonstrated HAPS-oriented 5G connectivity from an aircraft circling over Hachij\={o} Island while maintaining fixed ground-cell footprints \cite{softbank2025haps}. These examples establish localized cyclic patrol trajectories as a practically motivated representation of fixed-wing HAP station keeping.

\subsection{Related Work and Research Gap}
Prior research on HAP and aerial communication networks has examined flight trajectories, large-scale network performance, and spatial models that account for the underlying movement of aerial platforms and satellites.

\subsubsection{Trajectory-Aware HAPS and Constrained Aerial Networks} Several studies incorporate prescribed or optimized flight trajectories into aerial communication systems. Javed \textit{et al.} \cite{javed2023interdisciplinary,javed2024system} jointly consider communication performance, aerodynamic constraints, and non-orthogonal multiple access (NOMA) for solar-powered HAPs. Repetitive and circular HAPS trajectories are studied in \cite{hsieh2020uav}, while \cite{mori2023haps} optimizes HAPS trajectories according to spatial traffic distributions. More recently, \cite{jang2026analytical} develops a stochastic geometry-based downlink analysis for a HAP following a circular trajectory over an urban service region, whereas \cite{elamassie2024channel} investigates mobility-induced channel variations for a fleet of HAPs following a predefined circular trajectory in airborne FSO backhauling. Although these studies establish the relevance of trajectory-aware HAPS operation, they primarily consider individual platforms, prescribed fleets, or deterministic trajectories rather than the network-scale spatial dependence induced by geographically distributed patrol regions.

Related stochastic models have also been developed for lower-altitude UAV networks. Qin \textit{et al.} \cite{qin2023stochastic} combine stochastic geometry with UAV trajectory design, while Armeniakos \textit{et al.} \cite{armeniakos2024stochastic,armeniakos2026performance} analyze UAV base stations deployed along finite linear corridors using one-dimensional binomial and Poisson point processes. The latter study additionally incorporates shadowing-aware maximum power association, aggregate interference, and experimental air-to-ground validation. Furthermore, \cite{zaid2025aerial} models random urban air-mobility corridors using a \ac{PLP}, with eVTOL locations represented by a Mat\'ern hard-core process. These works demonstrate the value of trajectory-constrained spatial models, but focus on finite or planar open corridors rather than geographically anchored closed patrol trajectories over a spherical Earth.

\subsubsection{Network-Scale HAPS Models} Another research direction investigates HAPS performance at the network scale. For instance, \cite{wei2023spectrum} develops a stochastic and differential geometry framework for spectrum sharing between HAPS and terrestrial networks, whereas \cite{lou2025coverage} considers large-scale spherical HAPS deployments with directional antenna beams. Connectivity through percolation theory is investigated in \cite{lin2025connectivity}, while integrated terrestrial-aerial architectures, relaying, spectrum coexistence, and random aerial deployments are studied in \cite{shamsabadi2022handling,matracia2023uav, tarhouni2025performance,song2022cooperative}. More recently, \cite{tanash2025urban} investigates RIS-assisted HAP networks using stochastic geometry.

These models characterize spatial randomness, aggregate interference, and relevant physical-layer effects. However, they generally represent HAPs through their instantaneous spatial locations without explicitly coupling each platform to a geographical service-region center and its associated bounded patrol trajectory. Hence, the spatial dependence induced by the anchor-loiter relationship is not captured.

\subsubsection{Structured Trajectory Based Stochastic Models} Earlier studies \cite{okati2020downlink,talgat2020nearest, talgat2020stochastic} characterize \ac{LEO} coverage and rate using spherical binomial or Poisson processes. More recent orbit-aware Cox models \cite{choi2024novel,choi2025stochastic} represent random orbital planes populated by one-dimensional satellite point processes, thereby preserving trajectory-induced spatial dependence. Related studies \cite{wang2025non,wang2025modeling} highlight the importance of spherical geometry for high-altitude networks, particularly the errors introduced by planar approximations. Although these formulations capture important geometric and spatial dependencies, satellite Cox models describe global orbital trajectories, whereas aerodynamic HAPs perform localized station keeping around designated service regions.

Existing studies address aerial trajectories and spatial randomness but lack a unified framework linking geographically distributed service regions to closed patrol trajectories. We propose a framework that captures the spatial dependence between service-region centers and platform locations while enabling analysis of coverage probability and energy efficiency.
\begin{table*}[t]
\centering
\small
\begin{tabular}{|p{2.2cm}|p{4cm}|p{8.5cm}|}
\hline
\textbf{Model} & \textbf{Intra-ring structure} & \textbf{Representative scenarios} \\ \hline
\textbf{SCR-PCP} & 1-D \ac{PPP} with linear intensity & Dynamic or partially coordinated fleets; random ON/OFF activity; maintenance overlap; multi-operator sharing; demand-driven densification. \\ \hline
\textbf{SCR-BCP} & Binomial with fixed $n_{\rm H}$ points per ring & Dimensioned fleets with a prescribed number of platforms per region; primary and backup platforms loitering on essentially the same orbit; companion HAPs providing relaying, sensing, or RIS functions on the same trajectory; strictly CapEx-limited designs where the per-ring population is an explicit planning variable. \\ \hline
\end{tabular}
\caption{Comparison of SCR-PCP and SCR-BCP models.}
\label{tab:scr_models}
\end{table*}

\subsection{Motivation for SCR-PCP and SCR-BCP}
The identified research gap motivates a spatial framework that explicitly couples geographically distributed service-region centers with the instantaneous locations of patrolling \acp{HAP}. As discussed in the Background section, bounded station keeping is supported by HAPS engineering models and experimental deployments. We therefore introduce two \emph{small-circle ring} spatial models with a common hierarchical construction that adopts circular loitering as a canonical and analytically tractable representation of this behavior. The proposed construction places geographical anchors\footnote{Nominal loitering centers associated with fixed geographical service regions.} according to a homogeneous \ac{PPP} on a spherical shell, with each anchor defining a localized circular patrol ring in its tangent plane. The resulting hierarchical structure preserves the dependence between service-region centers and platform locations. The proposed geometry differs from finite linear and planar air-corridor models \cite{armeniakos2024stochastic,armeniakos2026performance, zaid2025aerial} and from Earth-centered orbital constructions \cite{choi2024novel,choi2025stochastic} by associating localized closed trajectories with geographically designated service regions.

We consider two intra-ring population models, summarized in Table~\ref{tab:scr_models}. In the \ac{SCR-PCP}, platforms form a one-dimensional \ac{PPP} of linear intensity $\nu$ on each ring, allowing the number of platforms per service region to vary randomly. Since the patrol rings are generated by a random process of geographical anchors, the conditional Poisson placement of platforms makes the \ac{SCR-PCP} a Cox process. In contrast, the \ac{SCR-BCP} assigns a fixed population $n_{\rm H}$ of independently and uniformly positioned platforms to each ring, resulting in a Poisson-parent binomial cluster process rather than a Cox process. The case $n_{\rm H}=1$ corresponds to a single-platform region, whereas $n_{\rm H}>1$ captures multi-platform operation and same-ring interference. Both models retain the same spherical anchor-patrol geometry, enabling a unified analysis of serving distances, interference, coverage, and energy-aware patrol optimization.

\subsection{Contributions}
The main contributions are summarized as follows.
\begin{itemize}
    \item \textbf{Small-circle ring models for \ac{HAP} networks:} We introduce two spherical point-process models for patrol-based \ac{HAP} deployments: \textcolor{blue}{the \ac{SCR-PCP}, a Poisson Cox process with one-dimensional \acp{PPP} on patrol rings, and the \ac{SCR-BCP}, a Poisson-parent binomial cluster process with a fixed number $n_{\rm H}$ of uniformly distributed platforms per ring.} These models capture geographically anchored patrol trajectories and spherical geometry under intensity-driven and fixed-population deployment regimes.
    \item \textbf{Isotropy and geometry-exact spatial statistics:} We establish the isotropy of the proposed models and derive spatial statistics, including the distributions of the nearest-anchor, nearest-ring, and nearest-\ac{HAP} distances. The analysis uses 3D patrol geometry and imposes Earth blockage by applying a minimum-elevation constraint to actual platform locations. We characterize the serving-ring angle conditioned on the serving distance and account for the probability that no visible \ac{HAP} is available.
    \item \textbf{Coverage probability analysis:} We derive coverage probability expressions for both \ac{SCR-PCP} and \ac{SCR-BCP}. The analysis decomposes aggregate interference into same-ring and other-ring components and yields conditional Laplace transforms that capture the distinct structures of the two models.
    \item \textbf{Energy-aware patrol optimization:} We integrate fixed-wing flight mechanics with \ac{SG} performance analysis by introducing a \ac{CEE} framework. Using a steady circular-flight propulsion model, we derive the stationarity condition for the energy-optimal patrol radius and quantify the interplay among patrol geometry, cruising velocity, and network topology in energy-efficient \ac{HAP} operation.
\end{itemize}

\subsection{Notation and Organization}
Throughout the paper, the subscript $k \in \{\mathrm{P},\mathrm{B}\}$ is used to denote quantities associated with the \ac{SCR-PCP} and \ac{SCR-BCP} models, respectively. The remainder of the paper is organized as follows. Section~\ref{sec:sys_model} introduces the spherical patrol-ring system model and the proposed \ac{SCR-PCP} and \ac{SCR-BCP} constructions. Section~\ref{sec:stats} establishes the fundamental spatial statistics of the two models. Section~\ref{sec:coverage} derives coverage probability expressions under nearest-\ac{HAP} association. Section~\ref{sec:energy} incorporates aerodynamic flight mechanics and develops the \ac{CEE} framework. Section~\ref{sec:numerical_results} presents numerical results, and Section~\ref{sec:conclusion} concludes the paper.

\section{System Model}
\label{sec:sys_model}
\subsection{Global Geometry and the Patrol-Ring Process}
We consider a spherical Earth of radius $R_\oplus$ centered at the origin in $\mathbb{R}^3$. The \ac{HAP} loitering centers, which we refer to as \emph{anchors}, are located at a nominal altitude $H$ and lie on the concentric spherical shell $\mathcal{S}_{R_{\rm H}} = \{ x \in \mathbb{R}^3 : \|x\| = R_{\rm H} \}$, where $R_{\rm H} = R_\oplus + H$. To describe a large-scale, globally distributed \acp{HAP} architecture, we model the anchor locations as a homogeneous \ac{PPP} $\Phi_u \subset \mathcal{S}_{R_{\rm H}}$ with surface intensity $\lambda_{\rm u}$ (anchors per km$^2$). By the isotropy of the proposed spatial models (Theorem~\ref{thm:isotropy}), we evaluate network performance at a representative ground user fixed at the North Pole, $\mathbf{o}=(0,0,R_\oplus)$, without loss of generality. The global geometry between the typical user at $\mathbf{o}$ and any anchor $u \in \Phi_u$ is fully characterized by the central angle $\phi = \angle(\mathbf{o}, u) = \arccos\left( \frac{\mathbf{o} \cdot u}{R_\oplus R_{\rm H}} \right)$, which corresponds to the polar angle in spherical coordinates with respect to the user’s zenith.

Unlike static \acp{TBS}, the fixed-wing \acp{HAP} considered here maintain continuous aerodynamic flight. To capture this local patrol behavior, we associate with each anchor $u\in\Phi_u$ a deterministic patrol ring embedded in the local tangent plane of the anchor shell. For the geometric analysis of an individual ring, rotational symmetry about the user's zenith allows its anchor azimuth $\theta$ to be set to zero. The anchor coordinates are then $u=(R_{\rm H}\sin\phi,\,0,\,R_{\rm H}\cos\phi)$, with outward unit normal vector $n_u=u/\|u\|$. The local tangent plane at the anchor is $\Pi(u) = \{ x \in \mathbb{R}^3 : (x - u) \cdot n_u = 0 \}$. To parameterize positions within $\Pi(u)$, we introduce an orthonormal basis $(e_1,e_2)$ satisfying $e_1 = (0,1,0)$ and $e_2 = (-\cos\phi, 0, \sin\phi)$. In this coordinate system, $e_2$ is tangent to the great circle passing through the user and the anchor, while $e_1$ is orthogonal to this plane. We model the patrol trajectory associated with anchor $u$ as a Euclidean circle of radius $a = R_{\rm H} \tan\gamma$, where $\gamma \in (0,\pi/2)$ controls the angular footprint of the patrol orbit. The corresponding patrol ring is
\begin{align*}
    \mathcal{C}(u,\gamma) = \left\{ u + a(\cos t \, e_1 + \sin t \, e_2) : t \in [0,2\pi) \right\}.
\end{align*}
Since $e_1$ and $e_2$ are orthogonal to $u$, every point $x\in\mathcal{C}(u,\gamma)$ satisfies $\|x\| = R_{\rm F} \triangleq \sqrt{R_{\rm H}^{2}+a^{2}}.$ Hence, all \acp{HAP} lie on the concentric platform shell $\mathcal{S}_{R_{\rm F}}$, and their geometric altitude is
\begin{align}
    H_{\rm F} = R_{\rm F}-R_\oplus = \sqrt{R_{\rm H}^{2}+a^{2}}-R_\oplus.
    \label{eq:platform_altitude}
\end{align}
Accordingly, $H$ denotes the patrol-anchor altitude, whereas $H_{\rm F}$ denotes the actual platform altitude. Moreover, since $\mathcal{C}(u,\gamma)$ lies simultaneously in the plane $x\cdot n_u=R_{\rm H}$ and on the sphere $\mathcal{S}_{R_{\rm F}}$, it forms an exact small circle.

The collection of all such rings forms a \textit{patrol-ring process} on the platform shell $\mathcal{S}_{R_{\rm F}}$: $\mathcal{L}=\big\{ \mathcal{C}(u,\gamma)\colon u\in\Phi_u \big\}$. This anchor-patrol construction plays the same structural role for \ac{HAP} as a \ac{PLP}~\cite{chetlur2018coverage} or \ac{BLP}~\cite{shah2024binomial} does for vehicular networks on random streets, and it serves as the geometric backbone of the \ac{SG} models developed in this paper.
\begin{figure}[t]
    \centering
    \includegraphics[trim={15.1cm 4.6cm 14cm 3.7cm},clip,width = 0.35\textwidth]{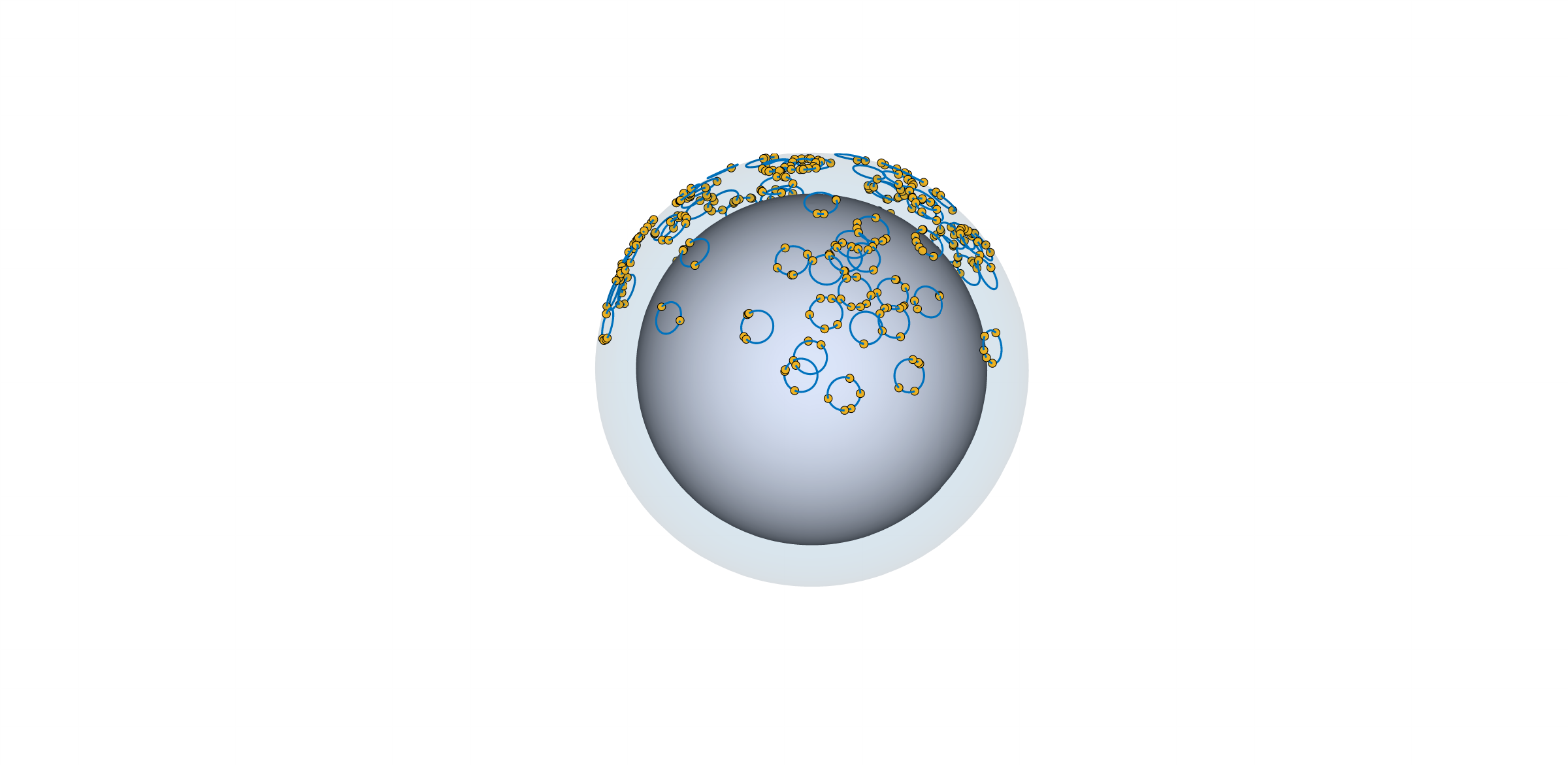}
    \caption{3D realization of the \ac{SCR-PCP} model. Anchors form a \ac{PPP} on the spherical shell of radius $R_{\rm H}$, while blue circles represent patrol rings of radius $a$. The patrol rings, and hence the \ac{HAP} locations shown by yellow points, lie on the concentric platform shell of radius $R_{\rm F}$.}
    \label{fig:fig_1}
\end{figure}

\subsection{Small-Circle Ring Poisson and Binomial Processes}
Conditioned on the patrol-ring process $\mathcal{L}$, the \acp{HAP} on each ring $\mathcal{C}(u,\gamma)$ are modeled using one of two intra-ring point processes.

\subsubsection{Small-Circle Ring Poisson Cox Process (SCR-PCP)}
In the \ac{SCR-PCP}, the \acp{HAP} associated with anchor $u$ form an independent one-dimensional homogeneous \ac{PPP} $\mathcal{P}_{u}$ on $\mathcal{C}(u,\gamma)$ with linear intensity $\nu$ (\acp{HAP} per km) with respect to arc length. The aggregate point process on $\mathbb{R}^3$ is therefore a Cox point process driven by $\mathcal{L}$:
\begin{align*}
    \Phi_{\rm H}^{\rm P} = \bigcup_{u \in \Phi_u} \mathcal{P}_{u}.
\end{align*}
We refer to this two-level construction as \ac{SCR-PCP}, and it is illustrated in Fig.~\ref{fig:fig_1}.

\subsubsection{Small-Circle Ring Binomial Cluster Process (SCR-BCP)}
In practical deployments, the number of platforms assigned to each patrol region is explicitly dimensioned rather than derived from a linear density. To capture such finite-population configurations, we introduce the \ac{SCR-BCP}.\footnote{Conditioned on the random patrol-ring process, the \ac{SCR-PCP} places an independent one-dimensional \ac{PPP} on each ring and is therefore a \textbf{Cox process}. In contrast, the \ac{SCR-BCP} places exactly $n_{\rm H}$ independently and uniformly distributed platforms on each ring, yielding binomial counts on ring subarcs. Thus, the \ac{SCR-BCP} is a Poisson-parent binomial \textbf{cluster process} rather than a Cox process.} In the \ac{SCR-BCP}, the anchor \ac{PPP} $\Phi_u$ and patrol-ring geometry $\mathcal{L}$ are unchanged, but the intra-ring point process is modified as follows.

Conditioned on an anchor $u$, the associated \acp{HAP} form a \ac{BPP} $\mathcal{B}_{u}$ on $\mathcal{C}(u,\gamma)$ consisting of $n_{\rm H}$ platforms. Specifically, $\mathcal{B}_{u}=\{x_{u,1},\ldots,x_{u,n_{\rm H}}\}$, where the platform locations are given by $x_{u,k}=u+a(\cos\Gamma_{u,k} e_1+\sin\Gamma_{u,k} e_2)$ and $\Gamma_{u,1},\ldots,\Gamma_{u,n_{\rm H}}$ are i.i.d.\ $\mathrm{Unif}[0,2\pi)$ random variables. The resulting process is
\begin{align*}
    \Phi_{\rm H}^{\rm B} = \bigcup_{u \in \Phi_u} \mathcal{B}_{u}.
\end{align*}
The \ac{SCR-BCP} is therefore a Poisson-parent binomial cluster process on $\mathbb{R}^3$, with the same anchor \ac{PPP} as \ac{SCR-PCP} but a fixed-size binomial intra-ring kernel. As we show later, this replacement changes the ring-level void probabilities and the interference statistics while preserving all geometric quantities.

Both models combine spherical geometry with localized patrol trajectories, preserving Earth's curvature and the altitude-dependent relationship between angles and distances. Hereafter, we use \ac{SCR-PCP} and \ac{SCR-BCP} to denote the Poisson and binomial intra-ring models, respectively, occasionally omitting the superscript when referring to a single model.

\subsection{Channel and Propagation Model}
Building upon the spatial models $\Phi_{\rm H}^{\rm P}$ and $\Phi_{\rm H}^{\rm B}$, we adopt a reference propagation model that emphasizes the impact of patrol geometry while retaining analytical tractability. The large-scale path loss between a \ac{HAP} at $x\in\mathbb{R}^3$ and the typical user at $\mathbf{o}$ follows the power law $\ell(\|x-\mathbf{o}\|)=K\|x-\mathbf{o}\|^{-\eta}$, where $\eta\geq2$ is the path-loss exponent and $K = \left( \frac{c}{4\pi f_c} \right)^2$ is the free-space reference gain at $1$ km for carrier frequency $f_c$ and speed of light $c$. Small-scale fading on each HAP-to-user link is modeled as independent Rayleigh fading \cite{haenggi2013stochastic,andrews2011tractable}, with power gain $h_x\sim\mathrm{Exp}(1)$. Each \ac{HAP} transmits with constant power $P_t$. The serving link has antenna gain $G_t$, whereas interfering links have gain $G_I\leq G_t$. Assuming an isotropic user antenna with $G_r=1$, the received power from the serving \ac{HAP} at distance $d$ is $P_{\rm r}(d)=P_tG_tG_rh_xKd^{-\eta}$.

These assumptions provide a tractable reference model for isolating the effects of patrol-ring geometry. In particular, the fixed antenna gains do not account for random interfering beam orientations, while the homogeneous path-loss and Rayleigh fading models do not distinguish between \ac{LOS} and N\ac{LOS} propagation. Such propagation effects may influence absolute coverage performance but do not alter the underlying SCR spatial construction. A generalized blockage- and antenna-aware formulation is presented in Remark~\ref{rem:general_propagation}.

\subsection{Visibility Region and Earth Blockage}
\label{subsec:visibility}
We explicitly account for Earth blockage through a minimum elevation constraint imposed on the actual \ac{HAP} locations. Recall from the patrol-ring construction that all \acp{HAP} lie on the concentric platform shell of radius $R_{\rm F}$. For a minimum elevation angle $0\leq E_{\min}<\pi/2$ measured at the typical user, the maximum visible Euclidean link distance is
\begin{align}
    \bar d_{\rm H} &= \sqrt{ R_{\rm F}^{2} -R_\oplus^{2}\cos^{2} E_{\min} } -R_\oplus\sin E_{\min}.
    \label{eq:visibility_distance}
\end{align}
Accordingly, a \ac{HAP} at location $x$ is visible to the typical user if and only if 
\begin{align*}
    \|x-\mathbf{o}\|\leq\bar d_{\rm H}.
\end{align*}
Equivalently, if $\chi_x$ denotes the geocentric angle between the typical user and the actual \ac{HAP} location $x$, the visibility condition is $\chi_x\leq\chi_{\max}$, where
\begin{align*}
    \chi_{\max}=\arccos\left(\frac{R_{\rm F}^{2}+R_\oplus^{2}-\bar d_{\rm H}^{\,2}}{2R_{\rm F}R_\oplus}\right).
\end{align*}
Visibility is imposed on platform locations, not directly on patrol anchors. A ring whose anchor lies outside the corresponding angular cap may therefore still contain an eligible arc.

The case $E_{\min}=0$ corresponds to the geometric Earth horizon, for which $\bar d_{\rm hor}=\sqrt{R_{\rm F}^{2}-R_\oplus^{2}}$. Positive values of $E_{\min}$ impose a stricter visibility constraint. Throughout the analysis, Earth blockage is modeled geometrically; atmospheric refraction and local terrain blockage are not considered.

\subsection{Association Rule and SIR Definition}
For $k\in\{{\rm P},{\rm B}\}$, let $\Phi_{{\rm H},{\rm v}}^{k} \triangleq \{x\in\Phi_{\rm H}^{k}: \|x-\mathbf{o}\|\leq\bar d_{\rm H}\}$ denote the set of \acp{HAP} visible to the typical user, where $\bar d_{\rm H}$ is the maximum visible link distance. In both \ac{SCR-PCP} and \ac{SCR-BCP}, the typical user associates with its nearest-visible \ac{HAP}. Under the homogeneous reference propagation model, this rule is equivalent to maximum average received-power association, since all candidate \acp{HAP} transmit with the same power $P_t$, provide the same serving-link gain $G_t$ when selected, and follow the common path-loss law $Kd^{-\eta}$. Thus, when $\Phi_{{\rm H},{\rm v}}^{k}\neq\varnothing$, the serving \ac{HAP} is 
\begin{align*}
    x_{0,k} = \arg\min_{x\in\Phi_{{\rm H},{\rm v}}^{k}} \|x-\mathbf{o}\| = \arg\max_{x\in\Phi_{{\rm H},{\rm v}}^{k}} \left\{P_tG_tK\|x-\mathbf{o}\|^{-\eta}\right\}.
\end{align*}
Hence, nearest-visible \ac{HAP} association is the maximum-average-received-power rule under the adopted homogeneous channel model \cite{andrews2011tractable}. This equivalence generally does not hold under heterogeneous \ac{LOS}/N\ac{LOS} propagation or link-dependent antenna gains, as discussed in Remark~\ref{rem:max_power_association}.

Let $d_{0,k}=\|x_{0,k}-\mathbf{o}\|$ denote the serving distance under model $k$. If $\Phi_{{\rm H},{\rm v}}^{k}=\varnothing$, we set $d_{0,k}=\infty$ and declare a coverage outage. Otherwise, the received signal power is $P_k(d_{0,k})=P_tG_th_0Kd_{0,k}^{-\eta}$, while the aggregate interference from the other visible \acp{HAP} is $I_k=\sum_{x\in\Phi_{{\rm H},{\rm v}}^{k} \setminus\{x_{0,k}\}} P_tG_Ih_xK\|x-\mathbf{o}\|^{-\eta}$. Assuming an interference-limited regime where thermal noise is neglected, and factoring out the common term $P_tK$, the resulting \ac{SIR} is
\begin{align}
    \mathrm{SIR}_k = \frac{G_th_0d_{0,k}^{-\eta}} {\displaystyle \sum_{x\in\Phi_{{\rm H},{\rm v}}^{k} \setminus\{x_{0,k}\}} G_Ih_x\|x-\mathbf{o}\|^{-\eta}}.
    \label{eq:sir_def_coverage}
\end{align}

\section{Spatial Statistics of \ac{SCR-PCP} and \ac{SCR-BCP}}
\label{sec:stats}
Before analyzing the network's coverage and interference, we must establish the fundamental statistical properties of the \ac{SCR-PCP} and \ac{SCR-BCP} models. In this section, we prove the network's spatial isotropy, formulate the 3D visible arc length, and derive the various distance distributions.

\subsection{Isotropy of the Patrol-Ring Processes}
\label{subsec:isotropy}
A central structural property of the proposed \ac{SCR-PCP} and \ac{SCR-BCP} models is that they inherit the rotational invariance of the anchor \ac{PPP}. A point process is isotropic if its distribution is invariant under rotations about the origin. That is, for any rotation $Q$ of $\mathbb{R}^3$ such that $Q(\mathbf{0})=\mathbf{0}$ and any Borel set $A$ in the space of point configurations, $\mathbb{P}(\Phi \in A) = \mathbb{P}(Q(\Phi) \in A),$ where $Q(\Phi) \triangleq \{Q(x) : x \in \Phi\}$ denotes the rotated configuration.

\begin{theorem}
\label{thm:isotropy}
Let $\Phi_u$ be a homogeneous \ac{PPP} of anchors on the spherical shell $\mathcal{S}_{R_{\rm H}}$, and let $\mathcal{L}=\{\mathcal{C}(u,\gamma):u\in\Phi_u\}$ be the associated patrol-ring process constructed as in Section~\ref{sec:sys_model}. Conditioned on $\Phi_u$ and $\mathcal{L}$, let $\Phi_{\rm H}^{\rm P}$ and $\Phi_{\rm H}^{\rm B}$ denote the resulting \ac{HAP} point processes obtained from the \ac{SCR-PCP} and \ac{SCR-BCP} intra-ring models, respectively. Then, for any rotation $Q$ of $\mathbb{R}^3$ such that $Q(\mathbf{0})=\mathbf{0}$,
\begin{align*}
    Q(\Phi_{\rm H}^{k}) \;\stackrel{d}{=}\; \Phi_{\rm H}^{k}, \qquad k\in\{{\rm P,B}\},
\end{align*}
i.e., both $\Phi_{\rm H}^{\rm P}$ and $\Phi_{\rm H}^{\rm B}$ are isotropic point processes on $\mathbb{R}^3$.
\end{theorem}
\begin{IEEEproof}
Please see~\cite{shah2026spherical} for complete proof.
\end{IEEEproof}
The visibility condition $\|x-\mathbf{o}\|\leq\bar d_{\rm H}$ is preserved under a common rotation of the network and user, since $\|Q(x)-Q(\mathbf{o})\|=\|x-\mathbf{o}\|$. Together with the homogeneous reference propagation assumptions, this implies that visible-HAP availability, the nearest-visible serving distance, aggregate visible interference, and coverage probability have the same distributions at every ground location. We therefore fix the typical user at the North Pole, $\mathbf{o}=(0,0,R_\oplus)$.

\subsection{Distance Distributions}
\label{sec:distance}
To characterize network performance and evaluate coverage probability, we first derive the distance distributions from the typical user to the anchors, patrol rings, and HAPs under \ac{SCR-PCP} and \ac{SCR-BCP}. Following the hierarchical anchor-ring construction, we begin with the anchor process, then consider the patrol rings, and finally characterize the induced HAP process. We define the following distances:\\
$1)$ $D_{\rm u}$: Euclidean distance from the typical user to the nearest anchor in the parent \ac{PPP} $\Phi_u$, i.e., $D_{\rm u}=\inf_{u\in\Phi_u} \|u-\mathbf{o}\|$. \\
$2)$ $D_{\rm r}(\phi)$: Minimum Euclidean distance from the user to the patrol ring $\mathcal{C}(u,\gamma)$ associated with an anchor $u$ located at central angle $\phi$.\\
$3)$ $D_{\rm R}$: Distance from the user to the nearest patrol ring in the network, i.e., $D_{\rm R} = \inf_{u \in \Phi_u} D_{\rm r}(\phi_u).$\\
$4)$ $d_{0,{\rm P}}$: Euclidean distance from the user to the nearest-visible \ac{HAP} in the \ac{SCR-PCP}, i.e., $d_{0,{\rm P}} = \inf\left\{ \|x-\mathbf{o}\|: x\in\Phi_{\rm H}^{\rm P},\, \|x-\mathbf{o}\|\leq\bar d_{\rm H} \right\}.$ \\
$5)$ $d_{0,{\rm B}}$: Euclidean distance from the user to the nearest-visible \ac{HAP} in the \ac{SCR-BCP}, i.e., $d_{0,{\rm B}} = \inf\left\{ \|x-\mathbf{o}\|: x\in\Phi_{\rm H}^{\rm B},\, \|x-\mathbf{o}\|\leq\bar d_{\rm H} \right\}.$

Throughout, we adopt the convention $\inf\varnothing=\infty$. In particular, the probability mass at infinity for the nearest-visible \ac{HAP} distances corresponds to the absence of a visible serving platform. All expressions retain the three-dimensional spherical geometry and local Euclidean patrol structure without planar or small-angle approximations.

\subsubsection{Distance to the Nearest Anchor}
We begin by deriving the \ac{CDF} of $D_{\rm u}$.
\begin{lemma}
\label{lem:dist_nearest_anchor}
The \ac{CDF} of the nearest-anchor distance, for finite $d$, is
\begin{align*}
    &F_{D_{\rm u}}(d) = \\
    &\begin{cases}
    0, \hspace*{5.25cm} d<H, \\
    1-\exp\left( -2\pi R_{\rm H}^2\lambda_{\rm u} \left( 1- \dfrac{R_\oplus^2+R_{\rm H}^2-d^2} {2R_\oplus R_{\rm H}} \right) \right), \\ \hspace*{5.5cm} H\leq d\leq R_\oplus+R_{\rm H},\\
    1-\exp\left(-4\pi R_{\rm H}^2\lambda_{\rm u}\right), \hspace*{2.25cm} d>R_\oplus+R_{\rm H}.
    \end{cases}
\end{align*}
The remaining probability mass $\mathbb{P}(D_{\rm u}=\infty) = \exp\left(-4\pi R_{\rm H}^2\lambda_{\rm u}\right)$ corresponds to the event that the anchor process is empty.
\end{lemma}
\begin{IEEEproof}
Please see~\cite{shah2026spherical} for complete proof.    
\end{IEEEproof}

\subsubsection{Distance from User to a Given Ring}
We next characterize the minimum distance from the user to a patrol ring associated with an anchor at angle $\phi$.
\begin{lemma}
\label{lem:dist_ring}
For an anchor located at central angle $\phi$ relative to the typical user, the minimum Euclidean distance from the user to the corresponding patrol ring $\mathcal{C}(u,\gamma)$ is
\begin{align}
    D_{\rm r}(\phi) = \sqrt{(R_\oplus \cos\phi - R_{\rm H})^2 + (R_\oplus \sin\phi - a)^2}.
    \label{eq:dist_ring}
\end{align}
\end{lemma}
\begin{IEEEproof}
Using the coordinate system introduced in Section~\ref{sec:sys_model}, the anchor is located at $u=(R_{\rm H}\sin\phi,0,R_{\rm H}\cos\phi)$, with outward unit normal $n_u=(\sin\phi,0,\cos\phi)$. Decomposing $\mathbf{o}-u$ into components normal and parallel to the tangent plane $\Pi(u)$, the perpendicular distance from the user to the plane is
\begin{align*}
    d_{\rm plane}(\phi) = |(\mathbf{o}-u)\cdot n_u| = |R_\oplus\cos\phi-R_{\rm H}|.
\end{align*}
The projection of the user onto $\Pi(u)$ lies at distance $\rho(\phi)=R_\oplus\sin\phi$ from the ring center $u$. Since the patrol ring has radius $a$, the minimum in-plane distance to the ring is $|\rho(\phi)-a|$. By orthogonality of the two components,
\begin{align*}
    D_{\rm r}(\phi) &= \sqrt{d_{\rm plane}(\phi)^2 +(\rho(\phi)-a)^2}\\ 
    &= \sqrt{(R_\oplus\cos\phi-R_{\rm H})^2 +(R_\oplus\sin\phi-a)^2},
\end{align*}
which completes the proof.
\end{IEEEproof}

\subsubsection{Distance Distribution to the Nearest Ring}
We now derive the distribution of the nearest-ring distance $D_{\rm R}=\inf_{u\in\Phi_u}D_{\rm r}(\phi_u)$, with $\inf\varnothing=\infty$. Let $\phi_0=\arctan(a/R_{\rm H})$. The minimum attainable ring distance is $D_{\min}=D_{\rm r}(\phi_0)=R_{\rm F}-R_\oplus$, whereas the maximum is $D_{\max}=D_{\rm r}(\pi) =\sqrt{(R_\oplus+R_{\rm H})^2+a^2}$.
\begin{lemma}
\label{lem:nearest_ring}
The \ac{CDF} of the nearest-ring distance, for finite $d$, is
\begin{align*}
    F_{D_{\rm R}}(d) &=
    \begin{cases}
        0, & d<D_{\min},\\
        1-\exp\big(-\lambda_{\rm u}A_{\rm R}(d)\big), & D_{\min}\leq d<D_{\max},\\
        1-\exp(-4\pi R_{\rm H}^2\lambda_{\rm u}), & d\geq D_{\max},
    \end{cases}
\end{align*}
where $A_{\rm R}(d)=2\pi R_{\rm H}^2 [\cos\phi_1(d)-\cos\phi_2(d)]$ is the area of the spherical region containing anchors whose patrol rings lie within distance $d$ of the user. The angular boundaries are
\begin{align}
    \phi_1(d) &=\max\left\{0,\phi_0-\arccos\Delta(d)\right\},
    \label{eq:phi1}\\
    \phi_2(d) &=\min\left\{\pi,\phi_0+\arccos\Delta(d)\right\},
    \label{eq:phi2}
\end{align}
with $\Delta(d) = \frac{R_\oplus^2+R_{\rm F}^2-d^2} {2R_\oplus R_{\rm F}}.$ The remaining probability mass $\mathbb{P}(D_{\rm R}=\infty) =\exp(-4\pi R_{\rm H}^2\lambda_{\rm u})$ corresponds to an empty anchor process.
\end{lemma}
\begin{IEEEproof}
For $D_{\min}\leq d<D_{\max}$, define $\mathcal B(d) = \{u\in\mathcal S_{R_{\rm H}}: D_{\rm r}(\phi_u)\leq d\}.$ This is the region of anchors whose patrol rings intersect the closed ball of radius $d$ centered at the user. Thus, $\{D_{\rm R}>d\}$ is the event that $\mathcal B(d)$ contains no anchor. From Lemma~\ref{lem:dist_ring},
\begin{align*}
    D_{\rm r}(\phi)^2 &= R_\oplus^2+R_{\rm H}^2+a^2 -2R_\oplus(R_{\rm H}\cos\phi+a\sin\phi)\\ 
    &= R_\oplus^2+R_{\rm F}^2 -2R_\oplus R_{\rm F}\cos(\phi-\phi_0),
\end{align*}
where $\phi_0=\arctan(a/R_{\rm H})$. Thus, an anchor generates a ring within distance $d$ if $\cos(\phi-\phi_0)\geq\Delta(d).$ Intersecting the resulting angular interval with $[0,\pi]$ yields \eqref{eq:phi1} and \eqref{eq:phi2}. The corresponding anchor-region area is
\begin{align*}
    |\mathcal{B}(d)| &= 2\pi R_{\rm H}^2 \int_{\phi_1(d)}^{\phi_2(d)} \sin\phi\,{\rm d}\phi \\
    &= 2\pi R_{\rm H}^2 [\cos\phi_1(d)-\cos\phi_2(d)] =A_{\rm R}(d).
\end{align*}
Since the anchors form a homogeneous \ac{PPP} of surface intensity $\lambda_{\rm u}$, the void probability gives $\mathbb{P}(D_{\rm R}>d) = \exp\big(-\lambda_{\rm u}A_{\rm R}(d)\big).$ For $d<D_{\min}$, no ring can intersect the ball, whereas for $d\geq D_{\max}$, every existing ring intersects it. The latter case yields $\mathbb{P}(D_{\rm R}>d) =\mathbb{P}(\Phi_u=\varnothing) =\exp(-4\pi R_{\rm H}^2\lambda_{\rm u})$, completing the proof.
\end{IEEEproof}

\subsubsection{Distance Distribution to the nearest-visible \ac{HAP} under \ac{SCR-PCP}}
\begin{theorem}
\label{thm:nearest_haps_pcp}
Let $d_{0,{\rm P}}$ denote the distance from the typical user to the nearest-visible \ac{HAP} under \ac{SCR-PCP}, with $d_{0,{\rm P}}=\infty$ when no visible \ac{HAP} exists. Its \ac{CCDF} is
\begin{align}
    \mathbb{P}(d_{0,{\rm P}}>d) =
    \begin{cases}
        1, & d<D_{\min},\\[0.3em]
        V_{\rm P}(d), & D_{\min}\leq d\leq\bar d_{\rm H},\\[0.3em]
        V_{\rm P}(\bar d_{\rm H}), & d>\bar d_{\rm H},
    \end{cases}
    \label{eq:ccdf_nearest_haps_pcp}
\end{align}
where, for $D_{\min}\leq d\leq\bar d_{\rm H}$,
\begin{align*}
    V_{\rm P}(d) \!=\! \exp\left(\!\! -2\pi R_{\rm H}^{2}\lambda_{\rm u} \int_{\phi_1(d)}^{\phi_2(d)} \!\! \big[1 \!-\! p_{{\rm void},{\rm P}}(\phi,d) \big] \sin\phi\,{\rm d}\phi \right),
\end{align*}
and $p_{{\rm void},{\rm P}}(\phi,d) = \exp[-\nu L(\phi,d)].$ The boundaries $\phi_1(d)$ and $\phi_2(d)$ are given in \eqref{eq:phi1} and \eqref{eq:phi2}. Here, $L(\phi,d)$ denotes the arc length of a patrol ring at central angle $\phi$ lying within Euclidean distance $d$ of the user:
\begin{align*}
    L(\phi,d) =
    \begin{cases}
        0, & \Psi(\phi,d)<0,\\[0.3em]
        2a\arccos\left(1-\Psi(\phi,d)\right),
        & 0\leq\Psi(\phi,d)<2,\\[0.3em]
        2\pi a, & \Psi(\phi,d)\geq2,
    \end{cases}
\end{align*}
where $\Psi(\phi,d)= \frac{d^2-D_{\rm r}(\phi)^2} {2aR_\oplus\sin\phi}$ for $\phi\in(0,\pi)$. The probability mass $\mathbb{P}(d_{0,{\rm P}}=\infty) =V_{\rm P}(\bar d_{\rm H})$ corresponds to the absence of a visible \ac{HAP}.
\end{theorem}
\begin{IEEEproof}
Consider a patrol ring whose anchor lies at central angle $\phi$. Let $\alpha\in[-\pi,\pi]$ denote the angular displacement of a \ac{HAP} from the closest point on the ring to the user. Using the geometric decomposition in Lemma~\ref{lem:dist_ring}, its squared Euclidean distance to the user is
\begin{align*}
    d(\phi,\alpha)^2 &= D_{\rm r}(\phi)^2 +2aR_\oplus\sin\phi(1-\cos\alpha).
\end{align*}
Hence, the condition $d(\phi,\alpha)\leq d$ is equivalent to
\begin{align*}
    1 - \cos\alpha \le \frac{d^2 - D_{\rm r}(\phi)^2}{2a R_\oplus \sin\phi} = \Psi(\phi,d).
\end{align*}
For $\Psi(\phi,d)<0$, the ball of radius $d$ does not intersect the ring. For $0\leq\Psi(\phi,d)<2$, the intersection is a symmetric arc with half-angle $\arccos(1-\Psi(\phi,d))$. Finally, $\Psi(\phi,d)\geq2$ corresponds to the entire ring being contained in the ball. These cases yield the stated expression for $L(\phi,d)$.

Conditioned on the anchor, the \acp{HAP} form a one-dimensional \ac{PPP} of intensity $\nu$ along the ring. Therefore, the probability that the ring contains no \ac{HAP} within distance $d$ is
\begin{align*}
    p_{{\rm void},{\rm P}}(\phi,d) = \exp\left(-\nu L(\phi,d)\right).
\end{align*}
For $D_{\min}\leq d\leq\bar d_{\rm H}$, the event $\{d_{0,{\rm P}}>d\}$ requires every ring to contain no \ac{HAP} within distance $d$. Conditioning on the anchor process and applying its \ac{PGFL} gives
\begin{align*}
    \mathbb{P}(d_{0,{\rm P}}>d) &= \mathbb{E}_{\Phi_u} \left[ \prod_{u_i\in\Phi_u} p_{{\rm void},{\rm P}}(\phi_i,d) \right]\\ 
    &\hspace*{-1cm} = \exp\Bigg( -2\pi R_{\rm H}^2\lambda_{\rm u} \int_0^\pi \left[1-p_{{\rm void},{\rm P}}(\phi,d)\right] \sin\phi\,{\rm d}\phi \Bigg).
\end{align*}
Using ${\rm d}S = R_{\rm H}^2 \sin\phi\, {\rm d}\phi\, {\rm d}\theta$ and integrating over $\theta$ gives
\begin{align*}
    \mathbb{P}(d_{0,{\rm P}} > d) &= \exp\bigg( -2\pi R_{\rm H}^2 \lambda_{\rm u} \int_0^\pi \\
    &\left[1 - \exp(-\nu L(\phi,d))\right] \sin\phi\, {\rm d}\phi \bigg).
\end{align*}
By Lemma~\ref{lem:nearest_ring}, $L(\phi,d)>0$ only for $\phi\in[\phi_1(d),\phi_2(d)]$, yielding $V_{\rm P}(d)$. 

For $d<D_{\min}$, no \ac{HAP} can lie within distance $d$. Conversely, for $d>\bar d_{\rm H}$, no additional \acp{HAP} become visible, so the \ac{CCDF} remains $V_{\rm P}(\bar d_{\rm H})$. This establishes \eqref{eq:ccdf_nearest_haps_pcp}, including the probability mass at infinity.
\end{IEEEproof}

The preceding theorem characterizes the serving-distance distribution. For coverage analysis, we also require the central angle of the ring containing the serving \ac{HAP}.
\begin{corollary}
\label{cor:serving_angle_pcp}
Let $\phi_{\rm s}$ denote the central angle of the anchor associated with the serving \ac{HAP} under \ac{SCR-PCP}. For $D_{\min}<x<\bar d_{\rm H}$, the conditional density of the serving-ring anchor angle given $d_{0,{\rm P}}=x$ is
\begin{align}
    &f_{\phi_{\rm s}\mid d_{0,{\rm P}}}(\phi\mid x) = \frac{ w_{\rm P}(\phi,x)\sin\phi }{ \displaystyle\int_{\phi_1(x)}^{\phi_2(x)} w_{\rm P}(\varphi,x) \sin\varphi\,{\rm d}\varphi }.
    \label{eq:conditional_serving_angle_pcp}
\end{align}
for $\phi\in[\phi_1(x),\phi_2(x)]$, and zero otherwise. The corresponding joint density of the finite serving distance and serving-ring angle is
\begin{align}
    f_{d_{0,{\rm P}},\phi_{\rm s}}(x,\phi) &= 2\pi\lambda_{\rm u}R_{\rm H}^2 V_{\rm P}(x) w_{\rm P}(\phi,x) \sin\phi.
    \label{eq:joint_pdf_pcp}
\end{align}
where $w_{\rm P}(\phi,x) \triangleq -\partial_xp_{{\rm void},{\rm P}}(\phi,x) = \nu e^{-\nu L(\phi,x)} \partial_xL(\phi,x).$ For $0<\phi<\pi$ and
$0<\Psi(\phi,x)<2$, this becomes
\begin{align}
    w_{\rm P}(\phi,x) = \frac{ 2\nu x e^{-\nu L(\phi,x)} }{ R_\oplus\sin\phi \sqrt{\Psi(\phi,x)[2-\Psi(\phi,x)]} }.
    \label{eq:explicit_derivative_pvoid_pcp}
\end{align}
Outside the partial-intersection regime, $w_{\rm P}(\phi,x)=0$ almost everywhere.
\end{corollary}
\begin{IEEEproof}
Please see~\cite{shah2026spherical} for complete proof.   
\end{IEEEproof}

\subsubsection{Distance Distribution to the nearest-visible \ac{HAP} under \ac{SCR-BCP}}
We next derive the distance distribution to the nearest-visible \ac{HAP} under the \ac{SCR-BCP}. The anchor and patrol-ring geometry are identical to the \ac{SCR-PCP} case; only the intra-ring point process changes.
\begin{theorem}
\label{thm:nearest_haps_bcp}
Let $d_{0,{\rm B}}$ denote the distance to the nearest-visible \ac{HAP} under \ac{SCR-BCP}. For finite $d$, its \ac{CCDF} is
\begin{align}
    \mathbb{P}(d_{0,{\rm B}}>d) =
    \begin{cases}
        1, & d<D_{\min},\\[0.3em]
        V_{\rm B}(d), & D_{\min}\leq d\leq\bar d_{\rm H},\\[0.3em]
        V_{\rm B}(\bar d_{\rm H}), & d>\bar d_{\rm H},
    \end{cases}
    \label{eq:ccdf_nearest_haps_bcp}
\end{align}
where, for $D_{\min}\leq d\leq\bar d_{\rm H}$,
\begin{align}
    V_{\rm B}(d) = \exp\Bigg( &-2\pi R_{\rm H}^{2}\lambda_{\rm u} \int_{\phi_1(d)}^{\phi_2(d)} \big[ 1-p_{{\rm void},{\rm B}}(\phi,d) \big] \nonumber\\ &\hspace{2.5cm} \times\sin\phi\,{\rm d}\phi \Bigg),
    \label{eq:void_nearest_haps_bcp}
\end{align}
and $p_{{\rm void},{\rm B}}(\phi,d) = \left( 1-\frac{L(\phi,d)}{2\pi a} \right)^{n_{\rm H}}.$ The boundaries $\phi_1(d)$ and $\phi_2(d)$ are given in \eqref{eq:phi1} and \eqref{eq:phi2}, and $L(\phi,d)$ is defined in Theorem~\ref{thm:nearest_haps_pcp}. The remaining probability mass is $\mathbb P(d_{0,{\rm B}}=\infty) = V_{\rm B}(\bar d_{\rm H}),$ corresponding to the absence of a visible HAP.
\end{theorem}
\begin{IEEEproof}
For a ring at central angle $\phi$, the probability that a uniformly distributed \ac{HAP} lies beyond distance $d$ is $1-L(\phi,d)/(2\pi a)$. Since each ring contains $n_{\rm H}$ independently and uniformly positioned \acp{HAP}, the ring-wise void probability is 
\begin{align*}
    p_{{\rm void},{\rm B}}(\phi,d) = \left(1-\frac{L(\phi,d)}{2\pi a}\right)^{n_{\rm H}}.
\end{align*}
For $D_{\min}\leq d\leq\bar d_{\rm H}$, all \acp{HAP} within distance $d$ are geometrically visible. Thus, $\{d_{0,{\rm B}}>d\}$ requires every ring to contain no \ac{HAP} within this distance. Independence across rings and the \ac{PGFL} of the parent \ac{PPP} give
\begin{align*}
    \mathbb{P}(d_{0,{\rm B}}>d) &= \mathbb{E}_{\Phi_u} \left[ \prod_{u_i\in\Phi_u} p_{{\rm void},{\rm B}}(\phi_i,d) \right]\\ 
    &\hspace*{-1cm} = \exp\Bigg( -2\pi R_{\rm H}^2\lambda_{\rm u} \int_0^\pi \left[1-p_{{\rm void},{\rm B}}(\phi,d)\right] \sin\phi\,{\rm d}\phi \Bigg).
\end{align*}
By Lemma~\ref{lem:nearest_ring}, $L(\phi,d)>0$ only for $\phi\in[\phi_1(d),\phi_2(d)]$, yielding \eqref{eq:void_nearest_haps_bcp}. For $d<D_{\min}$, no \ac{HAP} can lie within distance $d$, whereas for $d>\bar d_{\rm H}$, the \ac{CCDF} remains $V_{\rm B}(\bar d_{\rm H})$ because no additional \acp{HAP} become visible. This establishes \eqref{eq:ccdf_nearest_haps_bcp}.
\end{IEEEproof}
As in the \ac{SCR-PCP} case, the conditional serving-ring-angle distribution follows directly from the ring-wise void probability.
\begin{corollary}
\label{cor:serving_angle_bcp}
For $D_{\min}<x<\bar d_{\rm H}$, the conditional density of the serving-ring anchor angle $\phi_{\rm s}$ given $d_{0,{\rm B}}=x$ is
\begin{align}
    &f_{\phi_{\rm s}\mid d_{0,{\rm B}}}(\phi\mid x) = \frac{ w_{\rm B}(\phi,x)\sin\phi }{ \displaystyle \int_{\phi_1(x)}^{\phi_2(x)} w_{\rm B}(\varphi,x) \sin\varphi\,{\rm d}\varphi },
    \label{eq:conditional_serving_angle_bcp}
\end{align}
for $\phi\in[\phi_1(x),\phi_2(x)]$, and zero otherwise. The corresponding joint density of the finite serving distance and serving-ring angle is
\begin{align}
    f_{d_{0,{\rm B}},\phi_{\rm s}}(x,\phi) &= 2\pi\lambda_{\rm u}R_{\rm H}^2 V_{\rm B}(x)w_{\rm B}(\phi,x)\sin\phi.
    \label{eq:joint_pdf_bcp}
\end{align}
where $w_{\rm B}(\phi,x) \triangleq -\partial_xp_{{\rm void},{\rm B}}(\phi,x) = \frac{n_{\rm H}}{2\pi a} \left( 1-\frac{L(\phi,x)}{2\pi a} \right)^{n_{\rm H}-1} \partial_xL(\phi,x).$ For $0<\phi<\pi$ and $0<\Psi(\phi,x)<2$, this becomes
\begin{align}
    w_{\rm B}(\phi,x) &= \frac{n_{\rm H}x} {\pi aR_\oplus\sin\phi} \frac{ \left( 1-\dfrac{L(\phi,x)}{2\pi a} \right)^{n_{\rm H}-1} }{ \sqrt{\Psi(\phi,x)[2-\Psi(\phi,x)]} }.
    \label{eq:explicit_derivative_pvoid_bcp}
\end{align}
Outside the partial-intersection regime, $w_{\rm B}(\phi,x)=0$ almost everywhere.
\end{corollary}
\begin{IEEEproof}
See Appendix~\ref{app:cor2}.    
\end{IEEEproof}
\begin{figure*}[t]
\centering
\subfloat[]
{\includegraphics[width = 0.22\textwidth]{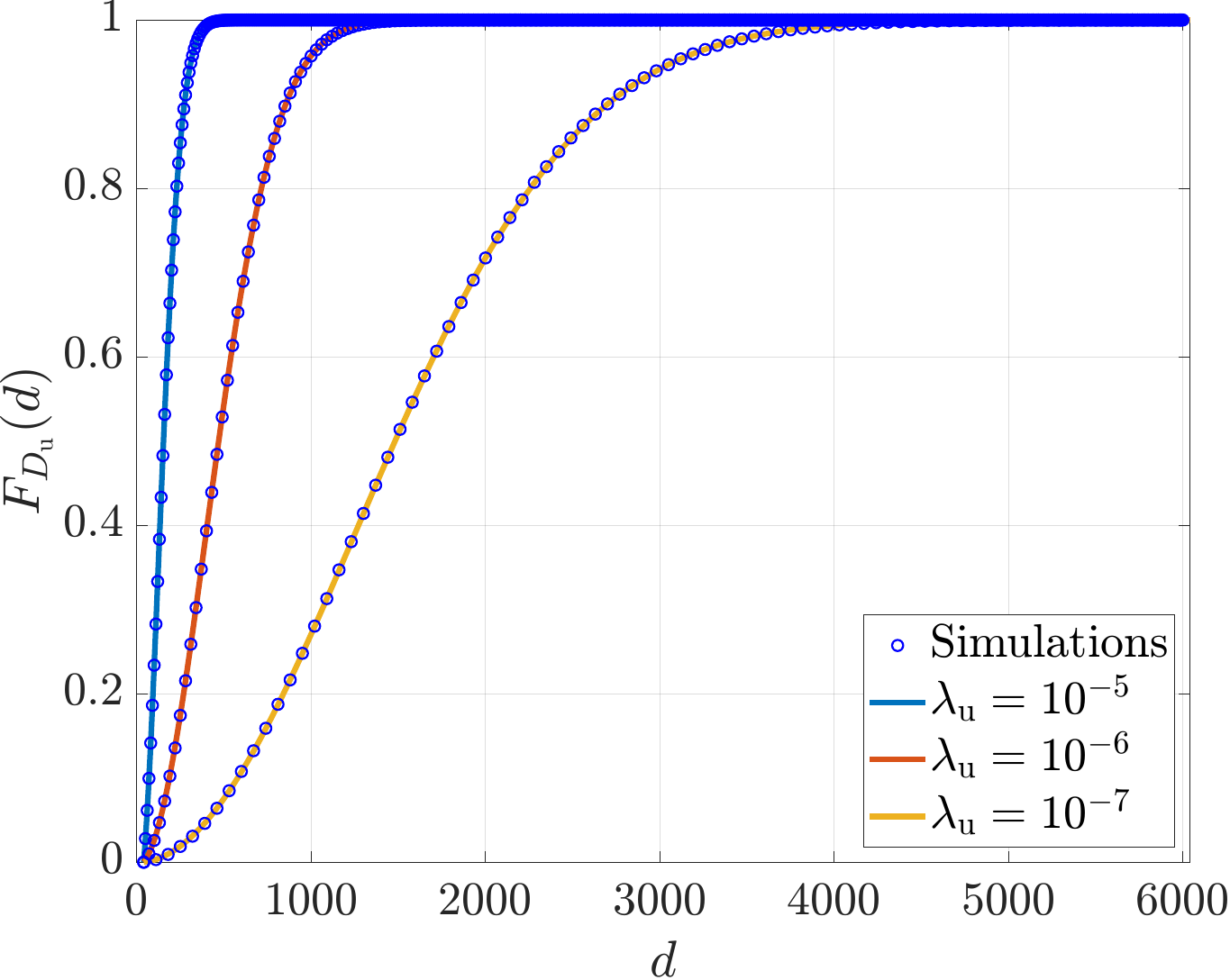}
\label{fig:result_anchor}}
\hfil
\subfloat[]
{\includegraphics[width = 0.22\textwidth]{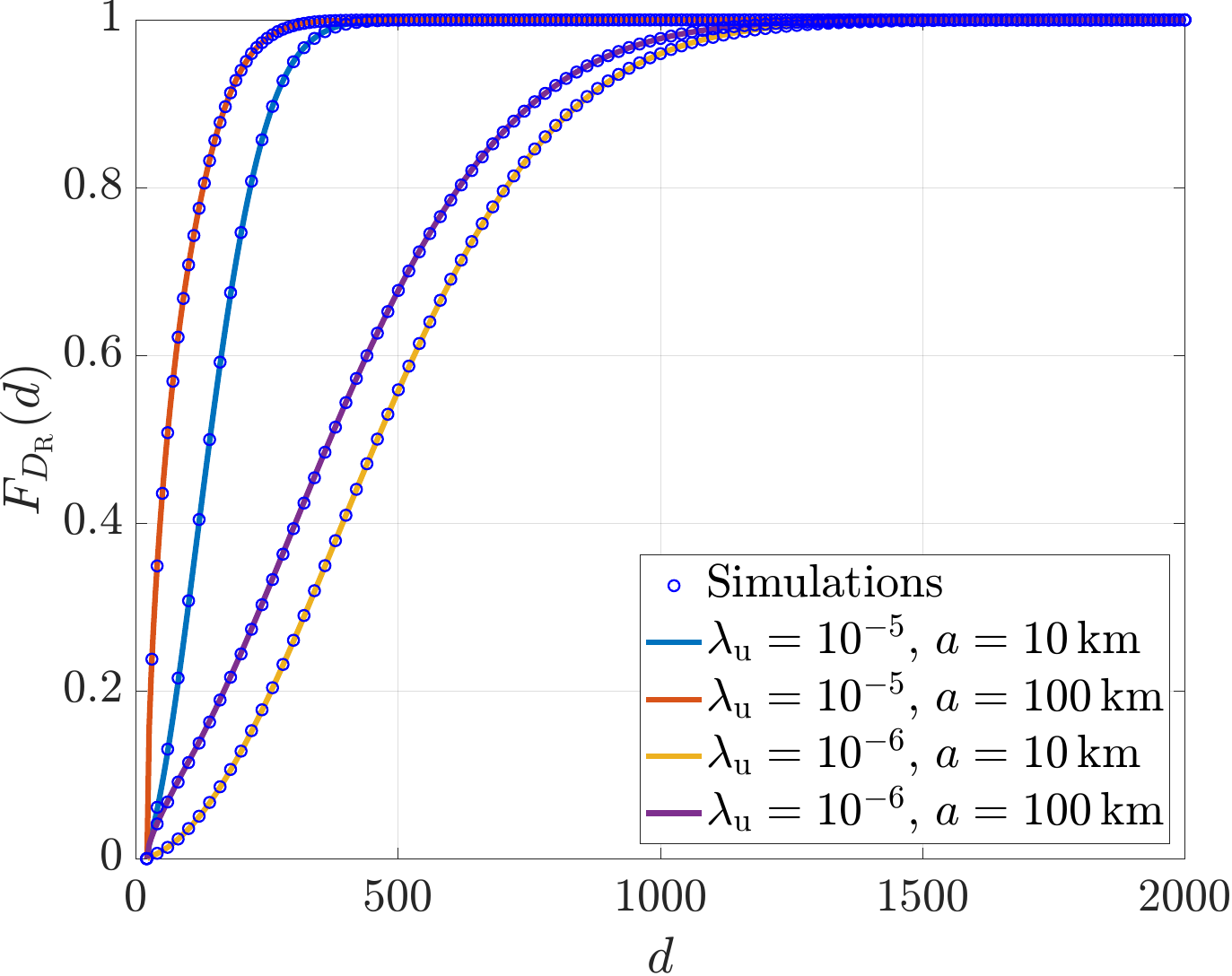}
\label{fig:result_ring}}
\hfil
\subfloat[]
{\includegraphics[width = 0.22\textwidth]{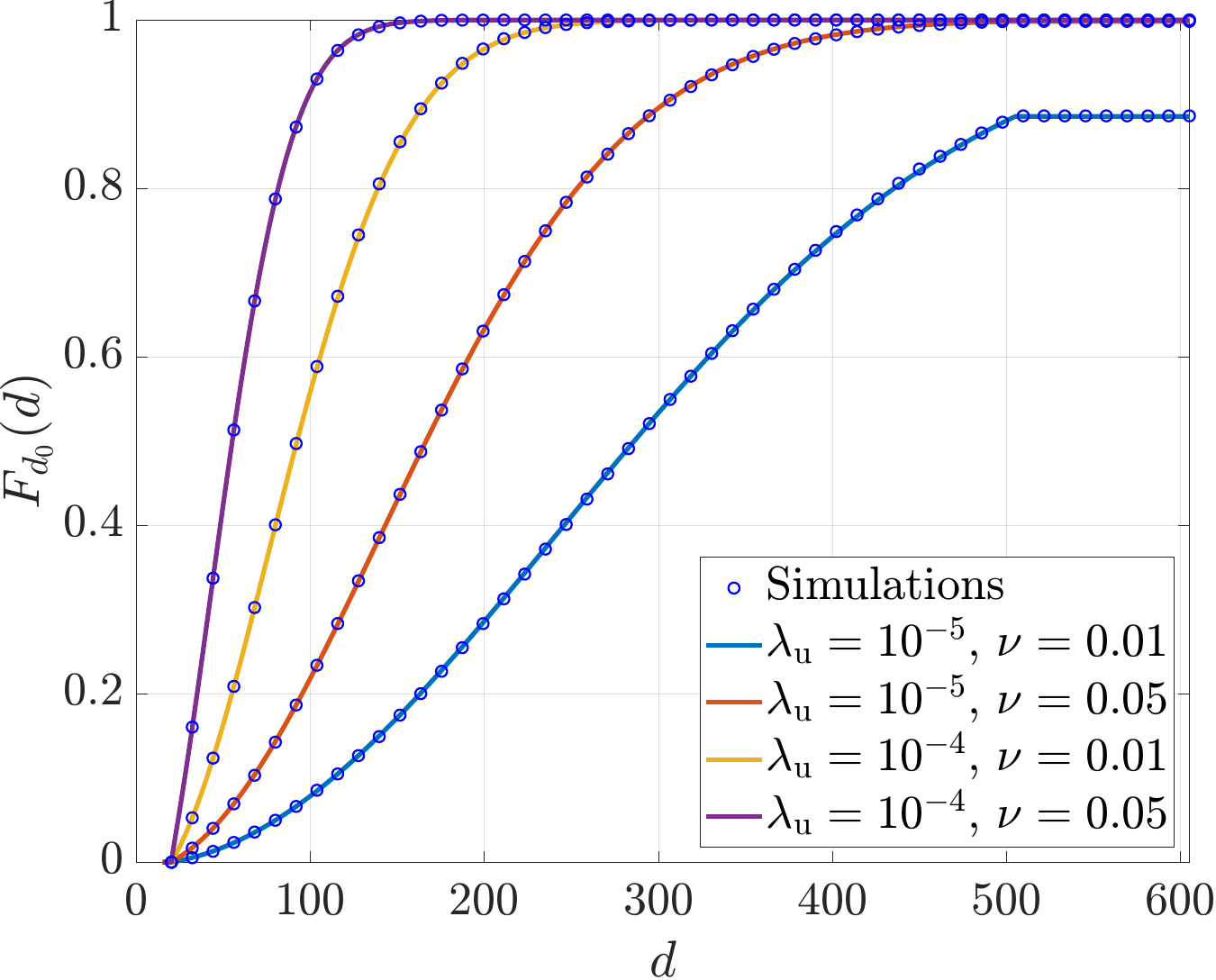}
\label{fig:result_HAP}}
\hfil
\subfloat[]
{\includegraphics[width = 0.22\textwidth]{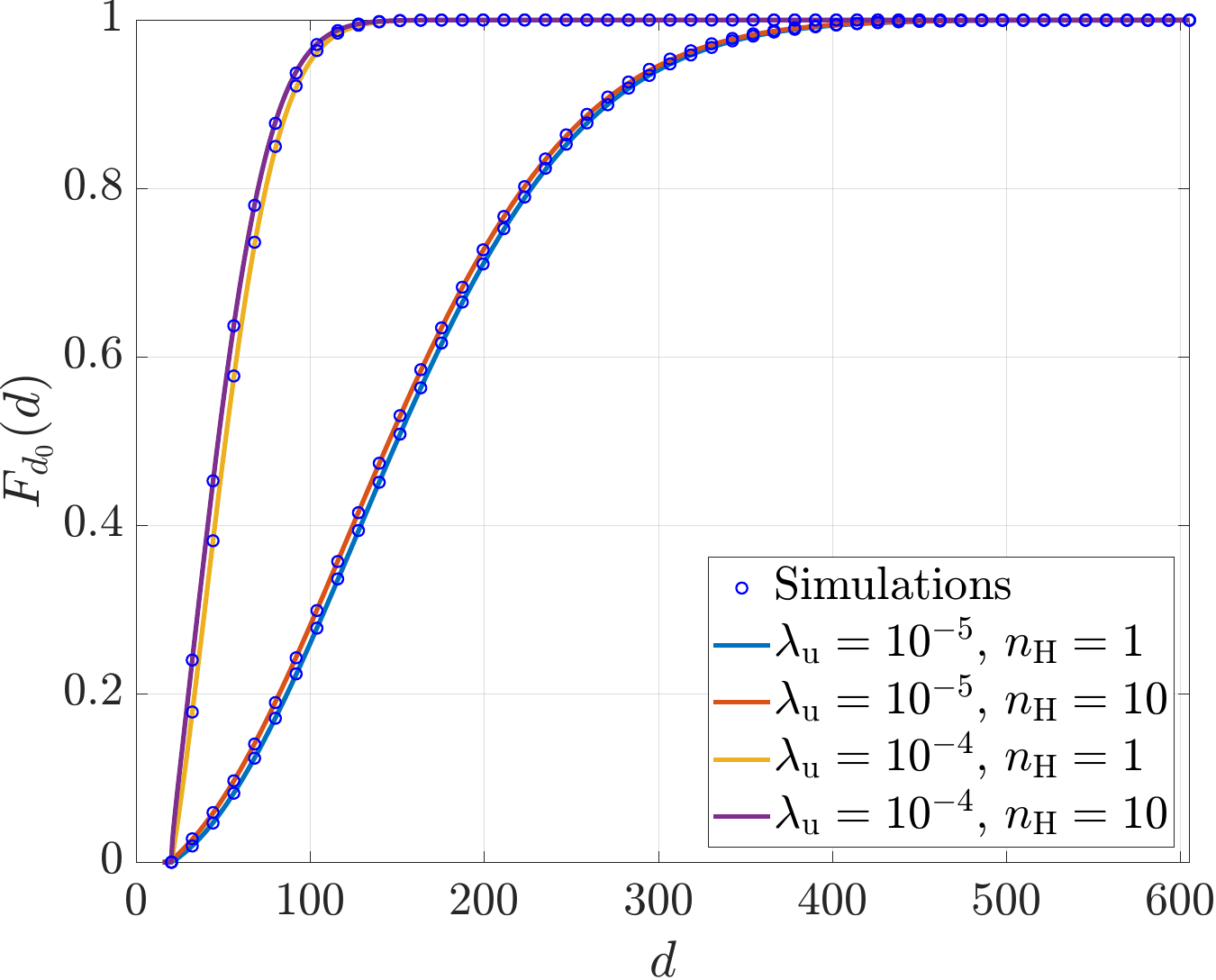}
\label{fig:nearest_haps_bcp}}
\caption{Analytical and simulated \acp{CDF} of the distances to (a) the nearest anchor, (b) the nearest patrol ring, and the nearest-visible \ac{HAP} under (c) \ac{SCR-PCP} and (d) \ac{SCR-BCP}.}
\label{fig:result_nearest}
\end{figure*}

\subsection{Validation of Distance Statistics}
\label{subsec:distance_validation}

The analytical distance distributions derived in Lemmas~\ref{lem:dist_nearest_anchor} and~\ref{lem:nearest_ring} and Theorems~\ref{thm:nearest_haps_pcp} and~\ref{thm:nearest_haps_bcp} are validated against Monte Carlo simulations in Fig.~\ref{fig:result_nearest}. The close agreement between analytical and simulated \acp{CDF} confirms the accuracy of the derived distance distributions.

Figure~\ref{fig:result_anchor} shows that increasing the anchor intensity $\lambda_{\rm u}$ shifts the nearest-anchor distance distribution toward smaller values. Similarly, Fig.~\ref{fig:result_ring} shows that the nearest-ring distance $D_{\rm R}$ is strongly influenced by $\lambda_{\rm u}$, whereas the effect of the patrol radius $a$ is relatively small over the considered range and becomes more pronounced for larger radii.

Figure~\ref{fig:result_HAP} validates the nearest-visible \ac{HAP} distance distribution under \ac{SCR-PCP}. Increasing either $\lambda_{\rm u}$ or the intra-ring intensity $\nu$ increases the probability of finding a nearby serving \ac{HAP}. Similarly, Fig.~\ref{fig:nearest_haps_bcp} validates the corresponding distribution under \ac{SCR-BCP}, where increasing either $\lambda_{\rm u}$ or the per-ring population $n_{\rm H}$ shifts the distribution toward smaller distances. The effect of $\lambda_{\rm u}$ is more pronounced over the considered configurations. For both models, the \acp{CDF} saturate at the visibility limit $\bar d_{\rm H}$, with $F_{d_{0,k}}(\bar d_{\rm H}) =1-V_k(\bar d_{\rm H})$, $k\in\{{\rm P},{\rm B}\}$. The limiting value represents the probability that at least one visible \ac{HAP} is available.

\section{Coverage Probability Analysis}
\label{sec:coverage}
This section derives coverage probability expressions for the \ac{SCR-PCP} and \ac{SCR-BCP} networks using the distance distributions developed in Section~\ref{sec:stats} and the SIR model.

\subsection{Interference Decomposition and Ring Geometry}
\label{subsec:IDRG}
Recall that the typical user associates with the nearest-visible \ac{HAP} under model $k\in\{{\rm P},{\rm B}\}$, with serving distance $d_{0,k}$ characterized in Theorems~\ref{thm:nearest_haps_pcp} and~\ref{thm:nearest_haps_bcp}. We condition on a finite serving distance $d_{0,k}=x$, since the event $d_{0,k}=\infty$ corresponds to a coverage outage. Under this conditioning, only the other visible \acp{HAP} contribute interference. We use the following notation.
\begin{enumerate}
    \item For an anchor at central angle $\phi$, the minimum distance from the user to its patrol ring $\mathcal{C}(u,\gamma)$ is $D_{\rm r}(\phi)$, as derived in Lemma~\ref{lem:dist_ring}. Thus, the Euclidean distance to a \ac{HAP} on that ring at local angle $\alpha$, measured from the closest point on the ring to the user, is
    \begin{equation}
        d(\phi,\alpha) = \sqrt{D_{\rm r}(\phi)^2 +2aR_\oplus\sin\phi(1-\cos\alpha)}.
        \label{eq:d_phi_alpha_cov_short}
    \end{equation}
    \item Conditioned on $d_{0,k}=x<\infty$, nearest-visible \ac{HAP} association excludes all other \acp{HAP} satisfying $d(\phi,\alpha)<x$. The corresponding exclusion half-angle is
    \begin{align}
        \alpha_0(x,\phi) =
        \begin{cases}
            0, & \Psi(\phi,x)\leq0,\\
            \arccos\big(1-\Psi(\phi,x)\big), & 0<\Psi(\phi,x)<2, \\
            \pi, & \Psi(\phi,x)\geq2,
        \end{cases}
        \label{eq:alpha0_cov}
    \end{align}
    where $\Psi(\phi,x)$ is defined in Theorem~\ref{thm:nearest_haps_pcp}. Equivalently, the excluded arc length is $L(\phi,x)=2a\alpha_0(x,\phi)$.
    \item The visibility constraint defines the maximum visible half-angle as $\alpha_{\rm v}(\phi)\triangleq\alpha_0(\bar d_{\rm H},\phi)$. Since $x\leq\bar d_{\rm H}$, we have $\alpha_0(x,\phi)\leq\alpha_{\rm v}(\phi)$. Thus, visible interferers can lie only on the residual arcs satisfying
    \begin{align*}
        \alpha_0(x,\phi) \leq|\alpha| \leq\alpha_{\rm v}(\phi).
    \end{align*}
    If $\alpha_0(x,\phi)=\alpha_{\rm v}(\phi)$, the residual visible arc has zero length and contributes no interference.
    \item The aggregate visible interference is decomposed as $I_k=P_tK(I_{{\rm s},k}+I_{{\rm o},k})$, where $I_{{\rm s},k}$ is the interference from other visible \acp{HAP} on the serving ring, and $I_{{\rm o},k}$ is the interference from visible \acp{HAP} on all other rings.
\end{enumerate}

\subsection{Coverage Probability under \ac{SCR-PCP}}
Conditioned on a finite serving distance $d_{0,{\rm P}}=x$, the \ac{SIR} at the typical user is
\begin{align*}
    \mathrm{SIR}_{\rm P}
    =
    \frac{G_t h_0 x^{-\eta}}
    {I_{\rm s,P}+I_{\rm o,P}}.
\end{align*}
The following theorem characterizes the coverage probability under nearest-visible \ac{HAP} association.
\begin{theorem}
\label{thm:coverage_pcp}
Consider the \ac{SCR-PCP} network constructed in Section~\ref{sec:sys_model}. Let $s_x\triangleq\tau x^\eta/G_t$. For $\tau>0$, the coverage probability is
\begin{align}
    \hspace*{-0.4cm} P_{\rm c,P}(\tau) = \int_{D_{\min}}^{\bar d_{\rm H}} \overline{\mathcal{L}}_{I_{\rm s,P}}(s_x\mid x) \mathcal{L}_{I_{\rm o,P}}(s_x\mid x) f_{d_{0,{\rm P}}}(x)\,{\rm d}x,
    \label{eq:Pc_main_pcp}
\end{align}
where
\begin{align}
    &\overline{\mathcal{L}}_{I_{\rm s,P}}(s\mid x) = \int_{\phi_1(x)}^{\phi_2(x)} \mathcal{L}_{I_{\rm s,P}}(s\mid x,\phi) f_{\phi_{\rm s}\mid d_{0,{\rm P}}}(\phi\mid x) \,{\rm d}\phi
    \label{eq:avg_laplace_same_cov_pcp}
\end{align}
is the serving-ring Laplace transform averaged over the conditional serving-ring angle distribution in \eqref{eq:conditional_serving_angle_pcp}. For a serving ring at central angle $\phi$, the conditional Laplace transform is
\begin{align}
    &\mathcal{L}_{I_{\rm s,P}}(s\mid x,\phi) \nonumber\\ 
    &\hspace*{-0.2cm} = \exp\left( -2a\nu  \int_{\alpha_0(x,\phi)}^{\alpha_{\rm v}(\phi)} \left[ 1-\frac{1} {1+sG_I d(\phi,\alpha)^{-\eta}} \right]{\rm d}\alpha \right).
    \label{eq:laplace_same_cov_pcp}
\end{align}
The conditional Laplace transform of the interference from all other rings is
\begin{align}
    &\mathcal{L}_{I_{\rm o,P}}(s\mid x) = \exp\bigg( -2\pi\lambda_{\rm u}R_{\rm H}^2 \nonumber\\ 
    &\quad\times \int_0^\pi p_{{\rm void},{\rm P}}(\phi,x) \left[ 1-\mathcal{L}_{I_{\rm s,P}}(s\mid x,\phi) \right] \sin\phi\,{\rm d}\phi \bigg),
    \label{eq:laplace_other_cov_pcp}
\end{align}
where
$p_{{\rm void},{\rm P}}(\phi,x) =\exp[-\nu L(\phi,x)]$ is the ring-wise void probability. The distance function $d(\phi,\alpha)$ and the angular boundaries $\alpha_0(x,\phi)$ and $\alpha_{\rm v}(\phi)$ are defined in subsection~\ref{subsec:IDRG}.
\end{theorem}
\begin{IEEEproof}
We condition on a finite serving distance $d_{0,{\rm P}}=x$ and serving-ring angle $\phi_{\rm s}=\phi$, with $D_{\min}<x<\bar d_{\rm H}$. Using $h_0\sim\mathrm{Exp}(1)$, independent of the interference, we obtain
\begin{align}
    &\mathbb{P} \left( \mathrm{SIR}_{\rm P}>\tau \mid x,\phi \right) = \mathbb{P} \left( h_0>s_x(I_{\rm s,P}+I_{\rm o,P}) \mid x,\phi \right) \nonumber\\ 
    &= \mathbb{E} \left[ e^{-s_x(I_{\rm s,P}+I_{\rm o,P})} \mid x,\phi \right] \nonumber\\ 
    &= \mathcal{L}_{I_{\rm s,P}}(s_x\mid x,\phi) \mathcal{L}_{I_{\rm o,P}}(s_x\mid x).
    \label{eq:pc_conditional_factor_pcp}
\end{align}
The last step follows because the conditional distribution of the other-ring interference depends on $x$ but not on $\phi$.

\emph{1) Same-ring interference.}
Consider the serving ring at central angle $\phi$. Since the user associates with its nearest-visible \ac{HAP}, the exclusion arc $|\alpha|<\alpha_0(x,\phi)$ contains no other \ac{HAP}. Hence, potential interferers must satisfy $|\alpha|\geq\alpha_0(x,\phi)$. Earth blockage further restricts them to the visible portion $|\alpha|\leq\alpha_{\rm v}(\phi)$. Combining these conditions, the residual visible arcs containing potential interferers satisfy 
\begin{align*} 
    \alpha_0(x,\phi) \leq|\alpha| \leq\alpha_{\rm v}(\phi). 
\end{align*}
Under the reduced Palm distribution of the intra-ring \ac{PPP}, removing the serving \ac{HAP} leaves a homogeneous \ac{PPP} of linear intensity $\nu$. Moreover, the independent-increments property ensures that conditioning on the exclusion arc being empty does not alter the point process on the residual arcs. For a visible interferer at distance $d(\phi,\alpha)$, the Rayleigh fading assumption gives $\mathbb{E}_h \left[ e^{-sG_Ih d(\phi,\alpha)^{-\eta}} \right]= \frac{1} {1+sG_I d(\phi,\alpha)^{-\eta}}.$ Applying the \ac{PGFL} of the residual one-dimensional \ac{PPP}, whose arc-length element is $a\,{\rm d}\alpha$, yields
\begin{align*}
    &\mathcal{L}_{I_{\rm s,P}}(s\mid x,\phi) = \mathbb{E} \left[ \prod_{\substack{ j:\ {\rm visible}\\ j\neq0 }} \frac{1}{ 1+sG_I d(\phi,\alpha_j)^{-\eta} } \,\middle|\, x,\phi \right] \\
    &=\exp\left(\!\! -a\nu \! \int_{\substack{ \alpha\in[-\pi,\pi):\\ \!\!\!\! \alpha_0(x,\phi) \leq |\alpha| \leq\alpha_{\rm v}(\phi)}}\!\! \left[ 1- \frac{1}{ 1+sG_I d(\phi,\alpha)^{-\eta} } \right] {\rm d}\alpha \right) \\
    &= \exp\left( -2a\nu \int_{\alpha_0(x,\phi)}^{\alpha_{\rm v}(\phi)} \left[ 1- \frac{1}{ 1+sG_I d(\phi,\alpha)^{-\eta} } \right] {\rm d}\alpha \right),
\end{align*}
Since $d(\phi,\alpha)$ is symmetric in $\alpha$, the contributions from the two residual arcs are equal. This gives \eqref{eq:laplace_same_cov_pcp}.

\emph{2) Interference from other rings.}
Under the reduced Palm distribution of the marked parent \ac{PPP}, removing the serving anchor leaves an independent homogeneous \ac{PPP} of anchors with their original intra-ring point processes, before imposing the nearest-association exclusion. For a remaining anchor $u$, let $\mathcal{E}_{u,x}$ denote the event that its ring contains no \ac{HAP} within distance $x$, and let $I_u$ denote the interference from its visible \acp{HAP}. For an anchor at central angle $\varphi$, the ring-wise exclusion probability is
\begin{align*}
    \mathbb{P}(\mathcal{E}_{u,x}\mid\varphi) = p_{{\rm void},{\rm P}}(\varphi,x) = e^{-\nu L(\varphi,x)}.
\end{align*}
By independence of Poisson counts on disjoint arcs, conditioning on $\mathcal{E}_{u,x}$ leaves a homogeneous \ac{PPP} of intensity $\nu$ on the residual arcs. Its visible interference therefore has the same conditional Laplace transform as the residual serving-ring process in \eqref{eq:laplace_same_cov_pcp}. Consequently,
\begin{align*}
    \mathbb{E}\left[ e^{-sI_u} \mathbf{1}\{\mathcal{E}_{u,x}\} \mid\varphi \right]&= \mathbb{P}(\mathcal{E}_{u,x}\mid\varphi) \mathbb{E}\left[ e^{-sI_u} \mid\mathcal{E}_{u,x},\varphi \right]\\ 
    &= p_{{\rm void},{\rm P}}(\varphi,x) \mathcal{L}_{I_{\rm s,P}}(s\mid x,\varphi).
\end{align*}
Let $\mathcal{E}_x$ denote the event that $\mathcal{E}_{u,x}$ holds for every remaining anchor. Since every \ac{HAP} within distance $x<\bar d_{\rm H}$ is visible, nearest-visible \ac{HAP} association imposes this exclusion condition on all other rings precisely. Hence, the conditional Laplace transform of their aggregate interference is
\begin{align}
    \mathcal{L}_{I_{\rm o,P}}(s\mid x) = \frac{ \mathbb{E}\left[ e^{-sI_{\rm o,P}} \mathbf{1}\{\mathcal{E}_x\} \right] }{ \mathbb{P}(\mathcal{E}_x) },
    \label{eq:conditional_other_lt_ratio_pcp}
\end{align}
where the expectation and probability are evaluated under the reduced Palm distribution of the marked parent process. By independence of the intra-ring processes, the joint exclusion-interference expectation factors over the remaining anchors. Applying the \ac{PGFL} of the parent \ac{PPP} gives
\begin{align}
    &\mathbb{E}\left[ e^{-sI_{\rm o,P}} \mathbf{1}\{\mathcal{E}_x\} \right] = \exp\bigg( -2\pi\lambda_{\rm u}R_{\rm H}^2 \nonumber\\ &\quad\times \int_0^\pi \big[ 1-p_{{\rm void},{\rm P}}(\varphi,x) \mathcal{L}_{I_{\rm s,P}}(s\mid x,\varphi) \big] \sin\varphi\,{\rm d}\varphi \bigg).
    \label{eq:joint_void_interference_pcp}
\end{align}
Similarly, applying the parent \ac{PGFL} to the ring-wise exclusion probabilities yields
\begin{align*}
    &\mathbb{P}(\mathcal{E}_x) = \exp\left( -2\pi\lambda_{\rm u}R_{\rm H}^2 \int_0^\pi \big[ 1-p_{{\rm void},{\rm P}}(\varphi,x) \big] \sin\varphi\,{\rm d}\varphi \right).
\end{align*}
Substituting these expressions into
\eqref{eq:conditional_other_lt_ratio_pcp}, we obtain
\begin{align*}
    &\mathcal{L}_{I_{\rm o,P}}(s\mid x) \\ 
    &= \exp\Bigg( -2\pi\lambda_{\rm u}R_{\rm H}^2 \int_0^\pi \Big[ 1-p_{{\rm void},{\rm P}}(\varphi,x) \mathcal{L}_{I_{\rm s,P}}(s\mid x,\varphi) \\ 
    &\hspace{3.5cm} -1+p_{{\rm void},{\rm P}}(\varphi,x) \Big] \sin\varphi\,{\rm d}\varphi \Bigg) \\ 
    &= \exp\Bigg( -2\pi\lambda_{\rm u}R_{\rm H}^2 \int_0^\pi p_{{\rm void},{\rm P}}(\varphi,x) \\ 
    &\hspace{1.8cm}\times \left[ 1-\mathcal{L}_{I_{\rm s,P}}(s\mid x,\varphi) \right] \sin\varphi\,{\rm d}\varphi \Bigg),
\end{align*}
which establishes \eqref{eq:laplace_other_cov_pcp}. The integration extends over all possible anchor angles $\varphi\in[0,\pi]$. Rings without a residual visible arc have a conditional single-ring Laplace factor equal to one and therefore contribute nothing to the interference exponent.

\emph{3) Coverage probability.}
Having established the conditional Laplace transforms of the serving-ring and other-ring interference, we return to \eqref{eq:pc_conditional_factor_pcp}. Conditioned on the finite serving distance $d_{0,{\rm P}}=x$, the serving-ring angle $\phi_{\rm s}$ has the conditional density derived in Corollary~\ref{cor:serving_angle_pcp}. Therefore, averaging over $\phi_{\rm s}$ gives
\begin{align*}
    &\mathbb{P}\left( \mathrm{SIR}_{\rm P}>\tau \mid d_{0,{\rm P}}=x \right)\\ 
    &= \int_{\phi_1(x)}^{\phi_2(x)} \mathbb{P}\left( \mathrm{SIR}_{\rm P}>\tau \mid x,\phi \right) f_{\phi_{\rm s}\mid d_{0,{\rm P}}} (\phi\mid x)\,{\rm d}\phi\\ 
    &= \int_{\phi_1(x)}^{\phi_2(x)} \mathcal{L}_{I_{\rm s,P}}(s_x\mid x,\phi) \mathcal{L}_{I_{\rm o,P}}(s_x\mid x) f_{\phi_{\rm s}\mid d_{0,{\rm P}}} (\phi\mid x)\,{\rm d}\phi\\ 
    &= \mathcal{L}_{I_{\rm o,P}}(s_x\mid x) \int_{\phi_1(x)}^{\phi_2(x)} \mathcal{L}_{I_{\rm s,P}}(s_x\mid x,\phi) f_{\phi_{\rm s}\mid d_{0,{\rm P}}} (\phi\mid x)\,{\rm d}\phi\\ 
    &= \overline{\mathcal{L}}_{I_{\rm s,P}}(s_x\mid x) \mathcal{L}_{I_{\rm o,P}}(s_x\mid x),
\end{align*}
where the last equality follows from \eqref{eq:avg_laplace_same_cov_pcp}. Finally, deconditioning over the finite serving-distance distribution yields
\begin{align*}
    P_{\rm c,P}(\tau) &= \int_{D_{\min}}^{\bar d_{\rm H}} \mathbb{P}\left( \mathrm{SIR}_{\rm P}>\tau \mid d_{0,{\rm P}}=x \right) f_{d_{0,{\rm P}}}(x)\,{\rm d}x\\ 
    &= \int_{D_{\min}}^{\bar d_{\rm H}} \overline{\mathcal{L}}_{I_{\rm s,P}}(s_x\mid x) \mathcal{L}_{I_{\rm o,P}}(s_x\mid x) f_{d_{0,{\rm P}}}(x)\,{\rm d}x,
\end{align*}
which establishes \eqref{eq:Pc_main_pcp}. Since $f_{d_{0,{\rm P}}}(x)$ is the density of the finite component of the serving-distance distribution, the event $d_{0,{\rm P}}=\infty$ contributes zero to the coverage probability and is automatically counted as outage.
\end{IEEEproof}

\subsection{Coverage Probability under \ac{SCR-BCP}}
As in the \ac{SCR-PCP} case, the coverage probability under \ac{SCR-BCP} is obtained by conditioning on the serving distance and serving-ring angle. The key distinction is that each ring contains exactly $n_{\rm H}$ \acp{HAP}. Consequently, the serving ring has $n_{\rm H}-1$ residual platforms, whereas each non-serving ring retains $n_{\rm H}$ platforms subject to the serving-distance exclusion. Conditioned on a finite serving distance $d_{0,{\rm B}}=x$, the \ac{SIR} is
\begin{align*}
    \mathrm{SIR}_{\rm B} = \frac{G_t h_0 x^{-\eta}} {I_{\rm s,B}+I_{\rm o,B}}.
\end{align*}
The following theorem characterizes the corresponding coverage probability.

\begin{theorem}
\label{thm:coverage_bcp}
Consider the \ac{SCR-BCP} network constructed in Section~\ref{sec:sys_model}. Let $s_x\triangleq\tau x^\eta/G_t$. For $\tau>0$, the coverage probability is
\begin{align}
    P_{\rm c,B}(\tau) = \int_{D_{\min}}^{\bar d_{\rm H}} \overline{\mathcal{L}}_{I_{\rm s,B}}(s_x\mid x) \mathcal{L}_{I_{\rm o,B}}(s_x\mid x) f_{d_{0,{\rm B}}}(x)\,{\rm d}x,
    \label{eq:Pc_bcp_final}
\end{align}
where $f_{d_{0,{\rm B}}}(x)$ is the finite serving-distance density, and
\begin{align}
    &\overline{\mathcal{L}}_{I_{\rm s,B}}(s\mid x) = \int_{\phi_1(x)}^{\phi_2(x)} \mathcal{L}_{I_{\rm s,B}}(s\mid x,\phi) f_{\phi_{\rm s}\mid d_{0,{\rm B}}}(\phi\mid x) \,{\rm d}\phi
    \label{eq:avg_LIs_bcp}
\end{align}
is the serving-ring Laplace transform averaged over the conditional serving-ring angle distribution in \eqref{eq:conditional_serving_angle_bcp}. The conditional serving-ring Laplace transform is
\begin{align}
    \mathcal{L}_{I_{\rm s,B}}(s\mid x,\phi) = \left( \mathcal{Q}_{\rm H}(s\mid x,\phi) \right)^{n_{\rm H}-1},
    \label{eq:LIs_bcp_final}
\end{align}
where, for a ring at central angle $\phi$ with $\alpha_0(x,\phi)<\pi$,
\begin{align}
    &\mathcal{Q}_{\rm H}(s\mid x,\phi) = \frac{ \displaystyle \pi-\alpha_{\rm v}(\phi) +\int_{\alpha_0(x,\phi)}^{\alpha_{\rm v}(\phi)} \frac{{\rm d}\alpha} {1+sG_I d(\phi,\alpha)^{-\eta}} }{ \pi-\alpha_0(x,\phi) },
    \label{eq:Q_single_bcp}
\end{align}
which is the conditional Laplace factor of a residual \ac{HAP} given that its distance exceeds $x$. The term $\pi-\alpha_{\rm v}(\phi)$ accounts for Earth-blocked locations, which contribute no interference. The conditional Laplace transform of the interference from all other rings is
\begin{align}
    &\mathcal{L}_{I_{\rm o,B}}(s\mid x) = \exp\bigg( -2\pi\lambda_{\rm u}R_{\rm H}^2 \nonumber\\ 
    &\times \int_0^\pi p_{{\rm void},{\rm B}}(\varphi,x) \left[ 1-\left( \mathcal{Q}_{\rm H}(s\mid x,\varphi) \right)^{n_{\rm H}} \right] \sin\varphi\,{\rm d}\varphi \bigg),
    \label{eq:LIo_bcp_final}
\end{align}
where $p_{{\rm void},{\rm B}}(\varphi,x) =\left(1-L(\varphi,x)/(2\pi a)\right)^{n_{\rm H}}$ is the ring-wise void probability. If $\alpha_0(x,\varphi)=\pi$, this void probability is zero and the corresponding integrand in \eqref{eq:LIo_bcp_final} becomes zero. 
\end{theorem}

\begin{IEEEproof}
We condition on a finite serving distance $d_{0,{\rm B}}=x$ and serving-ring angle $\phi_{\rm s}=\phi$, with $D_{\min}<x<\bar d_{\rm H}$. Using $h_0\sim\mathrm{Exp}(1)$, independent of the interference, we obtain
\begin{align}
    &\mathbb{P} \left( \mathrm{SIR}_{\rm B}>\tau\mid x,\phi \right) = \mathbb{P} \left( h_0>s_x(I_{\rm s,B}+I_{\rm o,B}) \mid x,\phi \right) \nonumber\\ 
    &= \mathbb{E} \left[ e^{-s_x(I_{\rm s,B}+I_{\rm o,B})} \mid x,\phi \right] \nonumber\\ 
    &= \mathcal{L}_{I_{\rm s,B}}(s_x\mid x,\phi) \mathcal{L}_{I_{\rm o,B}}(s_x\mid x).
    \label{eq:cond_success_bcp_cov}
\end{align}
The last equality follows because, conditioned on $(x,\phi)$, the residual serving-ring configuration and the other-ring process are independent. Moreover, the conditional distribution of the other-ring interference depends on $x$ but not on $\phi$.

\emph{1) Same-ring interference.}
Conditioned on $(x,\phi)$, one of the $n_{\rm H}$ platforms on the serving ring is located at distance $x$ and serves the typical user. By nearest-visible \ac{HAP} association, the other $n_{\rm H}-1$ platforms must lie outside the exclusion arc $|\alpha|<\alpha_0(x,\phi)$. Since the original platform locations are i.i.d.\ uniform on the ring, exchangeability implies that the remaining $n_{\rm H}-1$ angular locations are conditionally independent and uniformly distributed over the residual safe arc. For a uniformly distributed platform angle $\alpha$, let $D=d(\phi,\alpha)$ denote its distance to the user. The conditional angular density on $\alpha\in[-\pi,\pi)$ is
\begin{align}
    f_{\alpha\mid D>x}(\alpha\mid x,\phi) &= \frac{ \mathbf{1}\{d(\phi,\alpha)>x\} }{ \displaystyle \int_{-\pi}^{\pi} \mathbf{1}\{d(\phi,\beta)>x\} \,{\rm d}\beta } \nonumber\\ 
    &= \frac{ \mathbf{1}\{ \alpha_0(x,\phi)\leq|\alpha|\leq\pi \} }{ 2(\pi-\alpha_0(x,\phi)) },
    \label{eq:trunc_density_alpha_bcp_alpha0}
\end{align}
for $\alpha_0(x,\phi)<\pi$. A residual platform contributes interference only if it is also visible, i.e., $|\alpha|\leq\alpha_{\rm v}(\phi)$. Therefore, its conditional Laplace factor, after averaging over Rayleigh fading and the truncated angular distribution, is
\begin{align*}
    &\mathbb{E}_{\alpha,h} \bigg[ e^{-sG_Ih d(\phi,\alpha)^{-\eta} \mathbf{1}\{d(\phi,\alpha)\leq\bar d_{\rm H}\}} \,\bigg|\,D>x,\phi \bigg] \\
    &= \frac{1}{2(\pi-\alpha_0(x,\phi))} \bigg[ \int_{\substack{\alpha\in[-\pi,\pi):\\ \alpha_0(x,\phi)\leq|\alpha| \leq\alpha_{\rm v}(\phi)}} \frac{{\rm d}\alpha} {1+sG_I d(\phi,\alpha)^{-\eta}} \\
    &\hspace*{4cm} + \int_{\substack{\alpha\in[-\pi,\pi):\\ \alpha_{\rm v}(\phi)<|\alpha|\leq\pi}} 1\,{\rm d}\alpha \bigg] \\ 
    &= \frac{1}{\pi-\alpha_0(x,\phi)} \bigg[ \int_{\alpha_0(x,\phi)}^{\alpha_{\rm v}(\phi)} \frac{{\rm d}\alpha} {1+sG_I d(\phi,\alpha)^{-\eta}}  + \int_{\alpha_{\rm v}(\phi)}^\pi 1\,{\rm d}\alpha \bigg] \\ 
    &= \frac{ \displaystyle \pi-\alpha_{\rm v}(\phi) +\int_{\alpha_0(x,\phi)}^{\alpha_{\rm v}(\phi)} \frac{{\rm d}\alpha} {1+sG_I d(\phi,\alpha)^{-\eta}} }{ \pi-\alpha_0(x,\phi) } = \mathcal{Q}_{\rm H}(s\mid x,\phi).
\end{align*}
Here, the second integral accounts for Earth-blocked platform locations, whose Laplace factor equals one. Since the $n_{\rm H}-1$ residual platforms are conditionally independent, their aggregate interference transform is
\begin{align*}
    &\mathcal{L}_{I_{\rm s,B}}(s\mid x,\phi) = \mathbb{E} \left[ \prod_{j=1}^{n_{\rm H}-1} e^{-sI_j} \,\middle|\, x,\phi \right] \\ 
    &= \prod_{j=1}^{n_{\rm H}-1} \mathbb{E} \left[ e^{-sI_j} \mid x,\phi \right] = \left[ \mathcal{Q}_{\rm H}(s\mid x,\phi) \right]^{n_{\rm H}-1},
\end{align*}
where $I_j$ denotes the visible interference contribution of a residual platform. This establishes \eqref{eq:LIs_bcp_final}.

\emph{2) Interference from other rings.}
Under the reduced Palm distribution of the marked parent \ac{PPP}, removing the serving anchor leaves an independent homogeneous \ac{PPP} of anchors with their original intra-ring configurations, before imposing the nearest-association exclusion. For a remaining anchor $u$, let $\mathcal{E}_{u,x}$ denote the event that all $n_{\rm H}$ platforms on its ring lie farther than $x$ from the user, and let $I_u$ denote their aggregate visible interference. For an anchor at central angle $\varphi$, the ring-wise exclusion probability is
\begin{align*}
    &\mathbb{P} \left( \mathcal{E}_{u,x}\mid\varphi \right) = \prod_{j=1}^{n_{\rm H}} \mathbb{P} \left( d(\varphi,\alpha_j)>x \mid\varphi \right) \\ &\quad= \left( 1-\frac{L(\varphi,x)}{2\pi a} \right)^{n_{\rm H}} = p_{{\rm void},{\rm B}}(\varphi,x).
\end{align*}
Because the $n_{\rm H}$ platform locations are independent, conditioning on $\mathcal{E}_{u,x}$ leaves all of them independently distributed according to the truncated angular density in \eqref{eq:trunc_density_alpha_bcp_alpha0}. Their individual visible-interference Laplace factor is therefore $\mathcal{Q}_{\rm H}(s\mid x,\varphi)$. Consequently,
\begin{align*}
    \mathbb{E} \left[ e^{-sI_u} \mathbf{1}\{\mathcal{E}_{u,x}\} \mid\varphi \right] &= \mathbb{P}(\mathcal{E}_{u,x}\mid\varphi) \mathbb{E} \left[ e^{-sI_u} \mid\mathcal{E}_{u,x},\varphi \right] \\ 
    &= p_{{\rm void},{\rm B}}(\varphi,x) \left( \mathcal{Q}_{\rm H}(s\mid x,\varphi) \right)^{n_{\rm H}}.
\end{align*}
Let $\mathcal{E}_x$ denote the event that $\mathcal{E}_{u,x}$ holds for every remaining anchor. Since every \ac{HAP} closer than $x<\bar d_{\rm H}$ is visible, nearest-visible \ac{HAP} association imposes this exclusion condition on all other rings. Hence, the conditional Laplace transform of their aggregate interference is
\begin{align}
    \mathcal{L}_{I_{\rm o,B}}(s\mid x) = \frac{ \mathbb{E} \left[ e^{-sI_{\rm o,B}} \mathbf{1}\{\mathcal{E}_x\} \right] }{ \mathbb{P}(\mathcal{E}_x) },
    \label{eq:conditional_other_lt_ratio_bcp}
\end{align}
where both quantities on the right are evaluated under the reduced Palm distribution of the marked parent process. Independence of the ring configurations and the parent \ac{PGFL} gives
\begin{align}
    &\mathbb{E} \left[ e^{-sI_{\rm o,B}} \mathbf{1}\{\mathcal{E}_x\} \right] = \exp\bigg( -2\pi\lambda_{\rm u}R_{\rm H}^2 \nonumber\\
    & \times \int_0^\pi \bigg[ 1- p_{{\rm void},{\rm B}}(\varphi,x) \left( \mathcal{Q}_{\rm H}(s\mid x,\varphi) \right)^{n_{\rm H}} \bigg] \sin\varphi\,{\rm d}\varphi \bigg).
    \label{eq:joint_void_interference_bcp}
\end{align}
Similarly, applying the parent \ac{PGFL} to the ring-wise exclusion probabilities yields
\begin{align*}
    \mathbb{P}(\mathcal{E}_x) = \exp\left( -2\pi\lambda_{\rm u}R_{\rm H}^2 \int_0^\pi \left[ 1-p_{{\rm void},{\rm B}}(\varphi,x) \right] \sin\varphi\,{\rm d}\varphi \right).
\end{align*}
Substituting these expressions into \eqref{eq:conditional_other_lt_ratio_bcp}, we obtain
\begin{align*}
    &\mathcal{L}_{I_{\rm o,B}}(s\mid x) \\ 
    &= \exp\Bigg(\!\!-2\pi\lambda_{\rm u}R_{\rm H}^2 \int_0^\pi \Bigg[ 1- p_{{\rm void},{\rm B}}(\varphi,x) \left( \mathcal{Q}_{\rm H}(s\mid x,\varphi) \right)^{n_{\rm H}} \\ 
    &\hspace{3cm} -1+p_{{\rm void},{\rm B}}(\varphi,x) \Bigg] \sin\varphi\,{\rm d}\varphi \Bigg) \\ 
    &= \exp\Bigg(\!\! -2\pi\lambda_{\rm u}R_{\rm H}^2 \int_0^\pi p_{{\rm void},{\rm B}}(\varphi,x) \\ 
    &\hspace{1.6cm}\times \left[ 1- \left( \mathcal{Q}_{\rm H}(s\mid x,\varphi) \right)^{n_{\rm H}} \right] \sin\varphi\,{\rm d}\varphi \Bigg),
\end{align*}
which establishes \eqref{eq:LIo_bcp_final}. The integration extends over all possible anchor angles $\varphi\in[0,\pi]$. Rings with no residual visible portion contribute nothing to the interference exponent, either because their conditional Laplace factor equals one or because their ring-wise void probability vanishes.

\emph{3) Coverage probability.}
Having established the conditional Laplace transforms of the serving-ring and other-ring interference, we return to \eqref{eq:cond_success_bcp_cov}. Conditioned on the finite serving distance $d_{0,{\rm B}}=x$, the serving-ring angle has the conditional density derived in Corollary~\ref{cor:serving_angle_bcp}. Averaging over $\phi_{\rm s}$ gives
\begin{align*}
    &\mathbb{P} \left( \mathrm{SIR}_{\rm B}>\tau \mid d_{0,{\rm B}}=x \right) \\ 
    &= \int_{\phi_1(x)}^{\phi_2(x)} \mathbb{P} \left( \mathrm{SIR}_{\rm B}>\tau \mid x,\phi \right) f_{\phi_{\rm s}\mid d_{0,{\rm B}}} (\phi\mid x)\,{\rm d}\phi \\ 
    &= \int_{\phi_1(x)}^{\phi_2(x)} \mathcal{L}_{I_{\rm s,B}}(s_x\mid x,\phi) \mathcal{L}_{I_{\rm o,B}}(s_x\mid x) f_{\phi_{\rm s}\mid d_{0,{\rm B}}} (\phi\mid x)\,{\rm d}\phi \\ 
    &= \mathcal{L}_{I_{\rm o,B}}(s_x\mid x) \int_{\phi_1(x)}^{\phi_2(x)} \mathcal{L}_{I_{\rm s,B}}(s_x\mid x,\phi) f_{\phi_{\rm s}\mid d_{0,{\rm B}}} (\phi\mid x)\,{\rm d}\phi \\ 
    &= \overline{\mathcal{L}}_{I_{\rm s,B}}(s_x\mid x) \mathcal{L}_{I_{\rm o,B}}(s_x\mid x),
\end{align*}
where the last equality follows from \eqref{eq:avg_LIs_bcp}. Finally, deconditioning over the finite serving-distance distribution yields
\begin{align*}
    P_{\rm c,B}(\tau) = \int_{D_{\min}}^{\bar d_{\rm H}} \overline{\mathcal{L}}_{I_{\rm s,B}}(s_x\mid x) \mathcal{L}_{I_{\rm o,B}}(s_x\mid x) f_{d_{0,{\rm B}}}(x)\,{\rm d}x,
\end{align*}
which establishes \eqref{eq:Pc_bcp_final}. Since $f_{d_{0,{\rm B}}}(x)$ is the density of the finite component of the serving-distance distribution, the event $d_{0,{\rm B}}=\infty$ contributes zero to the coverage probability and is automatically counted as outage.
\end{IEEEproof}

\begin{remark}
\label{rem:general_propagation}
The Rayleigh-fading and fixed-gain assumptions can be relaxed through independent link marking while retaining nearest-visible \ac{HAP} association. For a visible interferer at distance $y$, let $p_q(y)$ denote the probability of propagation state $q\in\{{\rm L},{\rm N}\}$, with normalized path loss $C_qy^{-\eta_q}$ and fading $h_q\sim\Gamma(m_q,1/m_q)$ \cite{bai2014analysis,bai2014coverage}. Under a sectored antenna model, the interfering gain $G$ equals $G_t$ with probability $p_{\rm M}=\omega_{\rm B}/(2\pi)$ and $G_I$ otherwise, where $\omega_{\rm B}$ is the main-lobe beamwidth. Assuming independent marks conditioned on the geometry, the single-link Laplace factor becomes
\begin{align}
    \mathcal{M}(s,y) =\!\!\! \sum_{q\in\{{\rm L},{\rm N}\}}p_q(y) \mathbb{E}_G\left[ \left(1+\frac{sGC_qy^{-\eta_q}}{m_q}\right)^{-m_q} \right],
    \label{eq:general_link_factor}
\end{align}
where $C_q$ is normalized by the reference coefficient $K$. Replacing the original Rayleigh interference factor by $\mathcal{M}(s,d(\phi,\alpha))$ in \ac{SCR-PCP} gives the residual-ring interference transform $\exp\left( -2a\nu \int_{\alpha_0(x,\phi)}^{\alpha_{\rm v}(\phi)} \left[1-\mathcal{M}(s,d(\phi,\alpha))\right] {\rm d}\alpha \right)$ and the conditional single-platform Laplace factor in \ac{SCR-BCP} as $\frac{\pi-\alpha_{\rm v}(\phi)}{ \pi-\alpha_0(x,\phi) } +\frac{\int_{\alpha_0(x,\phi)}^{\alpha_{\rm v}(\phi)} \mathcal{M}(s,d(\phi,\alpha))\,{\rm d}\alpha }{ \pi-\alpha_0(x,\phi) }$ for $\alpha_0(x,\phi)<\pi$. Substituting these factors into the corresponding coverage expressions preserves the parent-process constructions, ring-wise void probabilities, and geometric exclusion. Generalizing the serving-link fading requires a corresponding modification of the conditional success probability.
\end{remark}

\begin{remark}
\label{rem:max_power_association}
nearest-visible \ac{HAP} association is equivalent to maximum average received-power association under the homogeneous reference model. With heterogeneous \ac{LOS}/\ac{NLOS} propagation, the serving \ac{HAP} instead maximizes $G_x^{({\rm A})}C_{S_x}\|x-\mathbf{o}\|^{-\eta_{S_x}}$ over visible candidates, where $S_x\in\{{\rm L},{\rm N}\}$ denotes the propagation state of a visible candidate \ac{HAP} at location $x$ and $G_x^{({\rm A})}$ denotes its effective association gain. The serving \ac{HAP} is selected as $x_{0,k}^{\rm MP} =\arg\max_{x\in\Phi_{{\rm H},{\rm v}}^k} \left\{ G_x^{({\rm A})}C_{S_x} \|x-\mathbf{o}\|^{-\eta_{S_x}} \right\}.$ Such association is commonly adopted under heterogeneous \ac{LOS}/\ac{NLOS} propagation \cite{bai2014coverage}. 

Consider a serving \ac{HAP} at distance $x$ with propagation state $q_0$ and association gain $g_0$. A competing \ac{HAP} in state $q$ with gain $g$ provides higher average received power if its distance $y$ satisfies $y<\chi_{q,g\mid q_0,g_0}(x)$, where $\chi_{q,g\mid q_0,g_0}(x)= (gC_q/(g_0C_{q_0}))^{1/\eta_q} x^{\eta_{q_0}/\eta_q}$. Thus, the corresponding visibility-constrained exclusion half-angle becomes  
\begin{align} 
    \alpha_{0,q,g}^{\rm MP}(x,\phi\mid q_0,g_0) =\alpha_0\left( \min\{\chi_{q,g\mid q_0,g_0}(x),\bar d_{\rm H}\}, \phi\right). 
    \label{eq:max_power_exclusion_angle} 
\end{align}
Thus, the patrol geometry and point-process constructions remain unchanged, while association and conditional interference must account for the link marks. Under homogeneous propagation and gains, $\chi_{q,g\mid q_0,g_0}(x)=x$, recovering the nearest-visible \ac{HAP} association model.
\end{remark}

\subsection{Discussion: \ac{SCR-PCP} versus \ac{SCR-BCP}}
Theorems~\ref{thm:coverage_pcp} and~\ref{thm:coverage_bcp} highlight the fundamental distinction between the two spatial models. Both \ac{SCR-PCP} and \ac{SCR-BCP} share the same parent anchor \ac{PPP}, patrol-ring geometry, visibility constraint, and nearest-visible \ac{HAP} association rule; the difference arises solely through the intra-ring point process. In \ac{SCR-PCP}, each patrol ring carries a one-dimensional \ac{PPP} with intensity $\nu$, so the number of \acp{HAP} on a ring is Poisson distributed with mean $2\pi a\nu$. In contrast, each \ac{SCR-BCP} ring contains exactly $n_{\rm H}$ platforms. Despite this distinction, both coverage expressions have the same structure: the serving-ring Laplace transform is averaged over the conditional serving-ring angle, while the other-ring interference is characterized using the corresponding ring-wise void probability and interference transform. The key difference lies in the residual point configuration after association. Under \ac{SCR-PCP}, the residual platforms form a \ac{PPP} on the admissible arcs, yielding an exponential Laplace transform. Under \ac{SCR-BCP}, the serving ring retains $n_{\rm H}-1$ conditionally independent platforms on the safe arc, giving the finite-power term $\big(\mathcal{Q}_{\rm H}(s\mid x,\phi)\big)^{n_{\rm H}-1}$.

For an equivalent mean per-ring population, set $n_{\rm H}=2\pi a\nu$. The corresponding ring-wise void probabilities satisfy
\begin{align}
    p_{{\rm void},{\rm B}}(\phi,x) &=\left(1-\frac{L(\phi,x)}{2\pi a}\right)^{n_{\rm H}} \nonumber\\ &\leq \exp\left(-n_{\rm H}\frac{L(\phi,x)}{2\pi a}\right) =p_{{\rm void},{\rm P}}(\phi,x),
    \label{eq:void_comparison_pcp_bcp}
\end{align}
where the inequality follows from $(1-z)^{n_{\rm H}}\leq e^{-n_{\rm H}z}$ for $0\leq z\leq1$. Thus, an arbitrary ring arc is less likely to be empty under \ac{SCR-BCP}. This also implies $V_{\rm B}(x)\leq V_{\rm P}(x)$ for $D_{\min}\leq x\leq\bar d_{\rm H}$, so \ac{SCR-BCP} has a higher probability of providing a visible \ac{HAP} within a prescribed distance.

Nevertheless, this ordering does not imply a universal ordering of coverage probabilities. While a lower void probability increases the likelihood of a nearby serving \ac{HAP}, it also affects the interference distribution. Furthermore, the conditional residual populations differ between the two models. Consequently, their relative coverage depends on the deployment geometry, anchor intensity, visibility constraint, and \ac{SIR} threshold. Unlike \ac{SCR-PCP}, where the per-ring population is random, \ac{SCR-BCP} explicitly controls the number of platforms assigned to each patrol ring through $n_{\rm H}$. It therefore represents a distinct deployment regime that cannot be obtained merely by matching the mean population of the two models.

\section{Energy-Aware Patrol Optimization}
\label{sec:energy}
In practice, the patrol geometry affects both communication performance and propulsion power. We therefore augment the coverage analysis with an aerodynamic model for steady circular flight and investigate the energy-optimal patrol design.

\subsection{Aerodynamic Flight Model}
Each \ac{HAP} performs steady circular flight (SCF) along a patrol ring of radius $a=R_{\rm H}\tan\gamma$ at constant airspeed $V$. As established in the geometric model, the patrol ring lies on the flight shell of radius $R_{\rm F}=\sqrt{R_{\rm H}^2+a^2}$, maintaining a constant geometric altitude $H_{\rm F}=R_{\rm F}-R_\oplus$. For propulsion analysis, we adopt the local coordinated-turn approximation, treating the trajectory as a horizontal circle in the anchor's local frame. This approximation is appropriate for $a/R_{\rm H}\ll1$. Following \cite{zeng2017energy}, we distinguish parasite power from lift-induced power and apply the load-factor penalty only to the induced component. The following result characterizes the propulsion power under these assumptions.
\begin{lemma}
\label{lem:power}
Consider a fixed-wing \ac{HAP} flying at constant airspeed $V$ along a patrol ring of radius $a=R_{\rm H}\tan\gamma$. Under the local coordinated-turn approximation, the required propulsion power is
\begin{align}
    P_{\rm SCF}(\gamma,V) = c_1V^3+ \frac{c_2}{V} \left( 1+ \frac{V^4} {g^2R_{\rm H}^2\tan^2\gamma} \right),
    \label{eq:power_scf}
\end{align}
where $c_1,c_2>0$ are the parasite- and induced-power coefficients, respectively, and $g$ is the gravitational acceleration.
\end{lemma}
\begin{IEEEproof}
Following \cite{zeng2017energy}, the propulsion power for steady horizontal flight at airspeed $V$ is $P_{\rm SHF}(V) = c_1V^3+\frac{c_2}{V}.$ For a coordinated circular turn of radius $a$, the lift components satisfy $L\cos\zeta=mg,$ and $L\sin\zeta=\frac{mV^2}{a},$ where $L$ is the lift force, $m$ is the platform mass, and $\zeta$ is the bank angle. Consequently, the load factor satisfies 
\begin{align*}
    n^2=\frac{1}{\cos^2\zeta}=1+\frac{V^4}{g^2a^2}.
\end{align*}
Since the induced-power component scales with the square of the load factor while the parasite-power component remains unchanged at fixed airspeed, the circular-flight propulsion power becomes 
\begin{align*}
    P_{\rm SCF} = c_1V^3+n^2\frac{c_2}{V} = c_1V^3+ \frac{c_2}{V} \left( 1+\frac{V^4}{g^2a^2} \right).
\end{align*} 
Substituting $a=R_{\rm H}\tan\gamma$ completes the proof.
\end{IEEEproof}
\begin{figure*}[t]
\centering
\subfloat[]
{\includegraphics[width = 0.27\textwidth]{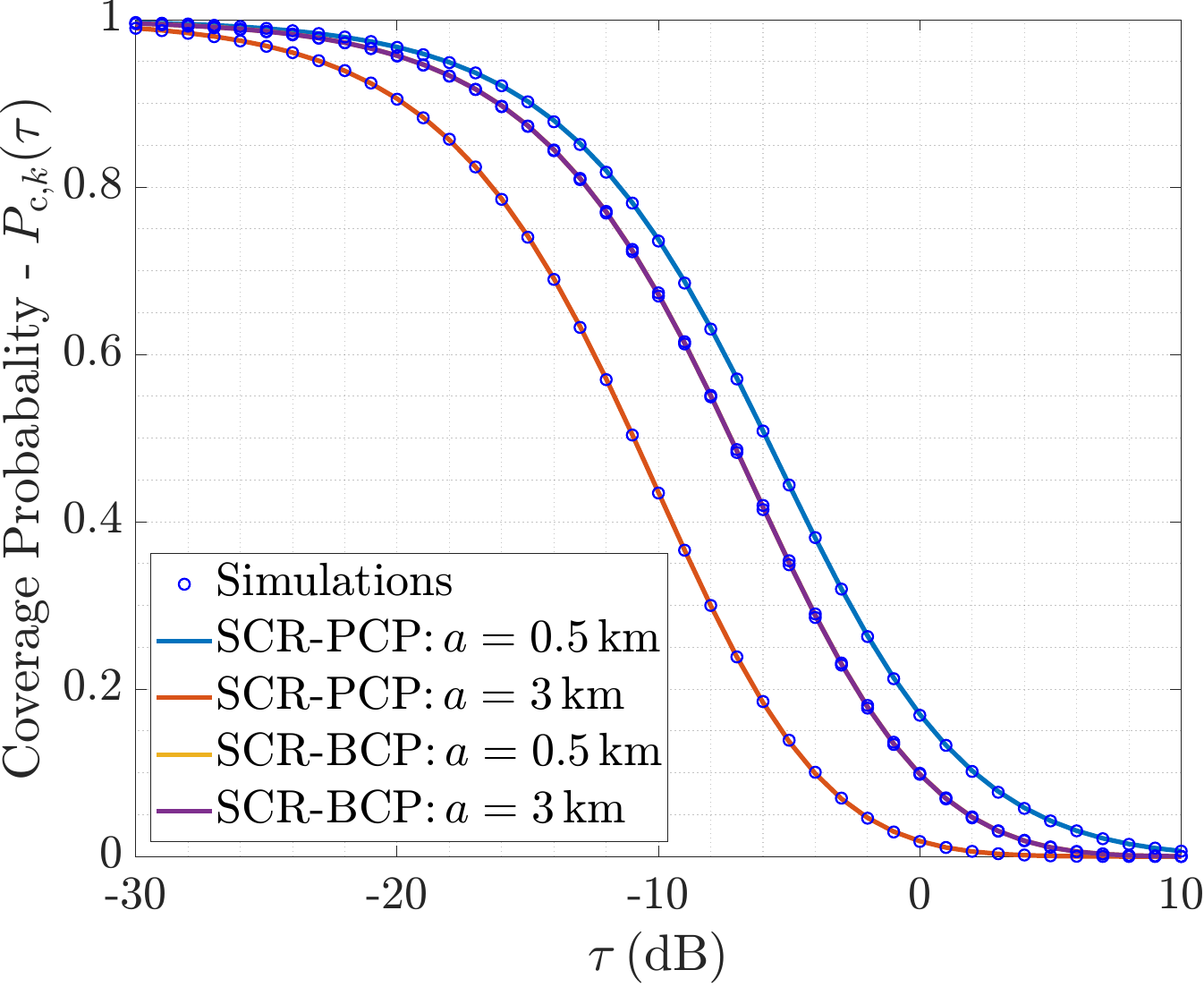}
\label{fig:n_tau}}
\hfil
\subfloat[]
{\includegraphics[width = 0.27\textwidth]{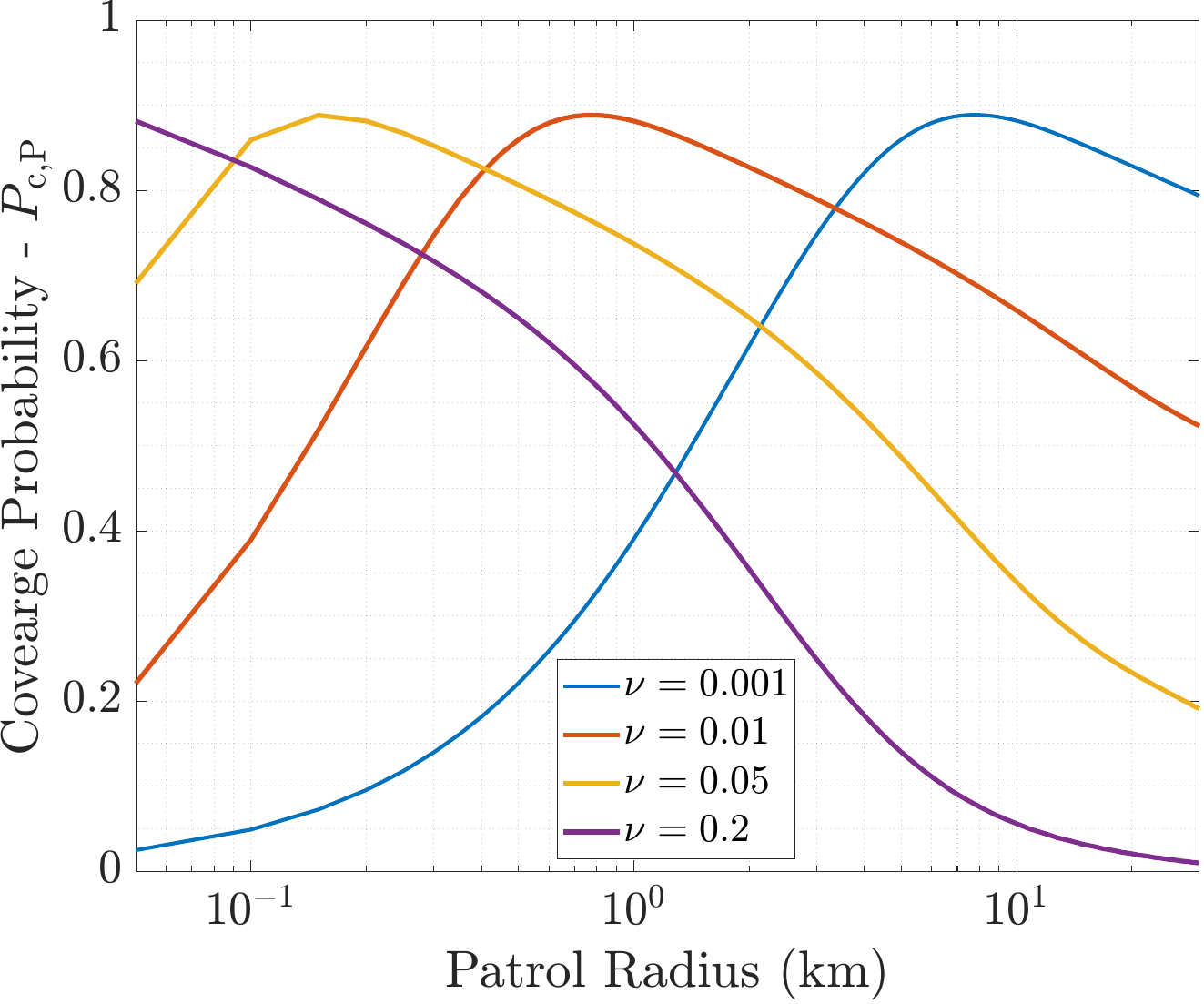}
\label{fig:result_1}}
\hfil
\subfloat[]
{\includegraphics[width = 0.27\textwidth]{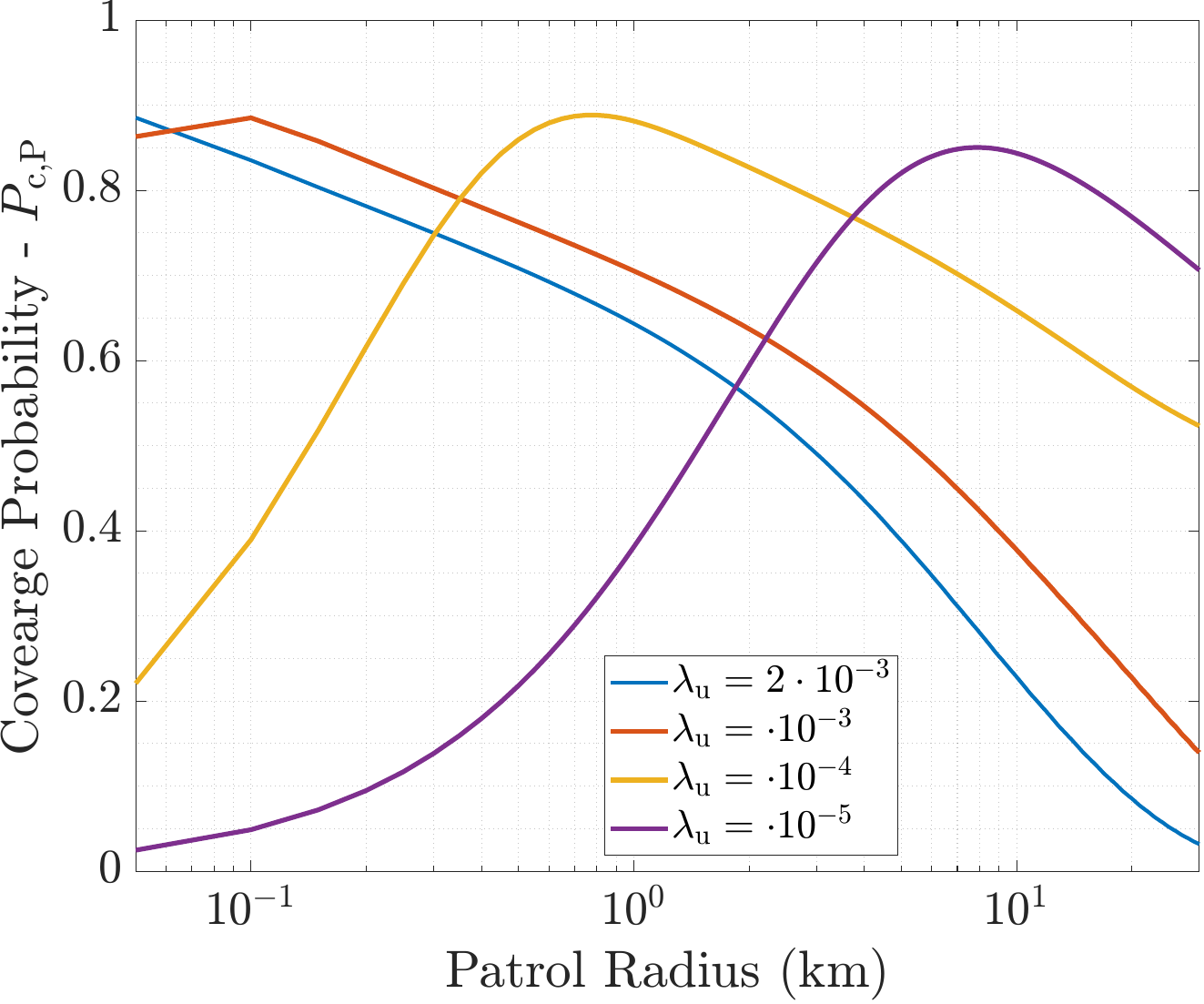}
\label{fig:result_2}}
\caption{(a) Coverage probability versus the target SIR threshold $\tau$ for varying patrol radii $a$ under \ac{SCR-PCP} and \ac{SCR-BCP}. (b) $P_{\rm c,P}$ versus the patrol-ring radius $a$ for different $\nu$ with fixed anchor intensity $\lambda_{\rm u}=10^{-6}$. (c) $P_{\rm c,P}$ versus the patrol-ring radius $a$ for different $\lambda_{\rm u}$ with fixed \ac{HAP} density $\nu=0.01$.}
\label{fig:r_1_2}
\end{figure*}

\subsection{Coverage Energy Efficiency}
We combine the coverage probability derived in Section~\ref{sec:coverage} with the SCF power model to define an energy-aware performance metric. Recall that $a=R_{\rm H}\tan\gamma$ and $R_{\rm F}=R_{\rm H}\sec\gamma$. Thus, varying $\gamma$ changes the patrol radius, platform altitude $H_{\rm F}(\gamma)$, and visibility distance $\bar d_{\rm H}(\gamma)$, all of which are accounted for in the coverage probability $P_{{\rm c},k}(\tau;\gamma)$. We consider a fixed airspeed $V>0$ and treat $c_1$, $c_2$, and $g$ as constants. All lengths used in the propulsion model are expressed in meters.
\begin{definition}
For a target SIR threshold $\tau$, model index $k\in\{{\rm P},{\rm B}\}$, and fixed airspeed $V$, the \ac{CEE} is defined as
\begin{align}
    \eta_{{\rm CEE}, k}(\gamma) = \frac{P_{{\rm c}, k}(\tau;\gamma)}{P_{\mathrm{SCF}}(\gamma)},
    \label{eq:cee_def}
\end{align}
where $P_{\rm SCF}(\gamma)$ is given by \eqref{eq:power_scf}.
\end{definition}
The \ac{CEE} measures coverage probability per unit propulsion power. Increasing $\gamma$ reduces the turning power penalty but also changes the serving-distance and interference distributions, potentially degrading coverage. The optimal patrol design therefore balances these competing effects. Since $a=R_{\rm H}\tan\gamma$ is strictly increasing over the feasible patrol-angle interval, optimization over $\gamma$ is equivalent to optimization over $a$.
\begin{lemma}
\label{lem:energy_opt}
For a fixed airspeed $V$, assume that $P_{{\rm c},k}(\tau;\gamma)$ is positive and differentiable over the feasible patrol-angle interval for $k\in\{{\rm P},{\rm B}\}$. Any interior maximizer $\gamma_{\rm EE}^\star$ of the \ac{CEE} satisfies
\begin{align*}
    &\left. \frac{\partial P_{{\rm c},k}(\tau;\gamma)} {\partial\gamma} \right|_{\gamma=\gamma_{\rm EE}^\star} \!\!\!\!\!=\! -P_{{\rm c},k} (\tau;\gamma_{\rm EE}^\star) \frac{ 2\beta(V)\sec^2\gamma_{\rm EE}^\star }{ \tan^3\gamma_{\rm EE}^\star+ \beta(V)\tan\gamma_{\rm EE}^\star },
\end{align*}
where $\beta(V) = \frac{c_2V^4} {g^2R_{{\rm H},m}^2 (c_1V^4+c_2)}.$
\end{lemma}
\begin{IEEEproof}
See Appendix~\ref{app:lemma5}.    
\end{IEEEproof}
\begin{remark}
The optimality condition separates the communication and aerodynamic contributions to patrol design. The coverage probability determines the geometry-dependent communication benefit, whereas $\beta(V)$ characterizes the relative propulsion penalty of tight turns. In particular, $\frac{{\rm d}\beta(V)}{{\rm d}V} = \frac{4c_2^2V^3} {g^2R_{{\rm H},m}^2 (c_1V^4+c_2)^2}>0.$ Thus, increasing airspeed strengthens the relative turning power penalty at smaller patrol radii, favoring wider trajectories.
\end{remark}
For physical interpretation in the numerical results, we equivalently express the optimum in terms of the patrol radius $a^\star = R_{\rm H}\tan \big( \gamma_{\rm EE}^\star \big).$

\section{Numerical Results and Discussion}
\label{sec:numerical_results}
We validate the analytical results and investigate the effects of geometric, deployment, and flight parameters on coverage probability and \ac{CEE}. Unless otherwise stated, the analytical curves are obtained from the expressions derived in Sections~\ref{sec:distance}-\ref{sec:energy}, with independent 3D Monte Carlo simulations used for validation. The default parameters are $\eta=2$, $H=20$ km, $G_t=G_I=1$, $\lambda_{\rm u}=10^{-4}$, $\nu=0.01$, $n_{\rm H} = 1, 10$, $R_\oplus=6371$~km, and $E_{\min}=0^\circ$.. For evaluations at a fixed SIR threshold, we use $\tau=-10$ dB for \ac{SCR-PCP} and $\tau=-15$ dB for \ac{SCR-BCP}.

We first validate the coverage probability expressions for \ac{SCR-PCP} and \ac{SCR-BCP}. Figure~\ref{fig:n_tau} plots the coverage probability against the \ac{SIR} threshold $\tau$ for patrol radii $a\in\{0.5,3\}$ km. The analytical and simulation results exhibit close agreement for both models, validating the derived coverage expressions. As expected, coverage decreases monotonically with increasing $\tau$. For the considered parameters, \ac{SCR-PCP} exhibits lower coverage at the larger patrol radius, whereas \ac{SCR-BCP} shows a less pronounced dependence on $a$. This difference reflects the distinct intra-ring population models: the mean number of \acp{HAP} per ring increases with $a$ under \ac{SCR-PCP}, while it remains fixed at $n_{\rm H}$ under \ac{SCR-BCP}.
\begin{figure*}[t]
\centering
\subfloat[]
{\includegraphics[width = 0.27\textwidth]{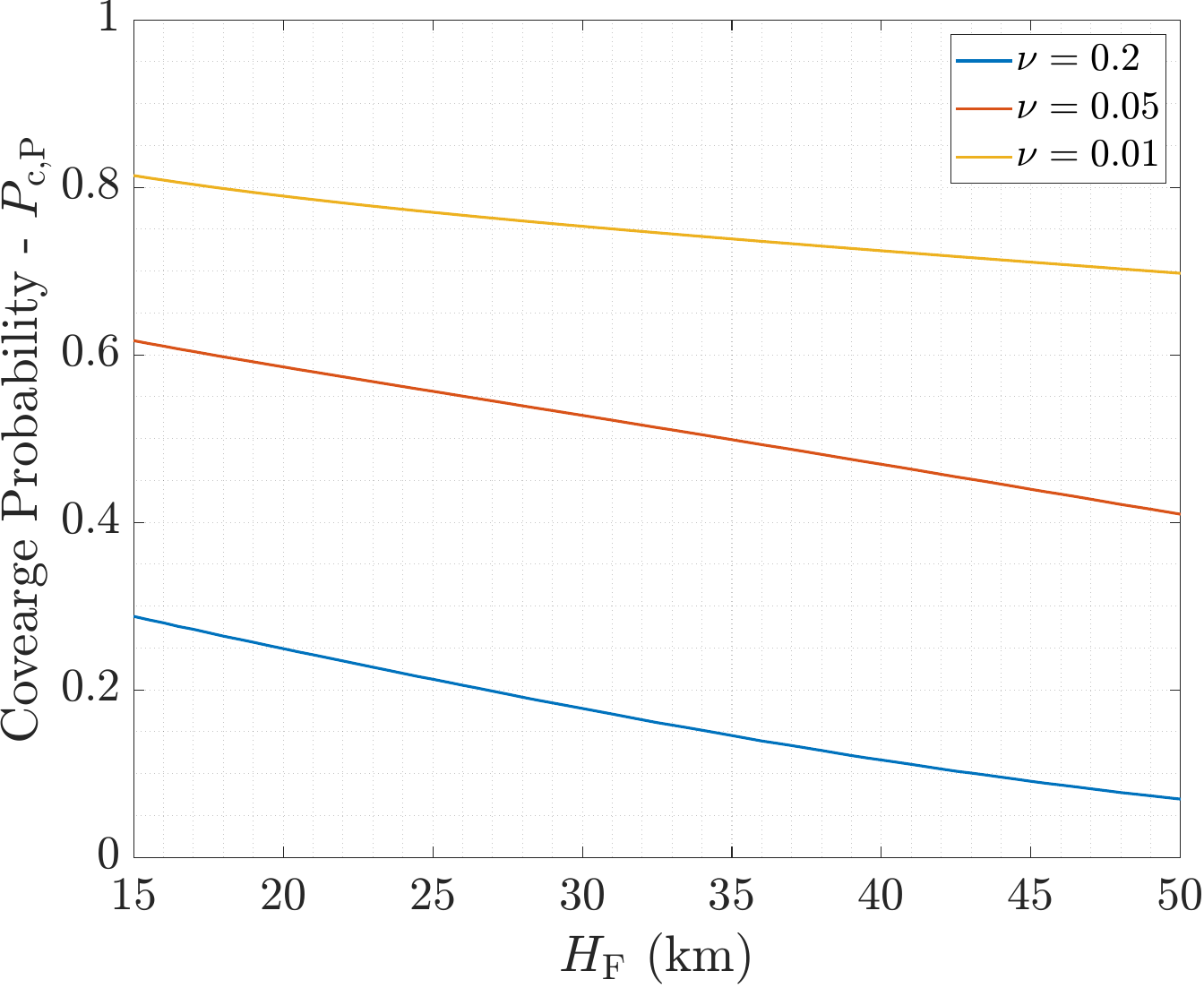}
\label{fig:result_3}}
\hfil
\subfloat[]
{\includegraphics[width = 0.27\textwidth]{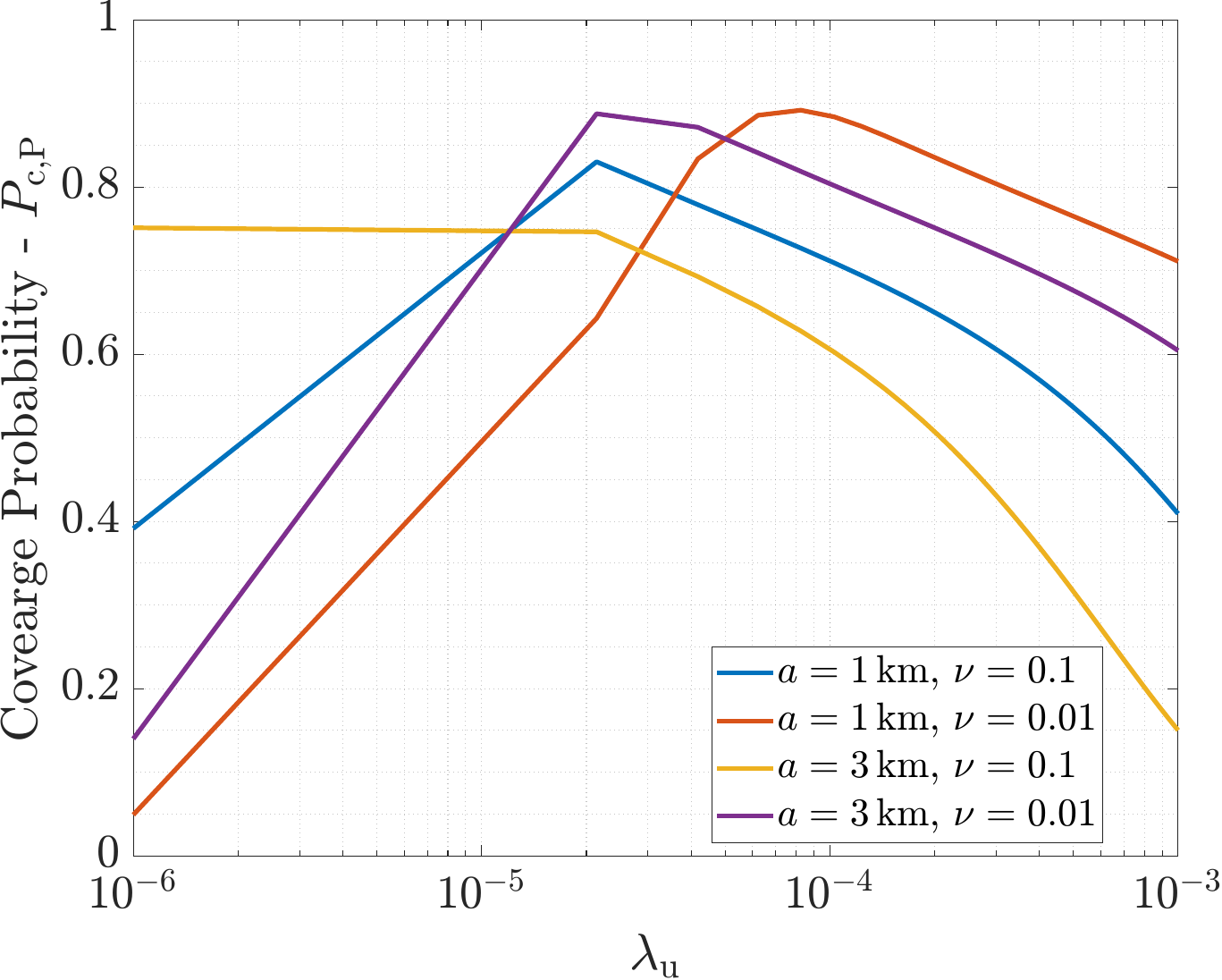}
\label{fig:result_4}}
\hfil
\subfloat[]
{\includegraphics[width = 0.27\textwidth]{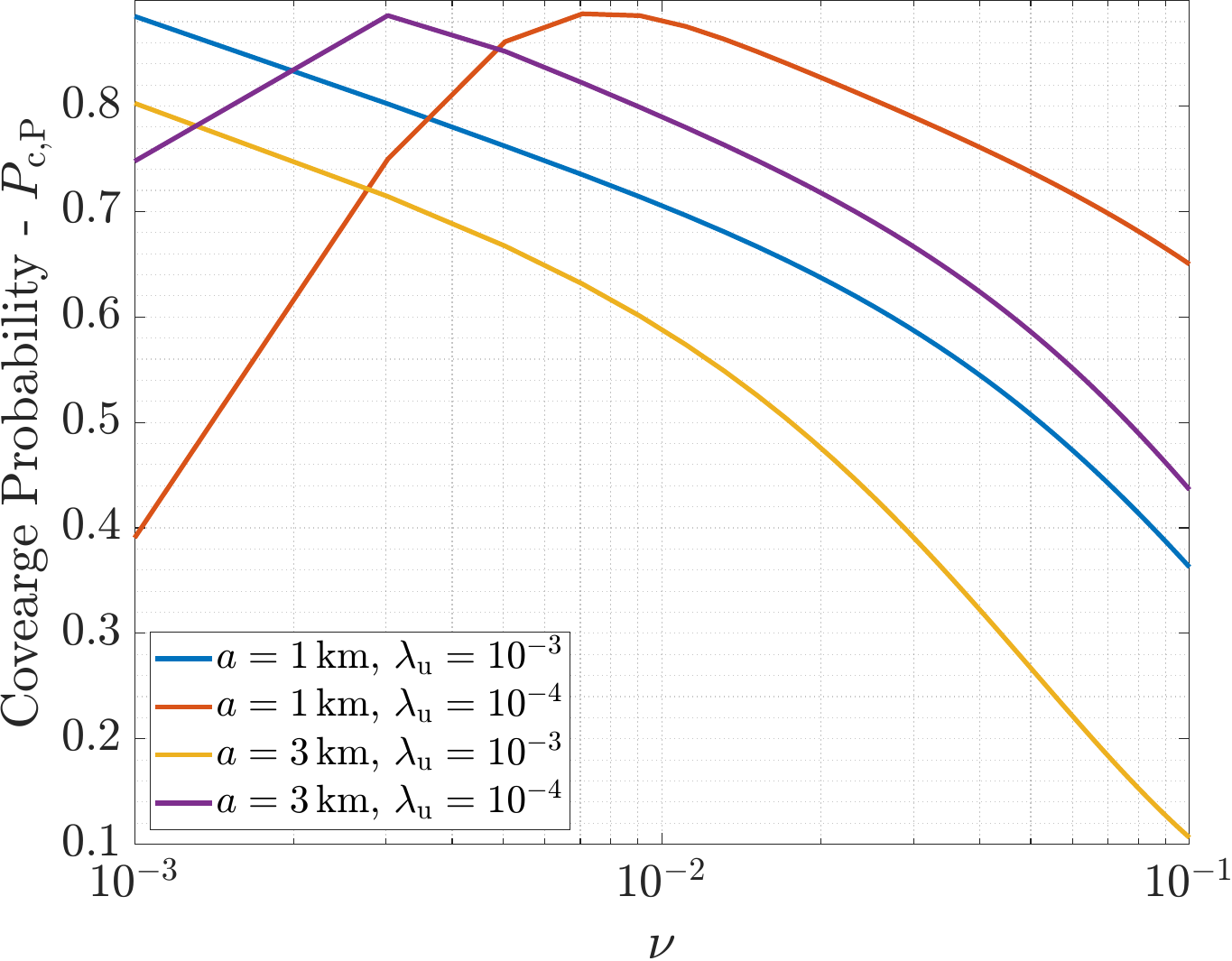}
\label{fig:result_5}}
\caption{\ac{SCR-PCP} coverage probability: (a) versus \ac{HAP} altitude $H_{\rm F}$, (b) versus anchor intensity $\lambda_{\rm u}$, (c) versus intra-ring HAP intensity $\nu$.}
\label{fig:r_3_to_4}
\end{figure*}

\subsection{Coverage Probability under \ac{SCR-PCP}}
\subsubsection{Impact of patrol radius}
Figures~\ref{fig:result_1} and~\ref{fig:result_2} illustrate the dependence of coverage probability on the patrol radius $a$ under \ac{SCR-PCP}. Figure~\ref{fig:result_1} considers different intra-ring \ac{HAP} intensities $\nu$ at fixed $\lambda_{\rm u}=10^{-4}$, whereas Fig.~\ref{fig:result_2} varies the anchor intensity $\lambda_{\rm u}$ at fixed $\nu=0.01$.

For fixed $\lambda_{\rm u}$, the shape of the coverage--radius curve depends strongly on $\nu$. For small $\nu$, the coverage is unimodal in $a$. Tight patrol rings have small circumference and therefore contain few \acp{HAP} on average, since the mean population per ring is $2\pi a\nu$. Increasing $a$ initially increases the availability of potential serving \acp{HAP} and improves the coverage probability. Beyond an intermediate radius, however, the increase in serving distance and interference dominates, causing the coverage probability to decrease. As $\nu$ increases, the maximizing patrol radius shifts rapidly toward smaller values.

A similar trend is observed in Fig.~\ref{fig:result_2} as $\lambda_{\rm u}$ increases. Sparse anchor deployments exhibit a unimodal coverage--radius relationship, whereas denser deployments favor smaller patrol radii. Thus, increasing either $\nu$ or $\lambda_{\rm u}$ shifts the coverage-maximizing radius toward smaller values over the considered parameter ranges. This reflects the diminishing benefit of enlarging patrol rings to improve serving-\ac{HAP} availability as the deployment becomes denser.

\subsubsection{Impact of altitude, anchor intensity, and HAP intensity}
Figure~\ref{fig:r_3_to_4} summarizes the impact of the remaining deployment parameters under \ac{SCR-PCP}. Figure~\ref{fig:result_3} shows that the coverage probability decreases with the actual \ac{HAP} altitude $H_{\rm F}$ for the considered intra-ring intensities $\nu$. Although increasing altitude enlarges the visibility region and potentially improves serving-\ac{HAP} availability, the associated increase in propagation distances and the number of visible interferers dominate, resulting in reduced coverage.

Figure~\ref{fig:result_4} shows the coverage probability as a function of $\lambda_{\rm u}$ for different patrol radii and HAP densities. In sparse HAP regimes, increasing $\lambda_{\rm u}$ initially improves coverage by increasing the likelihood of finding a useful serving \ac{HAP}, whereas excessive anchor density increases inter-ring interference and eventually reduces coverage. As the HAP population per ring increases through larger $a$ or $\nu$, the coverage becomes increasingly interference-limited and decreases with $\lambda_{\rm u}$ over the considered range. Figure~\ref{fig:result_5} plots the coverage probability versus $\nu$. For sparse anchor deployments, the curves are unimodal: increasing $\nu$ initially increases the probability of nearby serving \acp{HAP}, but excessive $\nu$ eventually creates strong intra-ring and inter-ring interference. For dense anchors, the network is interference-limited, and coverage decreases monotonically with $\nu$.

\begin{figure*}[t]
\centering
\subfloat[]
{\includegraphics[width = 0.27\textwidth]{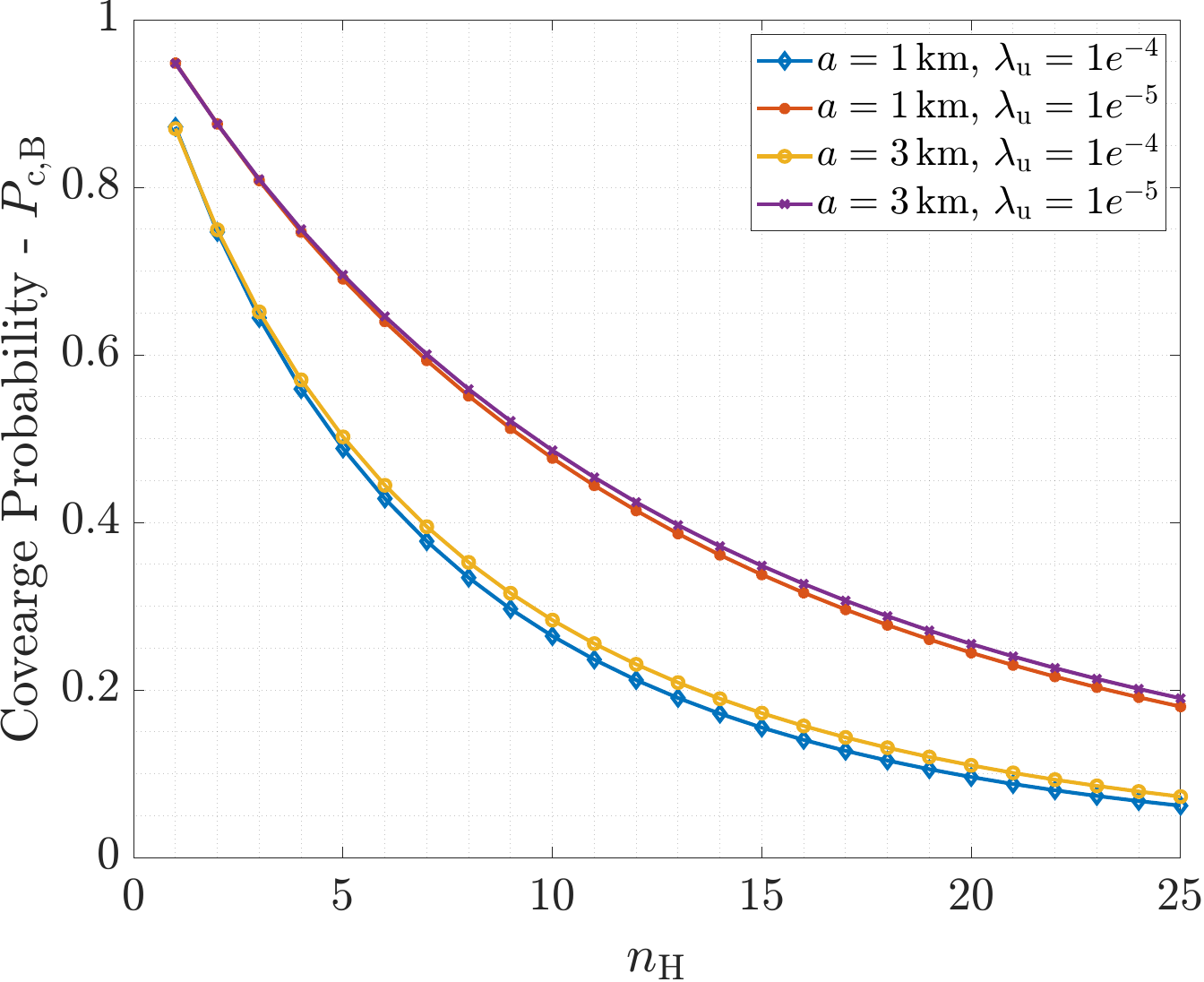}
\label{fig:cov_vs_nH}}
\hfil
\subfloat[]
{\includegraphics[width = 0.27\textwidth]{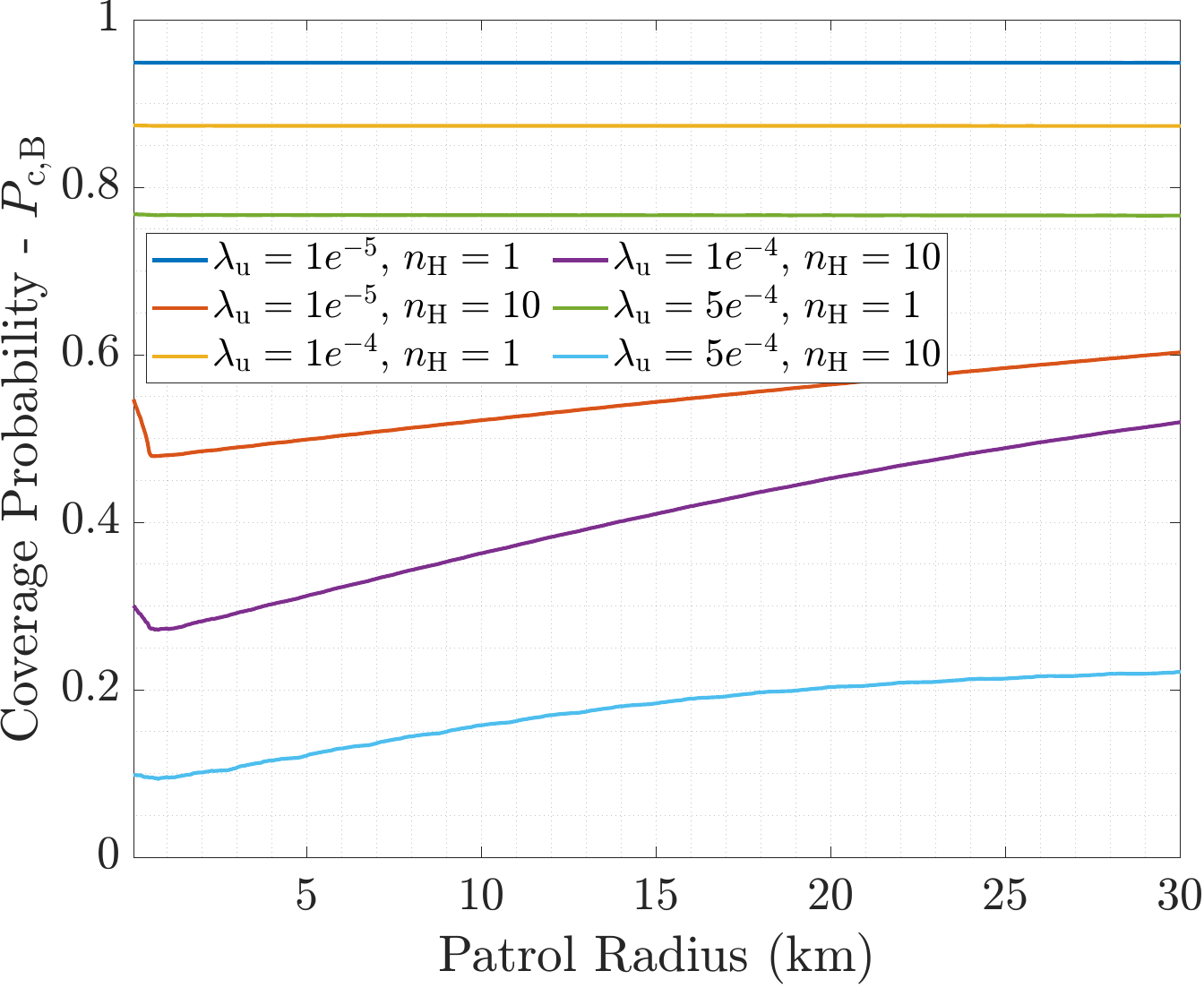}
\label{fig:cov_vs_radius_bcp}}
\hfil
\subfloat[]
{\includegraphics[width = 0.27\textwidth]{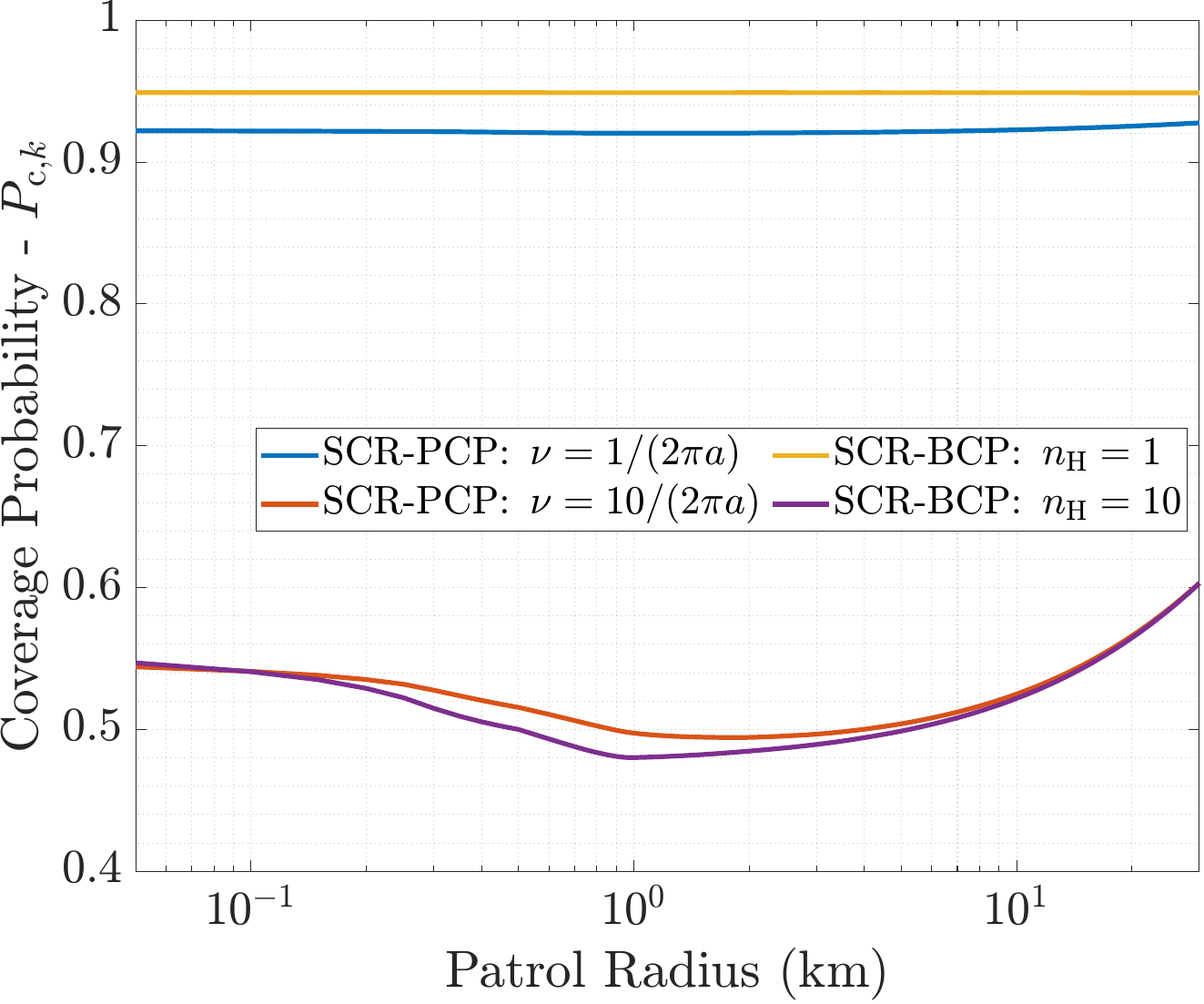}
\label{fig:apples_to_apples}}
\caption{\ac{SCR-BCP} coverage probability: (a) versus per-ring population $n_{\rm H}$, (b) versus patrol radius $a$, (c) Equivalent-density comparison of $P_{\rm c,P}$ and $P_{\rm c,B}$ versus patrol radius $a$ for matched average per-ring populations.}
\label{fig:r_b_1}
\end{figure*}

\subsection{Coverage Probability under \ac{SCR-BCP}}
Figure~\ref{fig:cov_vs_nH} shows the coverage probability versus $n_{\rm H}$ for different $(\lambda_{\rm u},a)$ pairs, while Fig.~\ref{fig:cov_vs_radius_bcp} shows the corresponding dependence on the patrol radius.

Coverage decreases monotonically with $n_{\rm H}$ across the considered configurations. For $\lambda_{\rm u}=10^{-4}$ and $a=1$~km, increasing $n_{\rm H}$ from $1$ to $10$ reduces coverage from approximately $87\%$ to $26\%$. Although additional platforms improve serving-\ac{HAP} availability, the resulting increase in same-ring and inter-ring interference dominates, leading to reduced coverage. The degradation is less pronounced at lower anchor intensities, where inter-ring interference is weaker.

Figure~\ref{fig:cov_vs_radius_bcp} shows that the effect of patrol radius depends strongly on $n_{\rm H}$. For $n_{\rm H}=1$, coverage exhibits limited sensitivity to $a$, consistent with the absence of same-ring interference. For larger $n_{\rm H}$, coverage varies non-monotonically with $a$: it initially decreases slightly before increasing at larger radii. Since the per-ring population remains fixed, enlarging the patrol circumference reduces the spatial concentration of platforms and can alleviate same-ring interference. This benefit becomes more pronounced as $n_{\rm H}$ increases and can outweigh the associated changes in serving distance.

These results highlight the distinct roles of platform population and patrol geometry under \ac{SCR-BCP}. Unlike \ac{SCR-PCP}, where the mean per-ring population $2\pi a\nu$ increases with patrol radius, \ac{SCR-BCP} maintains exactly $n_{\rm H}$ platforms per ring. Consequently, enlarging the patrol radius can improve coverage in configurations where reduced same-ring interference outweighs the geometric penalties.

\subsection{Comparison of \ac{SCR-PCP} and \ac{SCR-BCP}}
To isolate the effect of random versus fixed per-ring populations, we consider an equivalent-density setting in which both models have the same mean number of \acp{HAP} per patrol ring. Specifically, for a fixed $n_{\rm H}$ under \ac{SCR-BCP}, the intra-ring intensity of \ac{SCR-PCP} is set to $\nu=n_{\rm H}/(2\pi a)$, ensuring $\mathbb{E}[N_{\rm ring}]=2\pi a\nu=n_{\rm H}$ for all patrol radii. Figure~\ref{fig:apples_to_apples} compares the resulting coverage probabilities as functions of $a$ for $\lambda_{\rm u}=10^{-5}$ and $\tau=-15$ dB. For $n_{\rm H}=1$, \ac{SCR-BCP} consistently outperforms \ac{SCR-PCP} over the considered range, with both models exhibiting limited sensitivity to $a$. Under \ac{SCR-BCP}, the serving ring contains no residual interferers, whereas under \ac{SCR-PCP}, additional platforms may remain on the visible portion of the serving ring after association, contributing to interference.

For $n_{\rm H}=10$, neither model exhibits uniform dominance. At small patrol radii, \ac{SCR-BCP} achieves slightly higher coverage, whereas \ac{SCR-PCP} becomes superior as $a$ increases. Despite identical mean per-ring populations, the models differ in their conditional interference statistics. Under \ac{SCR-BCP}, the serving ring retains exactly $n_{\rm H}-1$ residual platforms, while under \ac{SCR-PCP}, the residual population remains Poisson distributed on the admissible arcs. The latter variability allows realizations with fewer nearby interferers, which can favor \ac{SCR-PCP} under certain deployment configurations. These results demonstrate that matching the mean platform density does not ensure equivalent coverage performance. The choice between \ac{SCR-PCP} and \ac{SCR-BCP} should therefore reflect whether the deployment permits a variable per-ring population or requires a fixed allocation of platforms.
\begin{figure*}[t]
\centering
\subfloat[]
{\includegraphics[width = 0.25\textwidth]{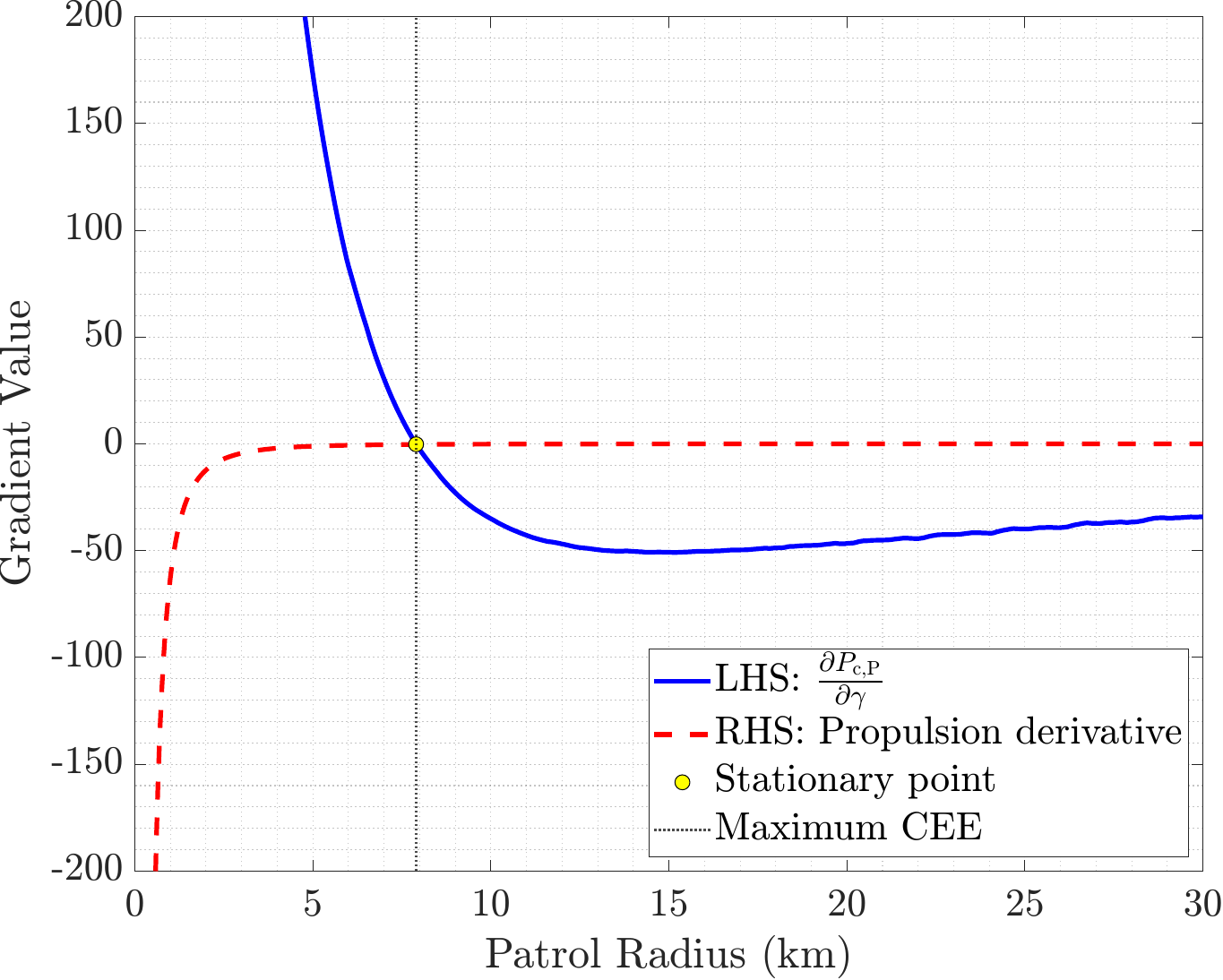}
\label{fig:result_6}}
\hfil
\subfloat[]
{\includegraphics[width = 0.23\textwidth]{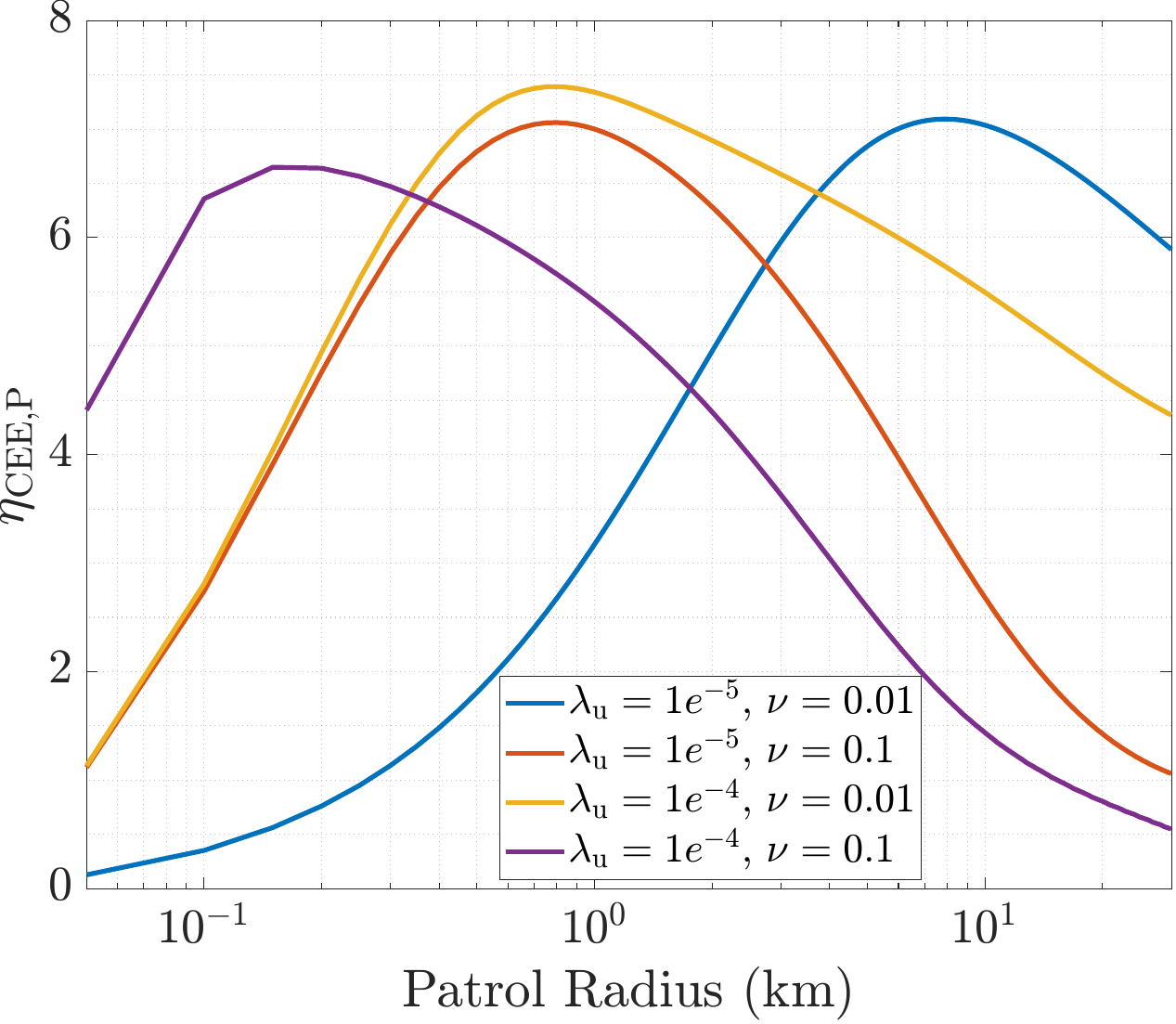}
\label{fig:result_7}}
\hfil
\subfloat[]
{\includegraphics[width = 0.23\textwidth]{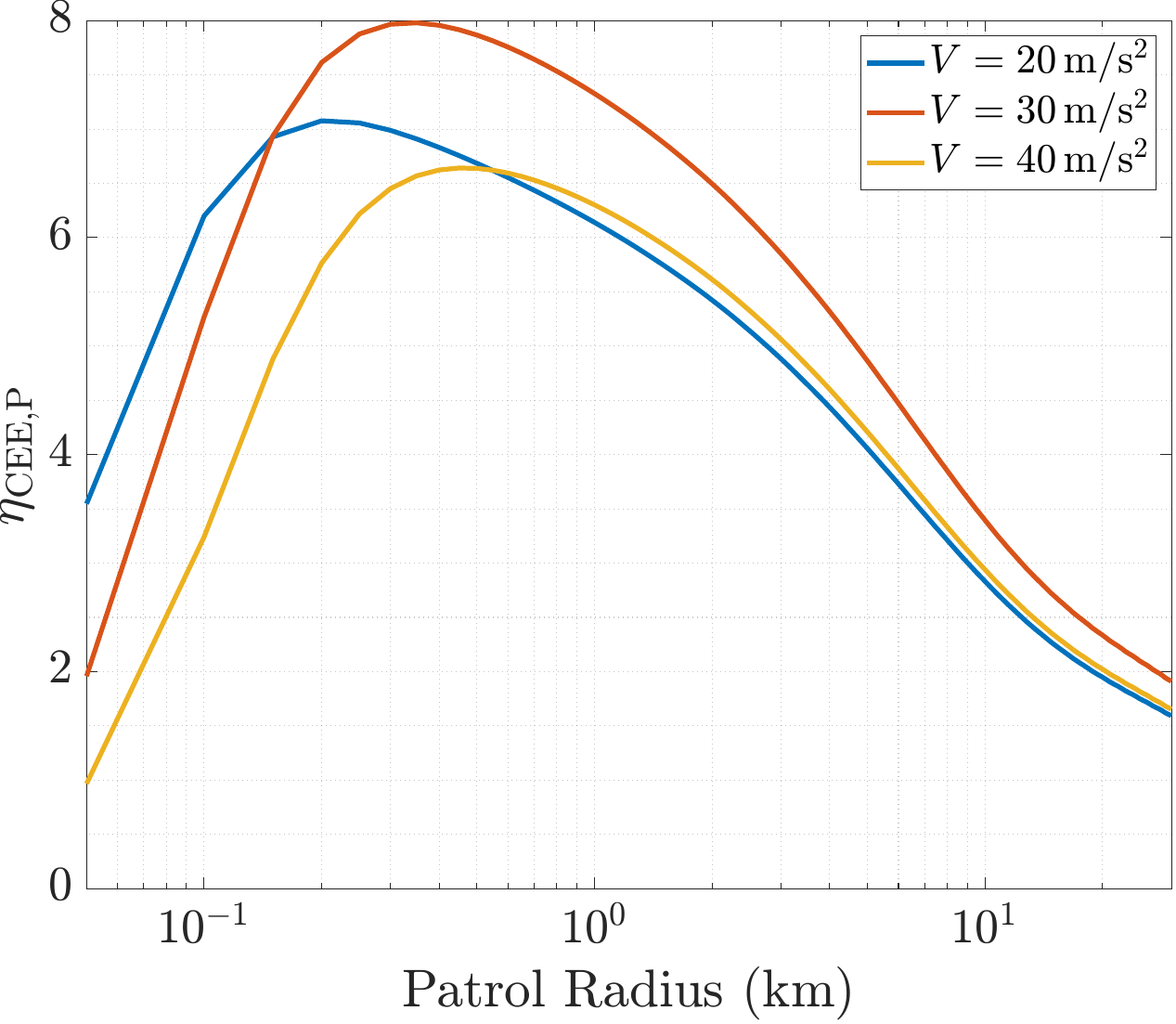}
\label{fig:result_8}}
\hfil
\subfloat[]
{\includegraphics[width = 0.23\textwidth]{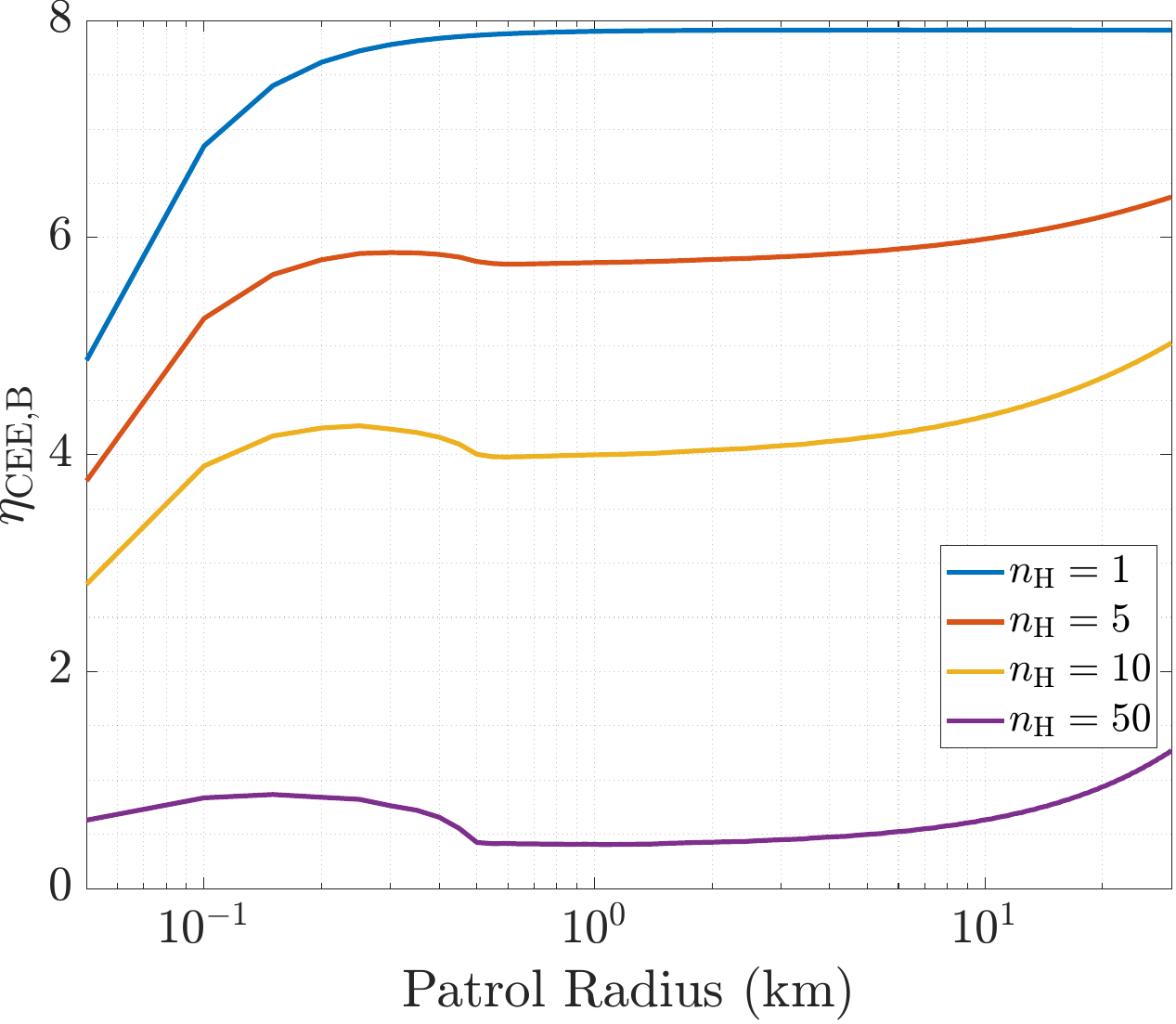}
\label{fig:cee_nH_variation}}
\hfil
\subfloat[]
{\includegraphics[width = 0.23\textwidth]{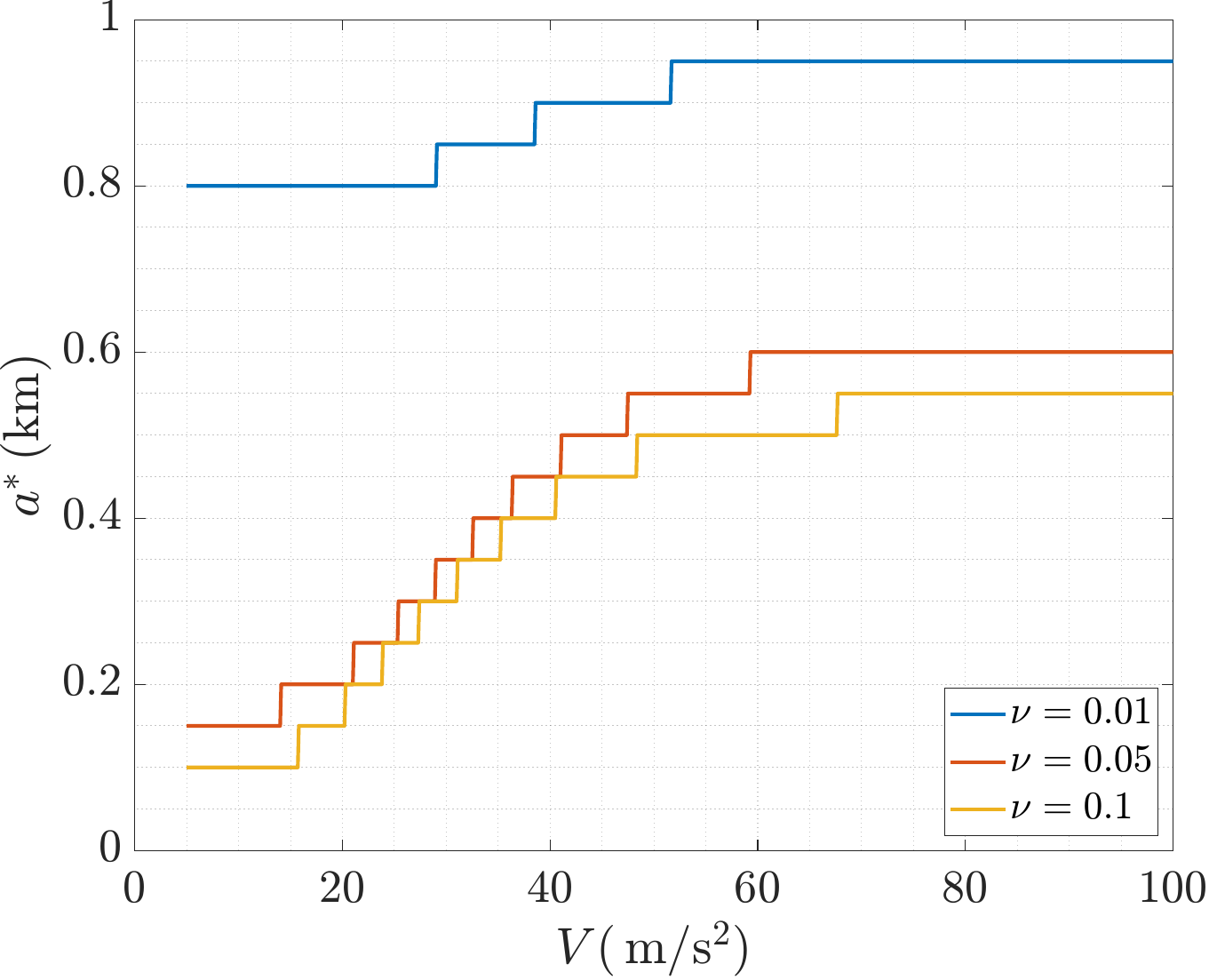}
\label{fig:result_9}}
\hfil
\subfloat[]
{\includegraphics[width = 0.23\textwidth]{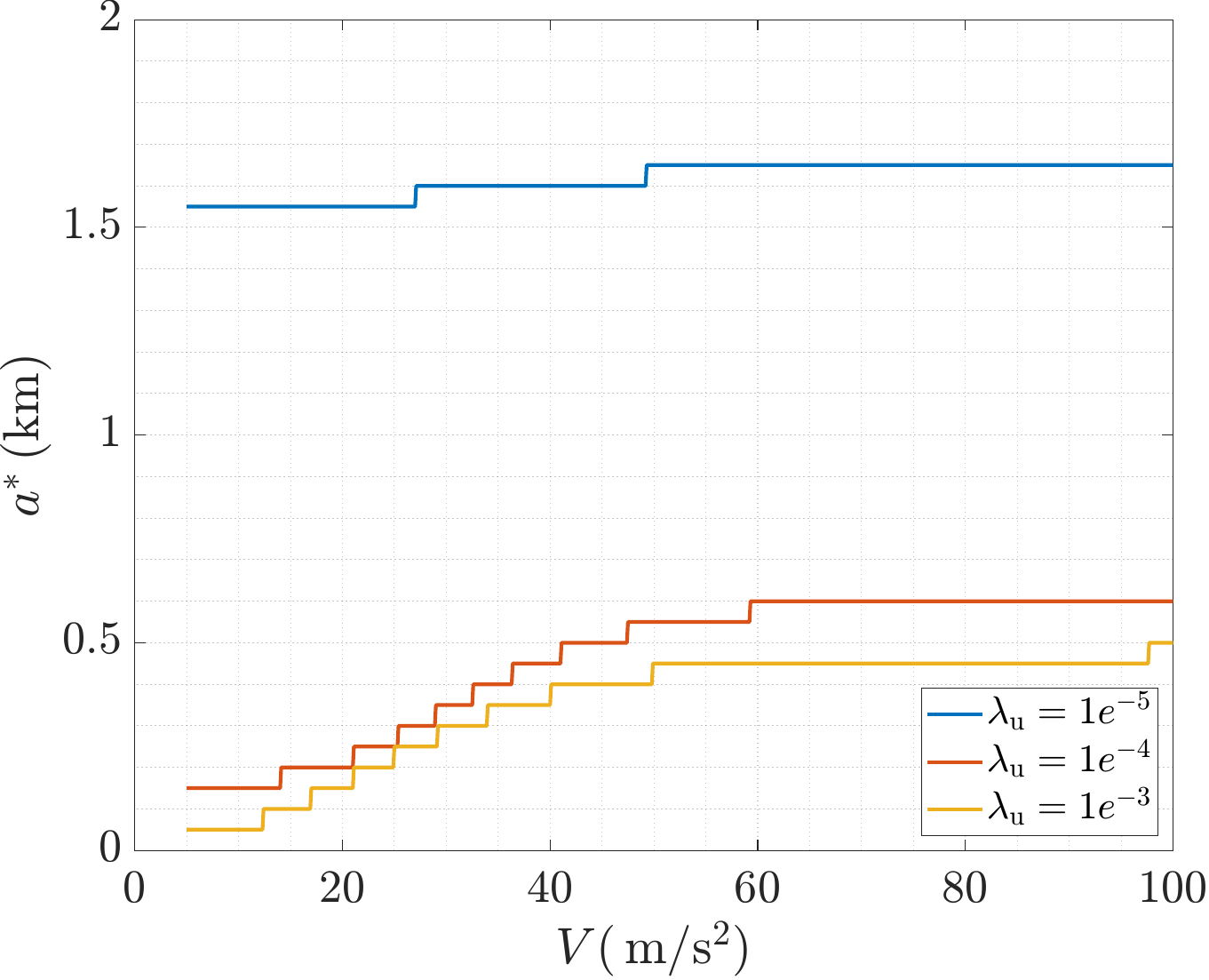}
\label{fig:result_10}}
\hfil
\subfloat[]
{\includegraphics[width = 0.23\textwidth]{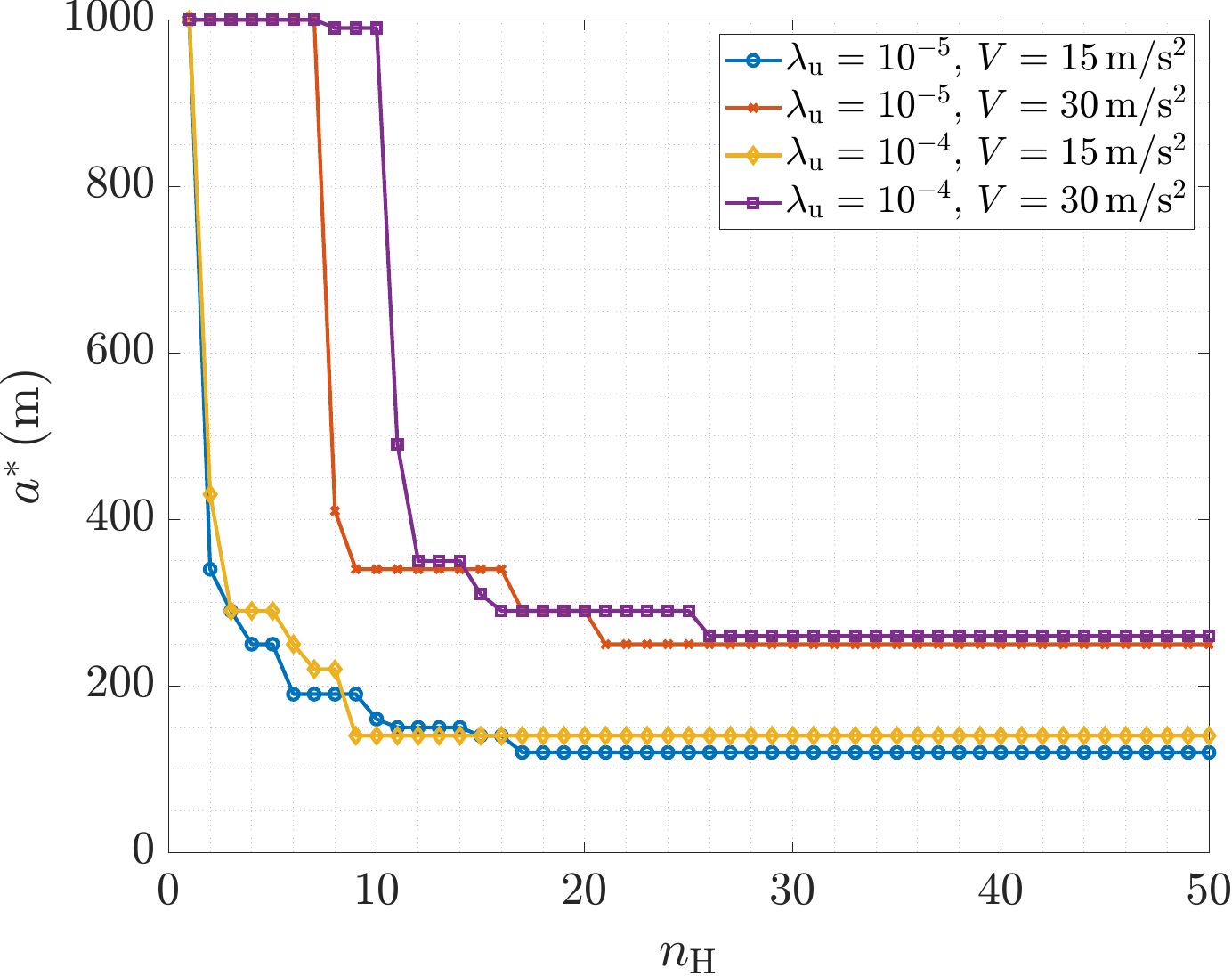}
\label{fig:opt_radius}}
\hfil
\subfloat[]
{\includegraphics[width = 0.23\textwidth]{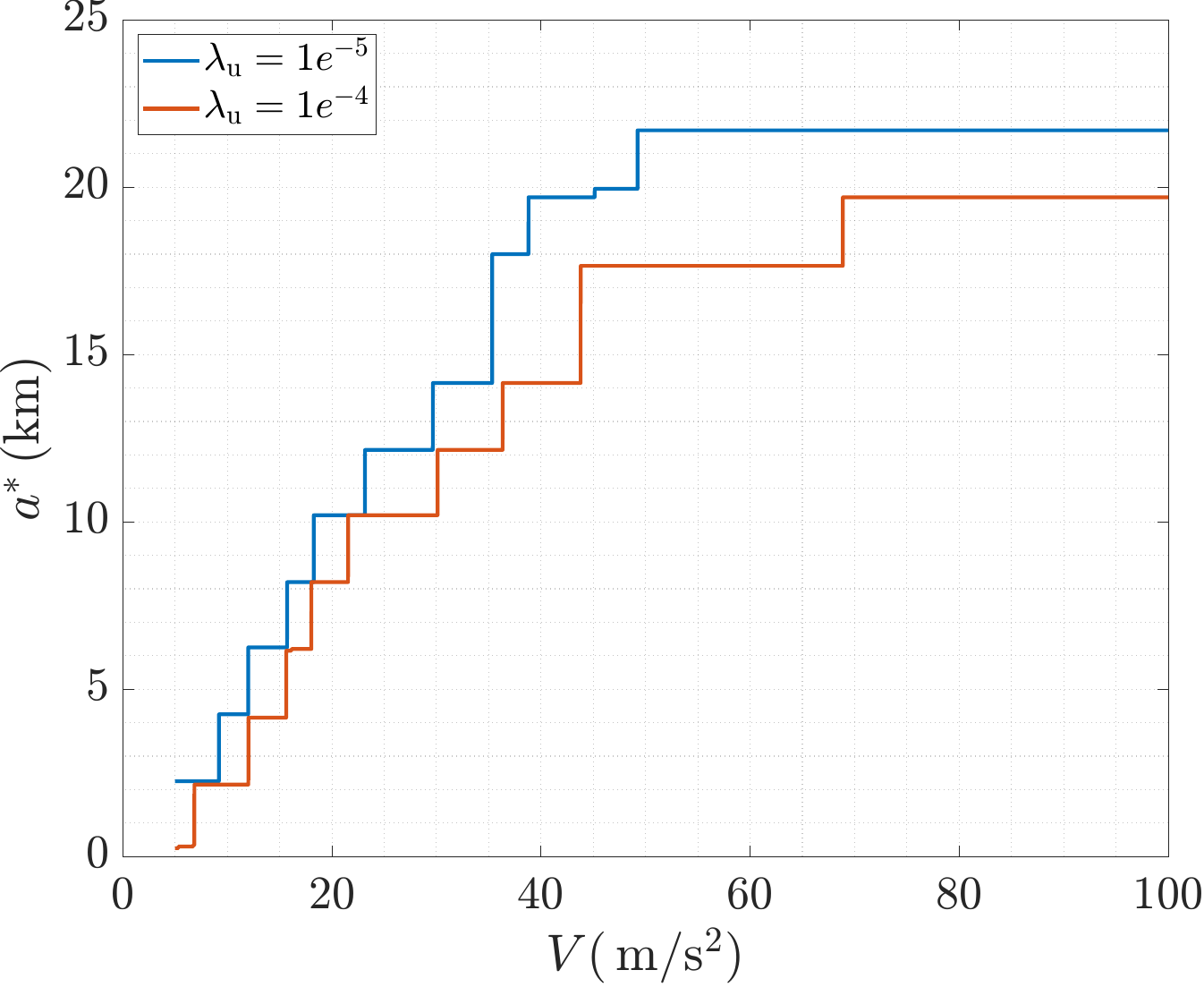}
\label{fig:opt_a_vs_velocity}}
\caption{(a) Validation of Lemma~\ref{lem:energy_opt}. (b)-(c) \ac{SCR-PCP} \ac{CEE} versus patrol radius. (d) \ac{SCR-BCP} \ac{CEE} versus patrol radius. (e)-(f) \ac{SCR-PCP} optimal patrol radius versus velocity. (g) \ac{SCR-BCP} optimal patrol radius versus $n_{\rm H}$. (h) \ac{SCR-BCP} optimal patrol radius versus velocity.}
\label{fig:r_6_to_8}
\end{figure*}

\subsection{Coverage Energy Efficiency under \ac{SCR-PCP}}
\label{subsec:numerical_cee}
We evaluate the \ac{CEE} under \ac{SCR-PCP} using the aerodynamic flight model developed in Section~\ref{sec:energy}. From \eqref{eq:cee_def} and \eqref{eq:power_scf}, the \ac{CEE}, with airspeed $V$ is
\begin{align*}
    \eta_{{\rm CEE},{\rm P}}(\gamma;V)
    =\frac{P_{{\rm c},{\rm P}}(\tau;\gamma)}
    {P_{\rm SHF}(V)\left(1+\beta(V)\tan^{-2}\gamma\right)},
\end{align*}
where $P_{\rm SHF}(V)=c_1V^3+c_2/V$ and $\beta(V)$ is defined in Lemma~\ref{lem:energy_opt}. Figure~\ref{fig:result_6} numerically validates the stationarity condition in Lemma~\ref{lem:energy_opt}: the intersection of its two sides coincides with the numerically determined \ac{CEE} maximum. At small patrol radii, the increased induced-power requirement for tight circular turns reduces the \ac{CEE}, highlighting the trade-off between coverage performance and propulsion power.

\subsubsection{Impact of deployment intensities and airspeed}
Figure~\ref{fig:result_7} shows the \ac{SCR-PCP} \ac{CEE} versus patrol radius for different combinations of anchor intensity $\lambda_{\rm u}$ and intra-ring \ac{HAP} intensity $\nu$. Each configuration exhibits a single maximum over the considered radius range. Decreasing $\nu$ shifts the energy-optimal radius $a^\star$ toward larger values, whereas increasing $\lambda_{\rm u}$ favors smaller radii. The peak \ac{CEE}, however, responds differently to anchor intensity: increasing $\lambda_{\rm u}$ improves the maximum for $\nu=0.01$ but reduces it for $\nu=0.1$. This highlights the coupled influence of serving-\ac{HAP} availability and interference on energy-efficient patrol design.

Figure~\ref{fig:result_8} examines the effect of cruising airspeed for $V\in\{20,30,40\}$ m/s. The optimal radius increases with $V$, as higher airspeeds impose a greater relative propulsion penalty on tight turns. However, the peak \ac{CEE} varies non-monotonically, attaining $7.08$, $7.98$, and $6.64$ $\mathrm{kW}^{-1}$, respectively. The maximum occurs at $V=30$ m/s, near the minimum-power airspeed of the adopted fixed-wing model. This behavior reflects the non-monotonic dependence of total propulsion power on $V$, despite the monotonic increase of the relative turning penalty $\beta(V)$.

\subsubsection{Dependence of optimal patrol radius on airspeed}
Figures~\ref{fig:result_9} and~\ref{fig:result_10} illustrate the dependence of the energy-optimal patrol radius $a^\star$ on airspeed $V$ under \ac{SCR-PCP}. The optimal radius is nondecreasing with $V$ across the considered configurations, consistent with the increasing relative penalty associated with tight circular turns. Figure~\ref{fig:result_9} examines the effect of $\nu$ at $\lambda_{\rm u}=10^{-4}$. Smaller $\nu$ favors larger patrol radii. Over $V\in[5,100]$ m/s, $a^\star$ increases from $0.8$ to $0.95$ km for $\nu=0.01$, compared with $0.15$ to $0.6$ km for $\nu=0.05$ and $0.1$ to $0.55$ km for $\nu=0.1$. Figure~\ref{fig:result_10} examines the effect of $\lambda_{\rm u}$ at fixed $\nu=0.05$, showing that higher anchor intensities favor smaller optimal radii. The curves gradually flatten at higher airspeeds, consistent with the saturation of the velocity-dependent parameter $\beta(V)$.

Overall, these results demonstrate that energy-efficient patrol design under \ac{SCR-PCP} requires accounting for both deployment intensity and cruising airspeed.

\subsection{Coverage Energy Efficiency under \ac{SCR-BCP}}
Figure~\ref{fig:cee_nH_variation} shows the \ac{CEE} under \ac{SCR-BCP} as a function of patrol radius $a$ for $n_{\rm H}\in\{1,5,10,50\}$ and $\lambda_{\rm u}=10^{-5}$, illustrating the interplay between propulsion power and interference in fixed-population deployments.

For $n_{\rm H}=1$, the absence of same-ring interference makes coverage relatively insensitive to $a$. Consequently, the \ac{CEE} increases with patrol radius over the considered range and gradually levels off as the turning power penalty diminishes. For $n_{\rm H}\geq10$, the behavior becomes non-monotonic. Increasing $a$ initially improves \ac{CEE} by reducing the propulsion cost of tight circular turns. At intermediate radii, however, the degradation in coverage outweighs the aerodynamic benefit, leading to a decline in \ac{CEE}. At larger radii, the reduced spatial concentration of same-ring interferers produces a modest recovery. Increasing $n_{\rm H}$ also reduces the overall \ac{CEE}, consistent with the interference effects observed in the preceding coverage analysis.

\subsubsection{Impact of airspeed on the optimal patrol radius}
Figure~\ref{fig:opt_a_vs_velocity} illustrates the \ac{CEE}-optimal patrol radius $a^\star$ versus airspeed $V\in[5,100]$ m/s for $n_{\rm H}=1$ and $\lambda_{\rm u}\in\{10^{-5},10^{-4}\}$. For both anchor intensities, $a^\star$ increases with $V$ before reaching the upper bound of the considered radius interval. This behavior is consistent with the increasing relative propulsion penalty of tight turns at higher airspeeds. Higher anchor intensity generally favors smaller optimal radii because the resulting coverage-radius relationship reflects stronger inter-ring interference. At sufficiently high velocities, the optimal radii converge toward the upper search boundary, reducing the difference between the two configurations.

\subsubsection{Impact of per-ring population on the optimal patrol radius}
Figure~\ref{fig:opt_radius} shows the \ac{CEE}-optimal patrol radius $a^\star$ as a function of $n_{\rm H}\in[1,50]$ for four $(\lambda_{\rm u},V)$ configurations. The optimization is restricted to $a\in[0.05,1]$ km to examine the compact-patrol regime, where the coverage-propulsion trade-off is most pronounced. Since the \ac{CEE} may exhibit a secondary increase at larger radii, the reported optimum is defined over this prescribed interval rather than the entire patrol-radius domain. For small $n_{\rm H}$, propulsion power dominates the trade-off, placing the optimal radius near the upper boundary. As $n_{\rm H}$ increases, stronger same-ring interference favors smaller radii, and $a^\star$ eventually stabilizes at a deployment-dependent value. Higher anchor intensities also favor smaller radii through their effect on the coverage--radius relationship. In contrast, increasing airspeed strengthens the relative propulsion penalty of tight turns, favoring wider patrol trajectories under the adopted aerodynamic model.

Overall, these results demonstrate that energy-efficient patrol design under \ac{SCR-BCP} depends jointly on the per-ring population, anchor intensity, and cruising airspeed.

\section{Conclusion}
\label{sec:conclusion}
This paper develops an \ac{SG} framework for patrol-based \ac{HAP} networks on the spherical Earth. To capture the anchor-driven and cyclic nature of realistic \ac{HAP} deployments, we introduced two small-circle ring point-process models: the \ac{SCR-PCP}, which represents intensity-driven deployments through a Poisson intra-ring structure, and the \ac{SCR-BCP}, which captures dimensioned finite-fleet deployments through a fixed per-ring population. Building on these models, we develop spatial statistics, derive expressions for coverage probability under nearest-visible \ac{HAP} association, and incorporate fixed-wing flight mechanics to formulate an energy-aware performance framework.

The analysis reveals that patrol geometry plays a fundamentally different role under the two deployment regimes. In \ac{SCR-PCP}, the patrol radius directly controls both serving-link accessibility and interference generation, leading to non-trivial coverage and \ac{CEE} optima. In contrast, \ac{SCR-BCP} is primarily governed by the fixed per-ring population, with the influence of patrol radius becoming more pronounced as the number of platforms per ring increases. By coupling communication performance with steady-circular-flight propulsion costs, we further showed that the patrol radius should be jointly designed with platform density, fleet dimensioning, and cruising velocity to achieve energy-efficient operation.

A promising future direction is the integration of patrol-based \ac{HAP} models into multi-layer terrestrial, aerial, and satellite architectures, enabling joint optimization of mobility, coverage, and spectrum coexistence in future deployments.

\bibliography{references}
\bibliographystyle{ieeetr}

\appendices
\section{Proof of Theorem~\ref{thm:isotropy}}
\label{app:theorem1}
\begin{IEEEproof}
We prove the invariance sequentially through the three tiers of the hierarchical construction.

\emph{1) Anchor Process:} We first establish the isotropy of the anchor \ac{PPP} $\Phi_u$. Let $U$ be a random point uniformly distributed on the unit sphere $\mathcal{S}_1$. It is well known that $U$ can be represented in the orthonormal basis $(e_x,e_y,e_z)$ as $U = \big(\sqrt{1-V^2}\cos\Theta, \sqrt{1-V^2}\sin\Theta, V\big)$, where $V \sim \mathrm{Uniform}(-1,1)$ and $\Theta \sim \mathrm{Uniform}(0,2\pi)$ are independent. In spherical coordinates $U = (\sin\Phi\cos\Theta, \sin\Phi\sin\Theta, \cos\Phi)$, the polar angle $\Phi = \arccos(V)$ has the \ac{CDF}
\begin{align*}
    \mathbb{P}(\Phi < x) = \mathbb{P}(V > \cos x) = \frac{1 - \cos x}{2}, \quad 0 \le x < \pi.
\end{align*}
which yields density $f_\Phi(x) = (\sin x)/2$. Together with the uniformity of $\Theta$, this implies that the intensity measure of a homogeneous \ac{PPP} on $\mathcal{S}_{R_{\rm H}}$ is proportional to the surface element $R_{\rm H}^2\sin\phi\,\mathrm{d}\phi\,\mathrm{d}\theta$, then for any rotation $Q$ of $\mathbb{R}^3$ about the origin, $Q(\Phi_u) \stackrel{d}{=} \Phi_u$. Hence $\Phi_u$ is isotropic.

\emph{2) Patrol-Ring Process:} Consider a specific anchor $u \in \Phi_u$. Since $Q$ preserves inner products and norms, the image of the tangent plane at an anchor $u$ is exactly the tangent plane at the rotated anchor, $Q(\Pi(u)) = \Pi(Q(u))$, and the image of the outward normal is $Q(n_u) = n_{Q(u)}$. Similarly, the images of the basis vectors $(e_1(u),e_2(u))$ form an orthonormal basis of $\Pi(Q(u))$. Therefore
\begin{align*}
    &Q(\mathcal{C}(u,\gamma)) \\
    &\hspace*{0.2cm} =\left\{ Q(u) + a(\cos t\, Q(e_1(u)) + \sin t\, Q(e_2(u))) : t \in [0,2\pi) \right\} \\
    &\hspace*{0.2cm} = \mathcal{C}(Q(u),\gamma),
\end{align*}
i.e., each patrol ring is mapped to the patrol ring associated with the rotated anchor $Q(u)$, with the same radius $a$ and parametrization up to a rotation. Applying $Q$ to the entire collection of rings yields
\begin{align*}
    Q(\mathcal{L}) = \big\{\mathcal{C}(Q(u),\gamma) : u \in \Phi_u\big\},
\end{align*}
which has the same distribution as $\mathcal{L}$ because $\Phi_u$ is isotropic and the mapping $u \mapsto \mathcal{C}(u,\gamma)$ is deterministic.

\emph{3) Intra-Ring \ac{HAP} Process:} Next, consider the intra-ring point processes. In the \ac{SCR-PCP} model, conditioned on $\Phi_u$, the active \ac{HAP} on ring $\mathcal{C}(u,\gamma)$ form a one-dimensional homogeneous \ac{PPP} on the circle with linear intensity $\nu$ (per unit arc length). A homogeneous \ac{PPP} on a circle is invariant in distribution under rotations of the circle, i.e., if $R_\theta$ denotes a rotation of the arc parameter by an angle $\theta$, then $R_\theta \mathcal{P}_u \stackrel{d}{=} \mathcal{P}_u$ for all $\theta$. Therefore, when the ambient space is rotated by $Q$, the induced \ac{PPP} on $Q(\mathcal{C}(u,\gamma))=\mathcal{C}(Q(u),\gamma)$ has exactly the same law as the \ac{PPP} on $\mathcal{C}(Q(u),\gamma)$ constructed directly from anchor $Q(u)$.

In the \ac{SCR-BCP} model, conditioned on anchor $u$, the associated \ac{BPP} $\mathcal{B}_u$ consists of exactly $n_{\rm H}$ i.i.d.\ uniform points on $\mathcal{C}(u,\gamma)$. A finite set of i.i.d.\ uniform points on a circle is also rotation invariant: for any rotation of the circle, the image of the configuration is again a set of $n_{\rm H}$ i.i.d.\ uniform points on the circle. Hence, after applying $Q$ to the ambient space, the point set on $\mathcal{C}(Q(u),\gamma)$ obtained from $Q(\mathcal{B}_u)$ has the same distribution as the binomial configuration $\mathcal{B}_{Q(u)}$ generated directly at $Q(u)$.

Since each rotation $Q$ acts independently and consistently on every anchor and its associated ring, the joint distribution of the collection of intra-ring processes is preserved under $Q$. In particular, for $k\in\{{\rm P,B}\}$,
\begin{align*}
    Q(\Phi_{\rm H}^{k}) &= Q\Big(\bigcup_{u\in\Phi_u} \Psi^{k}_{u}\Big) \\
    &= \bigcup_{u\in\Phi_u} Q(\Psi^{k}_{u}) \stackrel{d}{=} \bigcup_{u\in\Phi_u} \Psi^{k}_{Q(u)} = \bigcup_{v\in Q(\Phi_u)} \Psi^{k}_{v},
\end{align*}
where $\Psi^{\rm P}_u=\mathcal{P}_u$ and $\Psi^{\rm B}_u=\mathcal{B}_u$ denote the intra-ring processes. Since $Q(\Phi_u) \stackrel{d}{=} \Phi_u$, the right-hand side has the same distribution as $\Phi_{\rm H}^{k}$, and the claim follows.
\end{IEEEproof}

\section{Proof of Lemma~\ref{lem:dist_nearest_anchor}}
\label{app:lemma1}
\begin{IEEEproof}
Let $u \in \Phi_u$ be an anchor located at a central angle $\phi$ from the typical user $\mathbf{o}$. By the law of cosines in $\mathbb{R}^3$, the Euclidean distance between them is 
\begin{align*}
    D(\phi) = \sqrt{R_\oplus^2 + R_{\rm H}^2 - 2R_\oplus R_{\rm H} \cos\phi}.
\end{align*}
The \ac{CDF} of $D_{\rm u}$ is $F_{D_{\rm u}}(d) = 1 - \mathbb{P}(D_{\rm u} > d)$. The event $\{D_{\rm u} > d\}$ requires that no anchor lies within distance $d$ of the user. Geometrically, the exclusion region on $\mathcal{S}_{R_{\rm H}}$ is the spherical cap $\mathcal{C}_{\rm cap}(d) = \{x \in \mathcal{S}_{R_{\rm H}} : \|x - \mathbf{o}\| \le d\}.$ The boundary of this cap is determined by the maximum central angle $\phi_{\max}(d)$ such that $D(\phi_{\max}(d)) = d$, which yields 
\begin{align*}
    \cos\phi_{\max}(d) = \frac{R_\oplus^2 + R_{\rm H}^2 - d^2}{2R_\oplus R_{\rm H}}.
\end{align*}
The corresponding cap area is
\begin{align*}
    &|\mathcal{C}_{\rm cap}(d)| = \int_0^{2\pi} \int_0^{\phi_{\max}(d)} R_{\rm H}^2 \sin\phi\, {\rm d}\phi\, {\rm d}\theta \\
    &= 2\pi R_{\rm H}^2 \big[-\cos\phi\big]_0^{\phi_{\max}(d)} = 2\pi R_{\rm H}^2 \big(1 - \cos\phi_{\max}(d)\big).
\end{align*}
Substituting the expression for $\cos\phi_{\max}(d)$ gives
\begin{align*}
    |\mathcal{C}_{\rm cap}(d)| = 2\pi R_{\rm H}^2 \left(1 - \frac{R_\oplus^2 + R_{\rm H}^2 - d^2}{2R_\oplus R_{\rm H}}\right).
\end{align*}
Since anchors form a homogeneous \ac{PPP} with surface intensity $\lambda_{\rm u}$, the void probability is $\mathbb{P}(D_{\rm u} > d) = \exp\big(-\lambda_{\rm u} |\mathcal{C}_{\rm cap}(d)|\big),$ which, substituted into $F_{D_{\rm u}}(d) = 1 - \mathbb{P}(D_{\rm u} > d)$, yields the claimed \ac{CDF}. For $d>R_\oplus+R_{\rm H}$, the corresponding ball contains the entire anchor shell, so $\mathbb{P}(D_{\rm u}>d) = \mathbb{P}(\Phi_u=\varnothing) = \exp\left(-4\pi R_{\rm H}^2\lambda_{\rm u}\right).$ Hence the remaining probability mass is located at infinity.
\end{IEEEproof}

\section{Proof of Corollary~\ref{cor:serving_angle_pcp} and~\ref{cor:serving_angle_bcp}}
\subsection{Proof of Corollary~\ref{cor:serving_angle_pcp}}
\label{app:cor1}
\begin{IEEEproof}
For a ring at central angle $\phi$, the quantity $-\partial_x p_{{\rm void},{\rm P}}(\phi,x)\,{\rm d}x$ gives the first-order probability that its nearest \ac{HAP} lies in $[x,x+{\rm d}x]$. Multiplying by the anchor intensity measure $2\pi\lambda_{\rm u}R_{\rm H}^2\sin\phi\,{\rm d}\phi$ and the probability $V_{\rm P}(x)$ that no other ring contains a closer \ac{HAP} yields \eqref{eq:joint_pdf_pcp}. Integrating over $\phi$ gives $f_{d_{0,{\rm P}}}(x)=-V_{\rm P}'(x)$, and normalization yields \eqref{eq:conditional_serving_angle_pcp}. The joint density integrates to $1-V_{\rm P}(\bar d_{\rm H})$, with the remaining probability mass corresponding to the no-visible-\ac{HAP} event.

Finally, differentiating $L(\phi,x)=2a\arccos(1-\Psi(\phi,x))$ and using $\partial_x\Psi(\phi,x) =x/(aR_\oplus\sin\phi)$ establishes \eqref{eq:explicit_derivative_pvoid_pcp}.
\end{IEEEproof}

\subsection{Proof of Corollary~\ref{cor:serving_angle_bcp}}
\label{app:cor2}
\begin{IEEEproof}
For a ring at central angle $\phi$, the quantity $-\partial_x p_{{\rm void},{\rm B}}(\phi,x)\,{\rm d}x$ gives the first-order probability that its nearest \ac{HAP} lies in $[x,x+{\rm d}x]$. Multiplying by the anchor intensity measure $2\pi\lambda_{\rm u}R_{\rm H}^2\sin\phi\,{\rm d}\phi$ and the probability $V_{\rm B}(x)$ that no other ring contains a closer \ac{HAP} gives the joint finite-serving density \ref{eq:joint_pdf_bcp}. Integrating over $\phi$ yields the distribution of finite serving distance, while normalization gives \eqref{eq:conditional_serving_angle_bcp}. Finally, differentiating $p_{{\rm void},{\rm B}}(\phi,x) =(1-L(\phi,x)/(2\pi a))^{n_{\rm H}}$ and substituting $\partial_x\Psi(\phi,x) =x/(aR_\oplus\sin\phi)$ establishes \eqref{eq:explicit_derivative_pvoid_bcp}.
\end{IEEEproof}

\section{Proof of Lemma~\ref{lem:energy_opt}}
\label{app:lemma5}
\begin{IEEEproof}
From \eqref{eq:cee_def}, $\eta_{{\rm CEE},k}(\gamma) = \frac{P_{{\rm c},k}(\tau;\gamma)} {P_{\rm SCF}(\gamma,V)}.$ At an interior maximizer with positive coverage probability, logarithmic differentiation yields
\begin{align}
    \frac{1}{P_{{\rm c},k}} \frac{\partial P_{{\rm c},k}} {\partial\gamma} = \frac{1}{P_{\rm SCF}} \frac{\partial P_{\rm SCF}} {\partial\gamma}.
    \label{eq:log_stationarity}
\end{align}
From Lemma~\ref{lem:power}, the propulsion power can be expressed as $P_{\rm SCF}(\gamma,V) = P_{\rm SHF}(V) \left( 1+\frac{\beta(V)}{\tan^2\gamma} \right),$ for fixed $V$, differentiating with respect to $\gamma$ gives 
\begin{align*}
    \frac{1}{P_{\rm SCF}} \frac{\partial P_{\rm SCF}} {\partial\gamma} = -\frac{ 2\beta(V)\sec^2\gamma }{ \tan^3\gamma+ \beta(V)\tan\gamma }.
\end{align*}
Substitution into \eqref{eq:log_stationarity} establishes the lemma.
\end{IEEEproof}
\end{document}